\documentclass[orivec,runningheads]{llncs}
\usepackage[utf8]{inputenc}

\usepackage{cite}
\usepackage{paralist}
\usepackage{mathtools}
\usepackage{xspace}
\usepackage{amssymb}
\usepackage{amsmath}
\usepackage{thmtools}
\usepackage{wrapfig}
\usepackage{hyperref}
\usepackage[capitalize]{cleveref}
\usepackage[dvipsnames]{xcolor}
\usepackage{tikz}
\usepackage[linesnumbered,noend,plainruled]{algorithm2e}
\usepackage{caption}
\usepackage{subcaption}
\usepackage{thmtools} 
\usepackage{thm-restate}
\usepackage{array}
\usepackage{xcolor}
\usepackage{colortbl}
\usepackage{booktabs}

\usepackage{newtxmath}
\usepackage{newtxtext}

\usetikzlibrary{arrows,automata}

\usepackage{marginnote}
\usepackage[colorinlistoftodos,size=small,textsize=small,backgroundcolor=white]{todonotes}

\newcounter{claimcounter}
\renewenvironment{claim}{\refstepcounter{claimcounter}{\medskip\noindent \underline{Claim \theclaimcounter:}}\itshape}{\smallskip}
\crefname{claimcounter}{Claim}{Claims}
\crefname{section}{Sec.}{Secs.}

\def\orcidID#1{\smash{\href{http://orcid.org/#1}{\protect\raisebox{-1.25pt}{\protect\includegraphics{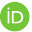}}}}}

\renewcommand{\epsilon}{\varepsilon}
\renewcommand{\phi}{\varphi}

\newcommand{\ol}[1]{\textcolor{blue}{\ifmmode \text{[OL: #1]}\else [OL: #1] \fi}}
\newcommand{\bs}[1]{\textcolor{violet}{\ifmmode \text{[BS: #1]}\else [BS: #1] \fi}}
\newcommand{\vh}[1]{\textcolor{green!60!black}{\ifmmode \text{[VH: #1]}\else [VH: #1] \fi}}

\renewcommand{\ol}[1]{}

\iftrue
	\usepackage{todonotes}
	\newcommand{\at}[1]{\todo[inline,color=teal!10,caption={AT}]{\textbf{AT:} #1}}
	\newcommand{\ly}[1]{\todo[inline,color=orange!10,caption={LY}]{\textbf{LY:} #1}}
\else
	\newcommand{\at}[1]{}
	\newcommand{\ly}[1]{}
\fi

\newcommand{\calM}[0]{\mathcal{M}}

\newcommand{\C}{\mathcal{C}}

\newcommand{\alphabet}{\Sigma}
\newcommand{\word}[0]{w}
\newcommand{\wordof}[1]{\word_{#1}}
\newcommand{\aut}[1][A]{\mathcal{#1}}
\newcommand{\but}[1][B]{\mathcal{#1}}
\newcommand{\autof}[1]{\aut{}[#1]}
\newcommand{\states}{Q}
\newcommand{\inits}{I}
\newcommand{\trans}{\delta}
\newcommand{\transacc}{\trans_{\acc}}
\newcommand{\transscc}{\trans_{\mathrm{SCC}}}
\newcommand{\colourset}[0]{\Gamma}
\newcommand{\colouring}[0]{\mathsf{p}}
\newcommand{\acccond}[0]{\mathsf{Acc}}
\newcommand{\cut}[0]{\aut[C]}
\newcommand{\autex}[0]{\aut_{\mathit{ex}}}
\newcommand{\butex}[0]{\but_{\mathit{ex}}}

\newcommand{\acc}[0]{F}

\newcommand{\transover}[1]{\overset{#1}{\rightarrow}}
\newcommand{\ltr}[1]{\transover{#1}}
\newcommand{\restrof}[2]{#1 \raisebox{-.5ex}{$|$}_{#2}}

\newcommand{\emersonlei}[0]{\mathbb{EL}}
\newcommand{\emersonleiof}[1]{\emersonlei(#1)}

\newcommand{\myinf}[0]{\mathrm{inf}}

\newcommand{\infoft}[1]{\myinf_{\trans}(#1)}

\newcommand{\infofcol}[1]{\myinf_{\colourset}(#1)}
\newcommand{\Inf}[0]{\mathsf{Inf}}
\newcommand{\Infof}[1]{\Inf(#1)}
\newcommand{\Fin}[0]{\mathsf{Fin}}
\newcommand{\Finof}[1]{\Fin(#1)}

\newcommand{\lang}[0]{\mathcal{L}}
\newcommand{\langof}[1]{\lang(#1)}

\newcommand{\bigO}[0]{\mathcal{O}}

\newcommand{\buchi}{B\"{u}chi\xspace}

\newcommand{\renum}[0]{\lambda}
\newcommand{\renumof}[2]{\renum(#1, #2)}

\newcommand{\typestyle}[1]{\mathtt{#1}}
\newcommand{\onetype}[0]{\typestyle{T}}
\newcommand{\tracktype}[0]{\typestyle{PT}}
\newcommand{\activetype}[0]{\typestyle{AT}}
\newcommand{\colourtype}[0]{\typestyle{Colours}}

\newcommand{\onetypeof}[1]{\onetype^{#1}}
\newcommand{\tracktypeof}[1]{\tracktype^{#1}}
\newcommand{\activetypeof}[1]{\activetype^{#1}}
\newcommand{\colourtypeof}[1]{\colourtype^{#1}}

\newcommand{\GetInit}[0]{\mathtt{Init}}
\newcommand{\GetSucc}[0]{\mathtt{Succ}}

\newcommand{\GetAcc}[0]{\mathtt{Acc}}
\newcommand{\GetSuccA}[0]{\mathtt{SuccAct}}
\newcommand{\GetSuccT}[0]{\mathtt{SuccPass}}

\newcommand{\Lift}[0]{\mathtt{Lift}}
\newcommand{\GetInitOf}[1]{\GetInit^{#1}}
\newcommand{\GetSuccOf}[1]{\GetSucc^{#1}}
\newcommand{\GetAccOf}[1]{\GetAcc^{#1}}
\newcommand{\GetSuccAOf}[1]{\GetSuccA^{#1}}
\newcommand{\GetSuccTOf}[1]{\GetSuccT^{#1}}

\newcommand{\LiftOf}[1]{\Lift^{#1}}
\newcommand{\iwc}[0]{\mathtt{MH}}
\newcommand{\dac}[0]{\mathtt{CSB}}
\newcommand{\idac}[0]{\mathtt{CoB}}
\newcommand{\iwcrr}[0]{\mathtt{MHRR}}
\newcommand{\dacrr}[0]{\mathtt{CSBRR}}
\newcommand{\nacrnk}[0]{\mathtt{RNK}}
\newcommand{\detalg}[0]{\mathtt{DET}}

\newcommand{\alg}[0]{\mathtt{Alg}}
\newcommand{\algrr}[0]{\mathtt{AlgRR}}
\newcommand{\algunion}[0]{\textsc{Union}}

\newcommand{\modelssurrp}{\models_{P}}
\newcommand{\modcompl}{\textsc{ModCompl}}
\newcommand{\modcomplrr}{\textsc{ModComplRR}}
\newcommand{\postponedcompl}{\textsc{PostpCompl}}
\newcommand{\reduce}{\mathtt{Red}}
\newcommand{\fsucc}{\mathtt{fSucc}}

\newcommand{\simul}{\preccurlyeq}

\newcommand{\extension}[0]{\vcenter{\hbox{\tikz[anchor=center,baseline]{%  ranking extension
  \fill (0mm,0.0mm) rectangle +(1.3mm,1.3mm);%
}}}}
\newcommand{\qext}[0]{Q^{\extension}}         % extended states
\newcommand{\qextof}[1]{\qext_{#1}}  % extended states
\newcommand{\rankleq}[0]{\mathrel{\leq}^{\bullet}}
\newcommand{\reach}{\mathit{reach}}
\newcommand{\reachof}[1]{\reach(#1)}
\newcommand{\maxrank}{\mathit{maxrank}}
\newcommand{\algmaxrank}{\textsc{MaxRank}}
\newcommand{\algdeelev}{\textsc{InitDet}}
\newcommand{\dom}{\mathit{dom}}
\newcommand{\domof}[1]{\dom(#1)}
\newcommand{\evenceil}[1]{\lfloor\!\!\lfloor #1 \rfloor\!\!\rfloor}
\newcommand{\iincrof}[1]{\texttt{++}^{#1}}
\newcommand{\deltaext}[2]{\delta^{#1}_{#2}}
\newcommand{\deltaextof}[4]{\deltaext{#1}{#2}(#3, #4)}
\newcommand{\partialto}[0]{\mathrel{\rightharpoonup}}
\newcommand{\rank}[0]{\mathit{rank}}        % rank
\newcommand{\rankof}[1]{\rank(#1)}          % rank of
\newcommand{\transconsist}[0]{%
  \begin{tikzpicture}[line width=0.6pt,transform shape,scale=0.7]
    \draw[-stealth] (0,0) -- (1.5em,0);
    \draw (0.6em,0ex) circle (0.7ex);
  \end{tikzpicture}%
}
\newcommand{\transconsistof}[2]{\mathrel{\transconsist^{#1}_{#2}}}
\newcommand{\wait}[0]{\mathit{wait}}         % waiting part
\newcommand{\tight}[0]{\mathit{tight}}       % tight part

\newcommand{\cola}{\textsc{COLA}\xspace}
\newcommand{\kofola}[1][]{\textsc{Kofola}$_{#1}$\xspace}
\newcommand{\kofolaids}{\kofola[S]}
\newcommand{\kofolaidsc}{\kofola[P]}
\newcommand{\ranker}{\textsc{Ranker}\xspace}
\newcommand{\seminator}{\textsc{Seminator}\xspace}
\newcommand{\spot}{\textsc{Spot}\xspace}
\newcommand{\vbs}[1][]{\textsc{VBS}$_{#1}$\xspace}
\newcommand{\vbsall}{\vbs[+]}
\newcommand{\vbsnokofola}{\vbs[-]}

\newcommand{\ultimizer}{\textsc{Ultimate Automizer}\xspace}
\newcommand{\pecan}{\textsc{Pecan}\xspace}

\newcommand{\benchexec}{\textsc{BenchExec}\xspace}
\newcommand{\automatabenchmarks}{\textsc{automata-benchmarks}\xspace}

\newcommand{\claimqed}[0]{\hfill $\blacksquare$}
\newenvironment{claimproof}[1]{\par\noindent\underline{Proof:}\space#1}{\claimqed\medskip}

\newcommand{\pr}[0]{\mathit{pr}}
\newcommand{\id}[0]{\mathit{id}}

\newcommand{\num}{\#}

\newcommand{\acctrans}{\delta_F}

\newcommand{\dagg}[0]{\mathcal{G}}
\newcommand{\dagof}[1]{\dagg_{#1}}
\newcommand{\dagw}[0]{\dagof{\word}}

\newcommand{\delayed}[0]{\textsc{Delayed}\xspace}

\newcommand{\rankrestr}[0]{\mathit{RankRestr}}

\usetikzlibrary{automata}
\usetikzlibrary{arrows.meta}
\usetikzlibrary{bending}
\usetikzlibrary{shapes.callouts}
\usetikzlibrary{quotes}
\usetikzlibrary{positioning}
\usetikzlibrary{calc}
\usetikzlibrary{matrix}

\definecolor{spotblue}{RGB}{31,120,180}
\definecolor{spotpink}{RGB}{255,77,160}
\definecolor{spotorange}{RGB}{255,127,0}
\definecolor{spotpurple}{RGB}{106,61,154}
\definecolor{spotgreen}{RGB}{51,160,44}
\definecolor{spotred}{RGB}{227,26,28}
\definecolor{spotyellowish}{RGB}{196,196,0}
\definecolor{spotgray}{RGB}{80,80,80}
\definecolor{spotlight blue}{RGB}{107,246,255}
\definecolor{spotlight pink}{RGB}{255,154,255}
\definecolor{spotlight orange}{RGB}{255,156,103}
\definecolor{spotlight purple}{RGB}{178,164,255}
\definecolor{spotlight green}{RGB}{167,237,121}
\definecolor{spotlight red}{RGB}{255,104,104}
\definecolor{spotlight yellowish}{RGB}{255,224,64}
\definecolor{spotlight gray}{RGB}{192,192,144}

\tikzset{
  >={Stealth[round,bend]},
}

\tikzstyle{automaton}=[
  semithick,shorten >=0pt,>={Stealth[round,bend]},
  node distance=1.5cm,
  initial text=,
  every initial by arrow/.style={every node/.style={inner sep=0pt}},
  every state/.style={
    align=center,
    fill=white,
    minimum size=7.5mm,
    inner sep=0pt,
    execute at begin node=\strut,
  },%
]
\tikzset{
  scc/.style={draw=gray,fill=black!10,rounded corners=2mm},
  lstate/.style={state},
  }

\tikzstyle{smallautomaton}=[
  automaton,
  node distance=7mm,
  every state/.style={minimum size=4mm,fill=white,inner sep=1.5pt}
]

\tikzstyle{mediumautomaton}=[
  automaton,
  node distance=1cm,
  every state/.style={minimum size=6mm,fill=white,inner sep=2pt}
]

\tikzstyle{cstate}=[state,capsule,text width=,inner xsep=-5pt]
\tikzstyle{dot}=[fill=black,circle,minimum size=4pt,inner sep=0]

\makeatletter
\tikzoption{initial angle}{\tikzaddafternodepathoption{\def\tikz@initial@angle{#1}}}
\makeatother

\tikzstyle{initial overlay}=[every initial by arrow/.append style={overlay}]

\tikzstyle{unreachable} = [densely dotted]

\tikzstyle{acclabel} = [
  fill=yellow!20,
  rounded corners=3pt,
  draw=black,
  inner ysep=0pt,
  inner xsep=3pt,
]
\tikzstyle{ltllabel} = [acclabel,fill=darkgreen!20]

\tikzstyle{namelabel} = [
  execute at begin node = {$($},%
  execute at end node = {$)$}
]

\tikzstyle{matrix of states} = [
matrix of nodes,
every node/.style={state},
column sep=.8cm,
row sep=.8cm,
execute at empty cell={
  \node[transparent,name=
    \tikzmatrixname-\the\pgfmatrixcurrentrow-\the\pgfmatrixcurrentcolumn]{};]},
]

\tikzstyle{accset}=[
  circle,inner sep=.5pt,
  draw=none, %draw=white,
  solid,
  collacc0, text=white,
  thin,
  anchor=center,
  minimum size=3.3mm,
  text width={}, % in case tacc{} is called inside a node with text width set
  font=\bfseries\sffamily\footnotesize
]
\tikzstyle{accsquare}=[accset,rectangle,inner sep=1.9pt,rounded corners=0pt]

\tikzset{
  collacc0/.style={fill=spotblue},
  collacc1/.style={fill=spotpink},
  collacc2/.style={fill=spotorange},
  collacc3/.style={fill=spotpurple},
  collacc4/.style={fill=spotgreen},
  collacc5/.style={fill=spotred},
  collacc6/.style={fill=spotyellowish,draw=black,text=black},
  collacc7/.style={fill=spotgray},
  collacc8/.style={fill=spotlight blue,draw=black,text=black},
  collacc9/.style={fill=spotlight pink},
  collacc10/.style={fill=spotlight orange},
  collacc11/.style={fill=spotlight purple},
  collacc12/.style={fill=spotlight green},
  collacc13/.style={fill=spotlight red},
  collacc14/.style={fill=spotlight yellowish},
  collacc15/.style={fill=spotlight gray},
}

\pgfkeyssetvalue{/sacc/where}{center}
\tikzset{
  sacc where/.code={
    \pgfkeyssetvalue{/sacc/where}{#1}
  }
}

\tikzset{
  l/.pic={\node[outer sep=2pt] {#1};},%
  acc/.pic={\node[accset,collacc#1]{#1};},%
  accsq/.pic={\node[accsquare,collacc#1]{#1};},%
  eacc/.pic={\node[accset,collacc#1]{\emptyacc};},%
  eaccsq/.pic={\node[accsquare,collacc#1]{\emptyacc};},%
  sacc/.style = {
    append after command=
      {pic at (\tikzlastnode.\pgfkeysvalueof{/sacc/where}) {acc=#1}}
  },
  saccsq/.style = {
    append after command=
      {pic at (\tikzlastnode.\pgfkeysvalueof{/sacc/where}) {accsq=#1}}
  },
  esacc/.style = {
    append after command=
      {pic at (\tikzlastnode.\pgfkeysvalueof{/sacc/where}) {eacc=#1}}
  },
  esaccsq/.style = {
    append after command=
      {pic at (\tikzlastnode.\pgfkeysvalueof{/sacc/where}) {eaccsq=#1}}
  },
  pics/cacc/.style 2 args={%
    code={\node[accset,collacc#1]{#2};}%
  },%
  pics/caccsq/.style 2 args={%
      code={\node[accsquare,collacc#1]{#2};}%
    },%
  csacc/.style 2 args = {%
      append after command=%
        {pic at (\tikzlastnode.\pgfkeysvalueof{/sacc/where}) {cacc={#1}{#2}}}%
  },%
  csaccsq/.style 2 args = {%
      append after command=%
        {pic at (\tikzlastnode.\pgfkeysvalueof{/sacc/where}) {caccsq={#1}{#2}}}%
  },%
}

\def\markbaseline{-.33em}
\def\tacc#1{\tikz[baseline=\markbaseline]\pic{acc=#1};\xspace}
\def\tcacc#1#2{\tikz[baseline=\markbaseline]\pic{cacc={#1}{#2}};\xspace}

\def\emptyacc{\phantom{0}}

\tikzstyle{removed} = [opacity=0, overlay]
\tikzstyle{SCC} = [
  rounded corners,
  draw=black!50, very thin,
  fill=black!7
]
\tikzstyle{trivial} = [dashed]
\tikzstyle{sccname} = [red,anchor=south east,outer xsep=13pt]

\title{
  Modular Mix-and-Match\\ Complementation of \buchi Automata\\
  (Technical Report)
}
\author{
Vojt\v{e}ch Havlena\inst{1}\orcidID{0000-0003-4375-7954} \and
Ond\v{r}ej Leng\'{a}l\inst{1}\orcidID{0000-0002-3038-5875} \and
Yong Li\inst{2,3}\orcidID{0000-0002-7301-9234} \and\\
Barbora \v{S}mahl\'{i}kov\'{a}\inst{1}\orcidID{0000-0002-1184-4669} \and
Andrea Turrini\inst{3,4}\orcidID{0000-0003-4343-9323}
}
\authorrunning{
  V. Havlena,
  O. Leng\'{a}l,
  Y. Li,
  B. \v{S}mahl\'{i}kov\'{a}, and
  A. Turrini
}

\institute{
Faculty of Information Technology,
Brno University of Technology,
Brno,
Czech Republic
\and
Department of Computer Science, University of Liverpool, UK
\and
State Key Laboratory of Computer Science,
Institute of Software,\\
Chinese Academy of Sciences,
Beijing,
P.\@ R.\@ China
\and
Institute of Intelligent Software, Guangzhou,
Guangzhou,
P.\@ R.\@ China
}

\begin{document}

\maketitle

\vspace{-8mm}
\begin{abstract}
  Complementation of nondeterministic \buchi automata (BAs) is an important
  problem in automata theory with numerous applications in formal verification,
  such as termination analysis of programs, model checking, or in decision
  procedures of some logics.
  We build on ideas from a~recent work on BA determinization by Li \emph{et al.} and
  propose a~new modular algorithm for BA complementation.
  Our algorithm allows to combine several BA complementation procedures
  together, with one procedure for a~subset of the BA's strongly connected components (SCCs).
  In this way, one can exploit the structure of particular SCCs (such as when
  they are inherently weak or deterministic) and use more efficient specialized
  algorithms, regardless of the structure of the whole~BA.
  We give a~general framework into which partial complementation procedures can
  be plugged in, and its instantiation with several algorithms.
  The framework can, in general, produce a complement with an Emerson-Lei
  acceptance condition, which can often be more compact.
  Using the algorithm, we were able to establish an exponentially better new
  upper bound of~$\bigO(4^n)$ for complementation of the recently introduced class of elevator
  automata.
  We implemented the algorithm in a~prototype and performed a~comprehensive set
  of experiments on a~large set of benchmarks, showing
  that our framework complements well the state of the art and that it can serve as
  a~basis for future efficient \mbox{BA complementation and inclusion checking algorithms.}
  % that our framework complements well existing approaches and significantly
  % improves the state of the art.

	% Recently, a determinization-based complementation algorithm proposed by Li et al. reveals the significance of considering each type of SCCs separately while taking advantage of their structural information.
	% We take a step further by proposing a more general and modular framework that complements each type of SCCs separately.
	% \ly{Not quite mix and match description below, need to be fixed}
	% By using our framework, we have improved exponentially the complementation complexity of a class called \emph{elevator} \buchi automata to $\bigO(4^n)$, same as the well-known class semideterministic \buchi automata.
	% Moreover, by considering the Emerson-Lei condition for the complement automaton, the resultant state space can be more compact than classical approaches.
	% As our last contribution, we conduct a comprehensive experiments on a large set of benchmarks in literature.
	% We show that our framework complements well the existing approaches and improves significantly the state of the arts.
	
\end{abstract}

%%%%%%%%%%%%%%%%%%%%%%%%%%%%%%%%%%%%%%%%%%%%%%%%%%%%%%%%%%%%%%%%%%%%%%%%%%%%%%%%
\vspace{-9.0mm}
\section{Introduction}
\label{sec:intro}
\vspace{-2.0mm}
%%%%%%%%%%%%%%%%%%%%%%%%%%%%%%%%%%%%%%%%%%%%%%%%%%%%%%%%%%%%%%%%%%%%%%%%%%%%%%%%

Nondeterministic \buchi automata (BAs)~\cite{Buchi90} are an elegant and
conceptually simple framework to model infinite behaviors of systems and the
properties they are expected to satisfy.
BAs are widely used in many important verification tasks, such as termination
analysis of programs~\cite{HeizmannHP14}, model checking~\cite{VardiW86}, or as
the underlying formal model of
decision procedures for some logics (such as S1S~\cite{Buchi90} or a~fragment of
the first-order logic over Sturmian words~\cite{HieronymiMOS0S22}).
Many of these applications require to perform \emph{complementation} of BAs:
For instance, in termination analysis of programs within
\ultimizer~\cite{HeizmannHP14}, complementation is used to keep track of the
set
of paths whose termination still needs to be proved.
% For instance, in termination analysis of programs within Ultimate
% Automizer~\cite{HeizmannHP14}, a~program is shown to be terminating on all
% inputs by
% starting with a~BA representing the set of all paths in the program, and
% removing from the set those paths whose termination has been proved, until the
% language of the BA is empty.
% Removing the paths is performed by computing the intersection of the BA with
% the complement of another BA that represents the terminating paths.
On the other hand, in model checking\footnote{Here, we consider model checking
w.r.t.\ a~specification given in some more expressive logic, such as
S1S~\cite{Buchi90}, QPTL~\cite{SistlaVW87}, or HyperLTL~\cite{ClarksonFKMRS14},
rather than LTL~\cite{Pnueli77}, where negation is simple.}
and decision procedures of logics, complement is usually used to implement
negation and quantifier alternation.
Complementation is often the most difficult automata operation performed here;
its worst-case state complexity is $\bigO((0.76n)^n)$~\cite{Schewe09,AllredU18}
(which is tight~\cite{Yan08}).

In these applications, efficiency of the complementation often determines the
overall efficiency (or even feasibility) of the top-level application.
For instance, the success of \ultimizer in the Termination category of
the International Competition on Software Verification (SV-COMP)~\cite{svcomp}
is to a~large degree due to an efficient BA complementation
algorithm~\cite{BlahoudekHSST16,ChenHLLTTZ18} tailored
for BAs with a~special structure that it often encounters (as of the time of
writing, it has won 6~gold medals in the years 2017--2022 and two silver
medals in 2015 and 2016).
The special structure in this case are the so-called \emph{semi-deterministic
BAs} (SDBAs), BAs consisting of two parts:
\begin{inparaenum}[(i)]
  \item  an initial part without accepting states/transitions and
  \item  a deterministic part containing accepting states/transitions that
    cannot transition into the first part.
\end{inparaenum}

Complementation of SDBAs using one from the family of the so-called NCSB
algorithms~\cite{BlahoudekHSST16,BlahoudekDS20,ChenHLLTTZ18,HavlenaLS22b} has
the worst-case complexity~$\bigO(4^n)$ (and usually also works much better in
practice than general BA complementation procedures).
Similarly, there are efficient complementation procedures for other subclasses of
BAs, e.g.,
\begin{inparaenum}[(i)]
  \item  \emph{deterministic BAs} (DBAs) can be complemented into BAs with~$2n$
    states~\cite{Kurshan87} (or into co-\buchi automata with $n+1$ states) or
  \item  \emph{inherently weak BAs} (BAs where in each \emph{strongly connected
    component} (SCC), either all cycles are accepting or all cycles are
    rejecting) can be complemented into DBAs with~$\bigO(3^n)$ states using the
    Miyano-Hayashi algorithm~\cite{MiyanoH84}.
\end{inparaenum}

For a~long time, there has been no efficient algorithm for complementation of
BAs that are highly structured but do not fall into one of the categories above,
e.g., BAs containing inherently weak, deterministic, and
some nondeterministic SCCs.
For such BAs, one needed to use a~general complementation
algorithm with the~$\bigO((0.76n)^n)$ (or worse) complexity.
To the best of our knowledge, only recently has there appeared works that
exploit the structure of BAs to obtain a~more efficient complementation
algorithm:
\begin{inparaenum}[(i)]
  \item  The work of Havlena \emph{et al.}~\cite{HavlenaLS22a}, who introduce
    the class of \emph{elevator automata} (BAs with an arbitrary mixture of
    inherently weak and deterministic SCCs) and give a~$\bigO(16^n)$ algorithm
    for them.
  \item  The work of Li \emph{et al.}~\cite{LiTFVZ22}, who propose a~BA
    determinization procedure (into a~deterministic Emerson-Lei automaton) that
    is based on decomposing the input BA into SCCs and using a~different
    determinization procedure for different types of SCCs (inherently weak,
    deterministic, general) in a~synchronous construction.
\end{inparaenum}

In this paper, we propose a~new BA complementation algorithm inspired
by~\cite{LiTFVZ22}, where we exploit the fact that complementation is, in
a~sense, more relaxed than determinization.
In particular, we present a~\emph{framework} where one can plug-in different
partial complementation procedures fine-tuned for SCCs with a~specific
structure.
The procedures work only with the given SCCs, to some degree
\emph{independently} (thus reducing the potential state space explosion) from
the rest of the BA.
Our top-level algorithm then orchestrates runs of the different procedures in
a~\emph{synchronous} manner (or completely independently in the so-called
\emph{postponed} strategy), obtaining a resulting automaton with potentially
a~more general acceptance condition (in general an Emerson-Lei condition),
which can help keeping the result small.
If the procedures satisfy given correctness requirements, our framework
guarantees that its instantiation will also be correct.
We also propose its optimizations by, e.g., using round-robin to decrease the
amount of nondeterminism, using shared breakpoint to reduce the size and
the number of colours for certain class of partial algorithms, and
generalize simulation-based pruning of macrostates.

We provide a~detailed description of partial complementation procedures for 
inherently weak, deterministic, and initial deterministic SCCs, which we use to
obtain a~\emph{new} exponentially better upper bound of~$\bigO(4^n)$ for the class of
elevator automata (i.e., the same upper bound as for its strict subclass of SDBAs).
Furthermore, we also provide two partial procedures for general SCCs based on
determinization (from~\cite{LiTFVZ22}) and the rank-based construction.
Using a~prototype implementation, we then show our algorithm complements well
existing approaches and significantly improves the state of the art.

\vspace{-3.0mm}
\section{Preliminaries}
\label{sec:preliminaries}
\vspace{-2.0mm}
%%%%%%%%%%%%%%%%%%%%%%%%%%%%%%%%%%%%%%%%%%%%%%%%%%%%%%%%%%%%%%%%%%%%%%%%%%%%%%%%

We fix a~finite non-empty alphabet~$\alphabet$ and the first infinite
ordinal~$\omega$.
An (infinite) word~$\word$ is a~function $\word\colon \omega \to \alphabet$ where
the $i$-th symbol is denoted as~$\wordof i$.
Sometimes, we represent~$\word$ as an~infinite sequence $\word = \wordof 0
\wordof 1 \dots$
We denote the set of all infinite words over~$\alphabet$ as $\alphabet^\omega$;
an \emph{$\omega$-language} is a~subset of~$\alphabet^\omega$.

%------------------------------------------------------------------------------
\vspace{-2mm}
\paragraph{Emerson-Lei Acceptance Conditions.}
Given a~set $\colourset = \{0, \ldots, k -1\}$ of~$k$
\emph{colours} (often depicted as \tacc{0}, \tacc{1}, etc.), we define the
set of \emph{Emerson-Lei acceptance conditions} $\emersonleiof \colourset$ as
the set of formulae constructed according to the following grammar:
\vspace{-2mm}
\begin{equation}
  % \alpha ::= \top \mid \bot \mid \Infof c \mid \Finof c \mid \alpha \land
  \alpha ::= \Infof c \mid \Finof c \mid (\alpha \land \alpha) \mid (\alpha \lor \alpha)\\[-2mm]
\end{equation}
for $c \in \colourset$.
The \emph{satisfaction} relation $\models$ for a~set of colours~$M \subseteq
\colourset$ and condition~$\alpha$ is defined inductively as follows (for $c
\in \colourset$):
\vspace{-2mm}
\begin{align*}
  M \models \Finof c & \text{~ iff ~} c \notin M, &
  M \models \alpha_1 \lor \alpha_2 & \text{~ iff ~} M \models \alpha_1
  \text{ or } M \models \alpha_2,
  \\
  M \models \Infof c & \text{~ iff ~} c \in M,  &
  M \models \alpha_1 \land \alpha_2 & \text{~ iff ~} M \models \alpha_1
  \text{ and } M \models \alpha_2.
\end{align*}

%------------------------------------------------------------------------------
\vspace{-6mm}
\paragraph{Emerson-Lei Automata.}
A~(nondeterministic transition-based\footnote{%
We only consider transition-based acceptance in order to avoid cluttering the
paper by always dealing with accepting states \emph{and} accepting transitions.
Extending our approach to state/transition-based (or just state-based) automata
is straightforward.
})
\emph{Emerson-Lei automaton} (TELA)
over~$\alphabet$ is a~tuple $\aut = (\states, \trans, \inits, \colourset,
\colouring, \acccond)$,
where $\states$ is a~finite set of \emph{states},
$\trans \subseteq \states \times \alphabet \times \states$ is a~set of
\emph{transitions}\footnote{%
Note that some authors use a~more general definition
of TELAs with $\trans \subseteq \states \times \alphabet \times 2^{\colourset}
\times \states$; since we only use them on the output, we suffice with the
simpler definition.},
$\inits \subseteq \states$ is the set of \emph{initial} states,
$\colourset$ is the set of \emph{colours},
$\colouring\colon \trans \to 2^{\colourset}$ is a~\emph{colouring
function} of transitions, and
$\acccond \in \emersonleiof \colourset$.
We use $p \ltr a q$ to denote that $(p,a,q) \in \trans$ and sometimes also
treat~$\trans$ as a~function with the signature $\trans\colon \states \times
\alphabet \to 2^{\states}$.
Moreover, we extend~$\trans$ to sets of states $P \subseteq \states$ as
$\trans(P, a) = \bigcup_{p \in P} \trans(p,a)$.
We use $\autof q$ for $q \in \states$ to denote the automaton $\autof q =
(\states, \trans, \{q\}, \colourset, \colouring, \acccond)$, i.e., the TELA
obtained from~$\aut$ by setting~$q$ as the only initial state.
$\aut$ is called \emph{deterministic} if $|\inits|\leq 1$ and $|\trans(q,a)|\leq 1$
for each $q\in \states$ and $a \in \alphabet$.
If $\colourset = \{\tacc{0}\}$ and $\acccond = \Infof {\tacc{0}}$, we
call~$\aut$ a~\emph{\buchi automaton} (BA) and denote it as $\aut = (\states, \trans,
\inits, \acc)$ where~$\acc$ is the set of all transitions coloured
by~$\tacc{0}$, i.e., $\acc = \colouring^{-1}(\{\tacc{0}\}$).
For a~BA, we use $\transacc(p, a) = \{q \in \trans(p, a) \mid \colouring(p \ltr
a q) = \{\tacc{0}\}\}$ (and extend the notation to sets of states as
for~$\trans$).
A~BA~$\aut = (\states, \trans, \inits, \acc)$ is called
\emph{semi-deterministic} (SDBA) if for every accepting transition $(p \ltr{a} q) \in
F$, the reachable fragment \mbox{of~$\autof{q}$ is deterministic.}

A~\emph{run}
of~$\aut$ from~$q \in \states$ on an input word~$\word$ is an infinite sequence $\rho\colon
\omega \to \states$ that starts in~$q$ and respects~$\trans$, i.e., $\rho_0 = q$ and
$\forall i \geq 0\colon \rho_i \ltr{\wordof i}\rho_{i+1} \in \trans$.
Let $\infoft \rho \subseteq \delta$ denote the set of
transitions occurring in~$\rho$ infinitely often and $\infofcol \rho =
\bigcup\{\colouring(x) \mid x \in \infoft \rho\}$ be the set of infinitely
often occurring colours.
A~run~$\rho$ is \emph{accepting} in~$\aut$ iff $\infofcol \rho \models \acccond$ and 
the \emph{language} of~$\aut$, denoted as $\langof{\aut}$, is defined as the set of words
$w \in \alphabet^\omega$ for which there exists an accepting run in~$\aut$
starting with some state in~$\inits$.

Consider a~BA $\aut = (\states, \trans, \inits, \acc)$.
For a set of states $S\subseteq Q$ we use $\aut_S$ to denote the copy of~$\aut$
where accepting transitions only occur between states from~$S$, i.e., the BA
$\aut_S = (Q, \trans, \inits, \acc\cap \restrof{\trans}{S})$ where
$\restrof{\trans}{S} = \{p \ltr a q \in \trans \mid p,q \in S\}$.
We say that a~non-empty set of states~$C \subseteq \states$ is a~\emph{strongly
connected component} (SCC) if every pair of states of~$C$ can reach each other
and~$C$ is a~maximal such set.
An SCC of~$\aut$  is \emph{trivial} if it consists of a~single state that does
not contain a~self-loop and \emph{non-trivial} otherwise.
% By default, when we talk about SCCs, we mean non-trivial SCCs.
An~SCC~$C$ is
\emph{accepting} if it contains at least one accepting transition and
\emph{inherently weak} iff either
\begin{inparaenum}[(i)]
  \item  every cycle in~$C$ contains a~transition from~$\acc$ or
  \item  no cycle in~$C$ contains any transitions from~$\acc$.
\end{inparaenum}
%
% An inherently weak SCC with at least one transition in~$\acc$ is
% called \emph{accepting}, otherwise it is called \emph{rejecting}.
An SCC~$C$ is \emph{deterministic} iff the BA $(C, \restrof{\trans}{C},
\{q\}, \emptyset)$ for any~$q \in C$ is deterministic.
We denote inherently weak components as IWCs, accepting deterministic components that
are not inherently weak as DACs (deterministic accepting), and the remaining accepting
components as NACs (nondeterministic accepting).
% Note that all trivial SCCs are IWCs.
A~BA~$\aut$ is called an \emph{elevator automaton} if it contains no NAC.

We assume that~$\aut$ contains no accepting transition outside its SCCs (no run
can cycle over such transitions).
We use $\transscc$ to denote the restriction of~$\trans$ to transitions that do
not leave their SCCs, formally, $\transscc = \{p \ltr a q \in \delta \mid p
\text{ and } q \text{ are in the same SCC}\}$.
A~\emph{partition block} $P \subseteq \states$ of~$\aut$ is a~nonempty union of its
accepting SCCs, and a~\emph{partitioning} of~$\aut$ is a~sequence $P_1, \ldots,
P_n$ of pairwise disjoint partition blocks of~$\aut$ that contains all accepting SCCs
of~$\aut$.

The complement (automaton) of a BA $\aut$ is a TELA %$\mathcal{C}$
that accepts the complement language $\alphabet^\omega\setminus \lang(\aut)$ of $\lang(\aut)$.
In the paper, we call a state and a run of a complement automaton a
\emph{macrostate} and a \emph{macrorun}, respectively.

% A~word~$\word$ is \emph{accepted by~$\aut$ from a~state~$q \in \states$} if $\aut$ has an
% accepting run~$\rho$ on~$\word$ from~$q$, i.e., $\rho_0 = q$.
% The set
% $\langautof{\aut} q = \{\word \in \Sigma^\omega \mid \aut \text{ accepts } \word
% \text{ from } q\}$ is called the \emph{language} of~$q$ (in~$\aut$). Given a~set
% of states~$R \subseteq \states$, we define the language of~$R$ as $\langautof \aut R =
% \bigcup_{q \in R} \langautof \aut q$ and the language of~$\aut$ as~$\langof \aut =
% \langautof \aut \inits$.

%%%%%%%%%%%%%%%%%%%%%%%%%%%%%%%%%%%%%%%%%%%%%%%%%%%%%%%%%%%%%%%%%%%%%%%%%%%%%%%%
\vspace{-3.0mm}
\section{A~Modular Complementation Algorithm}\label{sec:modular}
\vspace{-2.0mm}
%%%%%%%%%%%%%%%%%%%%%%%%%%%%%%%%%%%%%%%%%%%%%%%%%%%%%%%%%%%%%%%%%%%%%%%%%%%%%%%%

In a nutshell, the main idea of our BA complementation algorithm is that we first
decompose a~BA~$\aut$ into several partition blocks according to their properties, and
then perform complementation for each of the partition blocks (potentially using
a~different algorithm) independently, using either a~\emph{synchronous}
construction, synchronizing the complementation algorithms for all partition blocks in
each step, or a~\emph{postponed} construction, which proceeds by complementing
the partition blocks independently and combines the partial results later using
automata product construction.
The decomposition of~$\aut$ into partition blocks can either be trivial---i.e., with
one block for each accepting SCC---, or more elaborate, e.g.,
a~partitioning where one partition block contains all accepting IWCs, another
contains all DACs, and each NAC is given its own partition block.

% Then, we build a~TELA that tracks all runs of~$\aut$ and runs the
% complementation algorithms \emph{independently} for each of the partitions.
% By doing this, we can effectively decrease the blow-up of the complementations
% in two main ways:
% %
% \begin{enumerate}
%   \item  
% \end{enumerate}
%
%
%
%
% The partitioning is driven mainly by the types of the SCCs, since the type
% determines what might be the most efficient complementation algorithm for the
% partition (e.g., complementing an inherently weak BA with~$n$ states can be
% done by generating at most $\frac 2 3 3^{n-1}$ states using the Miyano-Hayashi
% construction~\cite{MiyanoH84}).
% use the most efficient algorithm available
% mainly according to the types of the SCCs (such as whether
% they
% are IWCs, DACs, or NACs)

%!!!!!!!!!!!!!!!!!!!!!!!!!!!!!!!!
\enlargethispage{2mm}
%!!!!!!!!!!!!!!!!!!!!!!!!!!!!!!!!

In this way, one can avoid running a~general complementation algorithm for
unrestricted BAs with the state complexity upper bound $\bigO((0.76n)^n)$ and,
instead, apply the most suitable complementation procedure for each of the
partition blocks.
This comes with three main advantages:
\vspace{-2mm}
\begin{enumerate}
  \item  The complementation algorithm for each partition block can be selected
    differently in order to exploit the properties of the block.
    For instance, for partition blocks with IWCs, one can use complementation based
    on the breakpoint (the so-called Miyano-Hayashi)
    construction~\cite{MiyanoH84} with $\bigO(3^n)$ macrostates (cf.\
    \cref{sec:iwc-complement}),
    %with the upper bound of~$\frac 2 3 3^{n-1}$ generated states,
    while for partition blocks with only DACs, one can use an algorithm
    with the state complexity~$\bigO(4^n)$ based on an adaptation of the NCSB
    construction~\cite{BlahoudekHSST16,BlahoudekDS20,ChenHLLTTZ18,HavlenaLS22b}
    for SDBAs (cf. \cref{sec:dac-complement}).
    For NACs, one can choose between, e.g.,
    rank-~\cite{KupfermanV01,FriedgutKV06,Schewe09,ChenHL19,HavlenaL21,HavlenaLS22a}
    or determinization-based \cite{Safra88,Piterman07,Redziejowski12}
    algorithms, depending on the \mbox{properties of the NACs (cf.\ \cref{sec:general}).}

  \item  The different complementation algorithms can focus only on the
    respective blocks and do not need to consider other parts of the BA.
    This is advantageous, e.g., for rank-based algorithms, which can use this
    restriction to obtain tighter bounds on the considered ranks (even
    tighter than using the refinement in~\cite{HavlenaLS22a}).

  \item  The obtained automaton can be more compact due to the use of a~more
    general acceptance condition than \buchi \cite{SafraV89}---in general, it can be
    a~conjunction of any $\emersonlei$ conditions (one condition for each
    partition block), depending on the output of the complementation procedures; this
    can allow a~more compact encoding of the produced automaton allowed by using
    a~mixture of conditions.
    E.g., a~deterministic BA can be complemented with constant extra
    generated states when using a co-\buchi condition rather than a~linear number of
    generated \mbox{states for a~\buchi condition (see \cref{sec:det-init}).}
\end{enumerate}
Those partial complementation algorithms then need to be orchestrated by
a~top-level algorithm to produce the complement of~$\aut$.

One might regard our algorithm as an optimization of an approach that would for
each partition block~$P$ obtain a~BA~$\aut_P$, complement~$\aut_P$ using the selected
algorithm, and perform the intersection of all~$\aut_P$'s obtained in this way
(which would, however, not be able to obtain the upper bound for elevator
automata that we give in \cref{sec:elevator-upper}).
Indeed, we also implemented the mentioned procedure (called the
\emph{postponed} approach, described in \cref{sec:postponed}) and compared it to
our main procedure (called the \emph{synchronous} approach) described below.

%%%%%%%%%%%%%%%%%%%%%%%%%%%%%%%%%%%%%%%%%%%%%%%%%%%%%%%%%%%%%%%%%%%%%
\newcommand{\figExampleAut}[0]{
\begin{figure}[t]
  \centering
  %\resizebox{\linewidth}{!}{
    \begin{tikzpicture}
      \node[scale=0.7] (aex) {% \begin{tikzpicture}[>=stealth',shorten >=0pt,node distance=1.4cm,
%                     scale=0.8,transform shape,initial text={}]
\begin{tikzpicture}[automaton]
  \tikzstyle{every state}=[inner sep=3pt,minimum size=5pt]
  \tikzstyle{empty}=[]
  \tikzstyle{initstate}=[fill=yellow!30]
  
  \path[use as bounding box] (-0.75,-2.75) rectangle (3.5,2.75);

  \node[state,initial,initstate] (p) at (0,0) {$p$};
  \node[state,above right of=p,xshift=10mm] (q) {$q$};
  \node[state,below right of=p,xshift=10mm] (r) {$r$};
  \node[state] (s) at ($(q)!0.5!(r)$) {$s$};

  \node[font=\Large] at (0,2) {$\autex$};
  
  \draw[dotted] ($(p.north west) + (-0.15,0.9)$) rectangle ($(p.south east) + (0.15,-0.15)$);
  \draw[dotted] ($(q.north west) + (-0.15,0.75)$) rectangle ($(q.south east) + (0.75,-0.15)$);
  \node[anchor=south] at ($(q.north) + (0.30,0.75)$) {$P_{0}$};
  \draw[dotted] ($(s.north west) + (-0.35,0.15)$) rectangle ($(r.south east) + (0.75,-0.15)$);
  \node[anchor=north] at ($(r.south) + (0.20,-0.15)$) {$P_{1}$};

  \path[->]
    (p) edge[loop above] node[auto] {$a,b$} (p)
    (p) edge node[above] {$a,b$} (q)
    % (q) edge[loop above] pic {l=$a$} pic[anchor=center] {acc=0} (q)
    (q) edge[loop above] node[auto] {$a$} node[anchor=center] {$\bullet$} (q)
    (q) edge[loop right] node {$b$} (q)
    (q) edge node[auto] {$a$} (s)
    (p) edge node[below] {$a$} (r)
    (r) edge[loop right] node[auto] {$b$} node[anchor=center] {$\bullet$} (r)
    (r) edge[bend left] node[left] {$b$} (s)
    (s) edge[bend left] node[auto] {$a$} node[anchor=center] {$\bullet$} (r)
    ;
\end{tikzpicture}};
      \node[anchor=west,scale=0.7] (comp) at (aex.east) {% \begin{tikzpicture}[>=stealth',shorten >=0pt,node distance=1.4cm,
%                     scale=0.8,transform shape,initial text={}]
\begin{tikzpicture}[automaton]
  \tikzstyle{every state}=[inner sep=3pt,minimum size=5pt]
  \tikzstyle{empty}=[]
  \tikzstyle{initstate}=[fill=yellow!30]
  % \tikzstyle{every state}=[rectangle,rounded corners,inner sep=3pt,minimum size=5pt]
  \tikzstyle{uberstate}=[
    rounded corners,draw,anchor=base,
    rectangle split,rectangle split horizontal,rectangle split parts=3,
    rectangle split part align=base,
    rectangle split part fill={black!20, blue!30, green!30}]
  \newcommand{\ustate}[6]{$#1$\nodepart{two}$#2, #3, #4$\nodepart{three}$#5, #6$}

  \path[use as bounding box] (8.5,0.4) rectangle (-1.5,-4.7);
  
%  \node [my shape=5, rectangle split horizontal] at (2,2)
%     {1\nodepart{two}2\nodepart{three}3\nodepart{four}4\nodepart{five}5};

  \node[uberstate,initial] (p) at (0,0) {\ustate p \emptyset \emptyset \emptyset \emptyset \emptyset};
  \node[uberstate,below of=p,xshift=0mm] (pq) {\ustate{p+q} q \emptyset q \emptyset \emptyset};
  \node[uberstate,below of=pq,xshift=0mm] (pqsafe) {\ustate{p+q} \emptyset q \emptyset \emptyset \emptyset};

  \node[uberstate,right of=p,xshift=30mm] (pqr) {\ustate{p+q+r} q \emptyset q r r};
  \node[uberstate,below of=pqr,xshift=0mm] (pqrs) {\ustate{p+q+r+s} q \emptyset q {r+s} {r+s}};
  \node[uberstate,below of=pqrs,xshift=0mm] (pqrsbreak) {\ustate{p+q+r+s} q \emptyset q {r+s} {r}};
  \node[uberstate,below of=pqrsbreak,xshift=0mm] (pqrssafe) {\ustate{p+q+r+s} \emptyset q \emptyset {r+s} {r+s}};

  \path[->]
    (p) edge pic[pos=0.3] {acc=0} pic[pos=0.6] {acc=1} pic[auto] {l=$b$} (pq)
    (pq) edge pic[pos=0.3] {acc=0} pic[pos=0.6] {acc=1} pic[auto] {l=$b$} (pqsafe)
    (pq.182) edge[out=210,in=150,loop,distance=9mm] pic {acc=1} pic[auto] {l=$b$} (pq.178)
    (pqsafe.182) edge[out=210,in=150,loop,distance=9mm] pic[pos=0.25] {acc=0} pic[pos=0.75] {acc=1} pic[auto] {l=$b$} (pqsafe.178)
    (p) edge pic[pos=0.41] {acc=0} pic[pos=0.59] {acc=1} pic[auto] {l=$a$} (pqr)
    (pqr) edge[bend left] pic {acc=1} pic[auto] {l=$a$} (pqrs)
    (pqr) edge[bend right] node[left] {$b$} (pqrs)
    (pqrs.12) edge[out=120,in=60,loop,distance=6mm] node[auto] {$b$} (pqrs.11)
    (pq) edge pic {acc=1} pic[auto] {l=$a$} (pqrs)
    (pqrs) edge[bend left] node[auto] {$a$} (pqrsbreak)
    (pqrsbreak) edge[bend left] node[auto] {$b$} (pqrs)
    (pqrsbreak.160) edge[bend left] pic {acc=1} pic[auto] {l=$a$} (pqrs.200)
    (pqrsbreak) edge pic {acc=0} pic[auto] {l=$b$} (pqrssafe)
    (pqrssafe.182) edge[out=210,in=150,loop,distance=9mm] pic {acc=0} pic[auto] {l=$b$} (pqrssafe.178)
    (pqrs.187) edge[bend right] pic {acc=0} pic[left] {l=$b$} (pqrssafe.173)
    (pqr.355) edge[bend left=40] pic {acc=0} pic[auto] {l=$b$} (pqrssafe.2)
    ;
  %   (p) edge[loop above] node {$a,b$} (p)
  %   (p) edge node[above] {$a,b$} (pq)
  %   (q) edge[loop above] pic {l=$a$} pic[anchor=center] {acc=0} (pq)
  %   (pq) edge[loop below] node {$b$} (pq)
  %   (q) edge node[auto] {$a$} (s)
  %   (p) edge node[below] {$a$} (r)
  %   (r) edge[loop below] pic{l=$b$} pic[anchor=center] {acc=0} (r)
  %   (r) edge[bend left] node[above] {$b$} (s)
  %   (s) edge[bend left] pic[below]{l=$a$} pic{acc=0} (r)
  %   ;
\end{tikzpicture}};
    \end{tikzpicture}
  %}
  \caption{%
    Left: BA $\autex$ (dots represent accepting transitions).
    Right: the outcome of $\modcompl(\dac_{P_0}, \iwc_{P_1},\autex)$ with
    $\acccond\colon \Infof{\protect\tacc{0}} \land \Infof{\protect\tacc{1}}$.
    States are given as $(H, (C_0, S_0, B_0), (C_1, B_1))$; 
    to avoid too many braces, sets are given as sums.
    }
  \label{fig:example}
\end{figure}
% 
% \begin{minipage}[b]{3.5cm}
% \resizebox{\textwidth}{!}{
%   \input{figs/example_aut.tikz}
% }
% \caption{$\autex$ with $P_0 = \{q\}$ and $P_1 = \{r,s\}$}
% \label{fig:exampleAut}
% \end{minipage}
% \begin{minipage}[b]{9cm}
% \resizebox{\textwidth}{!}{
%   \input{figs/example_result.tikz}
% }
% \caption{Result of running $\modcompl(\dac_{P_0}, \iwc_{P_1},\autex)$.
%   states are generated as $(H, (C_0, S_0, B_0), (C_1, B_1))$. To avoid too many braces, sets are represented as sums}
% \label{fig:exampleResult}
% \end{minipage}
% \end{figure}
}

%%%%%%%%%%%%%%%%%%%%%%%%%%%%%%%%%%%%%%%%%%%%%%%%%%%%%%%%%%%%%%%%%%%%%
% \newcommand{\figExampleResult}[0]{
% \begin{figure}[t]
% \begin{center}
% \input{figs/example_result.tikz}
% \end{center}
% \caption{\ol{Result on} example automaton $\autex$}
% \label{fig:exampleAut}
% \end{figure}
% }

%*******************************************************************************
\vspace{-3.0mm}
\subsection{Basic Synchronous Algorithm}\label{sec:basic_synchronous}
\vspace{-2.0mm}
%*******************************************************************************

In this section, we describe the basic \emph{synchronous} top-level algorithm.
Then, in \cref{sec:modular-elevator}, we provide its instantiation for elevator
automata and give a~new upper bound for their complementation;
in \cref{sec:optimizations}, we discuss several optimizations of the algorithm;
and in \cref{sec:general}, we give a~generalization for unrestricted BAs.
Let us fix a~BA $\aut = (\states, \trans, \inits, \acc)$ and, w.l.o.g., assume
that~$\aut$ is \emph{complete}, i.e., $|I| > 0$ and all states $q \in
\states$ have an outgoing transition over all symbols $a \in \alphabet$.

% We will start with a~basic synchronous algorithm for elevator automata and then proceed by
% giving optimizations.
% In \cref{sec:general}, we further generalize the algorithm to BAs with no
% structural restrictions.
% %
% In the following, by a \emph{component} we mean the union of (possibly several) SCCs.
% A component is \emph{non-trivial} \vh{a better name?} if it contains at least one accepting state/transition.
% %

The synchronous algorithm works with partial complementation algorithms for BA's partition blocks.
Each such algorithm~$\alg$ is provided with a structural condition $\phi_\alg$ characterizing 
properties of partition blocks that the algorithm is able to complement. 
For a BA $\but$, we abuse the notation and use $\but \models \phi$ to denote that $\but$ satisfies 
the condition~$\phi$. 
We say that $\alg$ is a~\emph{partial complementation algorithm for a~partition block~$P$} if $\aut_P \models \phi_\alg$.
We distinguish between $\alg$, a~general algorithm able to complement a
partition block of a~given type, and $\alg_{P}$, its instantiation for the
partition block~$P$. 
We require each instance~$\alg_P$ to provide the
following:
% \vh{there is another issue; instance of $\alg$ should have in fact access also to $\aut$. But the instance 
% $\alg_P$ is not parameterized by $\aut$ (maybe we could just say that it is implicit; for a fixed $\aut$)}
%
\begin{itemize}
  \item  $\onetypeof{\alg_P}$ --- the type of the macrostates produced by the
    algorithm;

  \item  $\colourtypeof{\alg_P} = \{0, \ldots, k^{\alg_P}-1\}$ --- the set of used colours;

  \item  $\GetInitOf{\alg_P} \in 2^{\onetypeof{\alg_P}}$ --- the set of initial macrostates;

  \item  $\GetSuccOf{\alg_P}\colon (2^\states \times \onetypeof{\alg_P} \times \alphabet)
    \to 2^{\onetypeof{\alg_P} \times \colourtypeof{\alg_P}}$ --- a~function
    returning the successors of a~macrostate such that $\GetSuccOf{\alg_P}(H, M, a) =
    \{(M_1, \alpha_1), \ldots, (M_k, \alpha_k)\}$,
    where $H$~is the set of all states of~$\aut$ reached over the same word,
    $M$~is the $\alg_P$'s macrostate for the given partition block,
    $a$~is the input symbol, and
    each $(M_i, \alpha_i)$ is a~pair (\emph{macrostate}, \emph{set of colours}) such that~$M_i$ is
    a~successor of~$M$ over~$a$ w.r.t.~$H$ and~$\alpha_i$ is a~set of colours on the
    edge from~$M$ to~$M_i$ ($H$~helps to keep track of \emph{new} runs coming
    into the partition block); and

  \item  $\GetAccOf{\alg_P}\in \emersonleiof{\colourtypeof {\alg_P}}$ ---
    the acceptance condition.
\end{itemize}

Let $P_1, \ldots, P_n$ be a~partitioning of~$\aut$ (w.l.o.g.,
we assume that $n > 0$), and $\alg^1, \ldots, \alg^n$ be
a~sequence of algorithms such that~$\alg^i$ is a~partial complementation
algorithm for~$P_i$.
Furthermore, let us define the following auxiliary \emph{renumbering}
function~$\renum$ as $\renumof c j = c + \sum_{i=1}^{j-1}
|\colourtypeof{\alg^i_{P_i}}|$, which is used to make the colours and acceptance
conditions from the partial complementation algorithms disjoint.
We also lift~$\renum$ to sets of colours in the natural way, and also to
$\emersonlei$ conditions such that~$\renumof \varphi j$ has the same structure
as~$\varphi$ but each atom $\Infof c$ is substituted with the atom
$\Infof{\renumof c j}$ (and likewise for $\Fin$ atoms).
The synchronous complementation algorithm then produces 
the~TELA~$\modcompl(\alg^1_{P_1}, \dots, \alg^n_{P_n},\aut) = (\states^{\cut},
\trans^{\cut}, \inits^{\cut}, \colourset^{\cut}, \colouring^{\cut},
\acccond^{\cut})$ with components defined as follows (we use
$[S_i]_{i=1}^n$ to abbreviate $S_1 \times \cdots \times S_n$):
\vspace{-1mm}

\noindent
\begin{minipage}[t]{0.49\textwidth}
\begin{itemize}
  \item  $\states^{\cut} = 2^\states \times  [\onetypeof{\alg^i_{P_i}}]_{i=1}^n$,
  \item  $\inits^{\cut} = \{\inits\} \times [\GetInitOf{\alg^i_{P_i}}]_{i=1}^n$,
\end{itemize}
\end{minipage}
\begin{minipage}[t]{0.49\textwidth}
\begin{itemize}
  \item $\colourset^{\cut} = \{0, \ldots, \renumof{k^{\alg^n_{P_n}}-1}{n}\}$,
  \item $\acccond^{\cut} = \bigwedge_{i=1}^n \renumof{\GetAccOf{\alg^i_{P_i}}} i$,\footnotemark and
\end{itemize}
\end{minipage}
\footnotetext{%    The FOOTNOTEMARK is above
  If we drop the condition that~$\aut$ is complete, we also need to add an
  \emph{accepting sink state} (representing the case for $H=\emptyset$)
  with self-loops over all symbols marked by a~new colour $\tcacc{7}{$s$}$, and
  enrich $\acccond^{\cut}$ with $\ldots\lor \Infof{\tcacc{7}{$s$}}$.
  }
\begin{itemize}
\vspace{-2mm}
  \item  $\trans^{\cut}$ and $\colouring^{\cut}$ are defined such that
    if
    \vspace{-1.5mm}
    \begin{equation*}
    ((M'_1, \alpha_1), \ldots, (M'_n, \alpha_n)) \in [\GetSuccOf{\alg^i_{P_i}}(H, M_i, a)]_{i=1}^n ,\\[-1.5mm]
    \end{equation*}
    then $\trans^{\cut}$ contains the transition
    $t\colon (H, M_1, \ldots, M_n) \ltr a (\trans(H, a), M'_1, \ldots, M'_n)$,
    coloured by
    $\colouring^{\cut}(t) = \bigcup\{\renumof{\alpha_i}{i} \mid 1 \leq i \leq n\}$,
    and~$\trans^{\cut}$ is the smallest such a~set.
\end{itemize}
In order for $\modcompl$ to be correct, the partial complementation algorithms
need to satisfy certain properties, which we discuss below.

% \ol{the rest is probably not consistent with what is above here; fix it!}

% %%%%%%%%%%%%%%%%%%%%%%%%%%%%%%%%%%%%%%%%%%%%%%%%%%%%%%%%%%%%%%%%
% \paragraph{Requirements on partial complementation algorithms}
%
%Let $\alg$ be a partial complementation algorithm with the structural condition $\phi$. 

For a structural condition~$\varphi$ and a~BA $\but = (\states, \trans, \inits, \acc)$, we define $\but\modelssurrp \phi$ iff $\but \models \phi$, $P$~is 
a~partition block of~$\but$, and
$\but$ contains no accepting transitions outside~$P$.
% $\but = \but_P$.
% Informally, $\but\modelssurrp \phi$ if $\but$ contains a~partition block satisfying~$\phi$ and, moreover, it has no other accepting transitions outside~$P$.
We can now provide the correctness condition on~$\alg$.
% a~partition block of~$\but$, and $\but = \but_P$.
% Informally, $\but\modelssurrp \phi$ if $\but$ contains a~partition block satisfying~$\phi$ and, moreover, it has no other accepting transitions outside~$P$. Based on this notion,
% we are ready to provide the correctness condition on~$\alg$.

\vspace{-2mm}
\begin{definition}\label{def:alg-correct}
  We say that $\alg$ is \emph{correct} if for each 
  $\but$ such that $\but \modelssurrp \phi_\alg$ we have $\langof{\modcompl(\alg_{P}, \but)} = \Sigma^\omega\setminus\langof{\but}$.
\end{definition}

\vspace{-2mm}
The correctness of the synchronous algorithm (provided that each partial complementation algorithm is correct) is then established by Theorem \ref{thm:correctness}. 

% \vh{todo}
% \ly{The following theorem should use $H $ as the set of reached states, just to be consistent.}
\vspace{-2mm}
\begin{restatable}{theorem}{thmSynchrCorr}\label{thm:correctness}
  Let $\aut$ be a BA, $P_1, \ldots, P_n$ be a~partitioning of $\aut$, and $\alg^1, \ldots, \alg^n$ be
  a~sequence of partial complementation algorithms such that~$\alg^i$ is \emph{correct} for~$P_i$. 
  Then, we have $\langof{\modcompl(\alg^1_{P_1}, \dots, \alg^n_{P_n},\aut)} = \Sigma^\omega\setminus\langof{\aut}$.
\end{restatable}

% Let us now give two partial complementation algorithms for two specific types of
% SCCs: IWCs and DACs.

%\ly{The figure caption should use $H$ as the set of reachable states.}

% \figExampleResult   %%%%%%%%%%%%%

%%%%%%%%%%%%%%%%%%%%%%%%%%%%%%%%%%%%%%%%%%%%%%%%%%%%%%%%%%%%%%%%%%%%%%%%%%%%%%%%
\vspace{-2.0mm}
\section{Modular Complementation of Elevator Automata}\label{sec:modular-elevator}
\vspace{-2.0mm}
%%%%%%%%%%%%%%%%%%%%%%%%%%%%%%%%%%%%%%%%%%%%%%%%%%%%%%%%%%%%%%%%%%%%%%%%%%%%%%%%

In this section, we first give partial algorithms to complement partition blocks with only
accepting IWCs (\cref{sec:iwc-complement}) and partition blocks with only
DACs (\cref{sec:dac-complement}).
Then, in \cref{sec:elevator-upper}, we show that using our algorithm, the upper
bound on the size of the complement of elevator BAs is in $\bigO(4^n)$, which
is \emph{exponentially better} than the known upper bound $\bigO(16^n)$
established in~\cite{HavlenaLS22a}.
% In this section, we first show how to complement partition block with only accepting IWCs in Sect. \ref{sec:iwc-complement}, partition block with only DACs in Sect. \ref{sec:dac-complement}, initial deterministic partition block in Sect. \ref{sec:det-init}.
% Then we show in Sect. \ref{sec:elevator-upper} that elevator BAs can be complemented in time $\bigO(4^n)$, improving \emph{exponentially} the upper bound $\bigO(16^n)$ established in \cite{HavlenaLS22a}. 

%*******************************************************************************
\vspace{-2.0mm}
\subsection{Complementation of Inherently Weak Accepting Components}\label{sec:iwc-complement}
\vspace{-1.0mm}
%*******************************************************************************

%\ol{}
%\ol{change $S_i$ to sth else not to confuse with the S in NCSB}

%\ly{We need to clarify the element in a partition and a partition. Like a partition should be the sequence of elements that partitioned $\aut$, while en element in a partition should not be called a partition.}

First, we introduce a~partial algorithm $\iwc$ with the condition~$\phi_\iwc$ 
specifying that all SCCs in the~partition block $P$ are \emph{accepting} IWCs.
Let~$P$ be a~partition block of~$\aut$ such that $\aut_P \models \phi_\iwc$. 
Our proposed approach makes use of the Miyano-Hayashi construction~\cite{MiyanoH84}.
Since in accepting IWCs, all runs are accepting, the idea of the construction
is to accept words such that all runs over the words eventually leave~$P$.

% The idea to complement $P$ is that we want to accept all $\omega$-words~$w$
% over which no accepting runs of $\aut$ will stay inside $P$.
% Since all cycles in $P$ are accepting, we only need to check whether there are runs over $w$ that eventually stay inside $P$.
% That is, the complement language of $P$ corresponds to $\omega$-words over which all runs eventually leave $P$.

Therefore, we use a pair $(C, B)$ of sets of states as a~macrostate for complementing~$P$.
Intuitively, we use $C$ to denote the set of all runs of~$\aut$ that are in~$P$ ($C$~for ``\emph{check}'').
% , including those just entering from other partition blocks.
The set $B\subseteq C$ represents the runs being inspected whether they
leave~$P$ at some point ($B$~for ``\emph{breakpoint}'').
Initially, we let $C = \inits \cap P$ and also sample into breakpoint all runs
in~$P$, i.e., set~$B = C$.
Along reading an $\omega$-word $w$, if all runs that have entered~$P$
eventually leave~$P$, i.e., $B$ becomes empty infinitely often, the complement
language of~$P$ should contain $w$ (when $B$ becomes empty, we sample $B$ with
all runs from the current~$C$).
% In order to capture the newly coming runs from other partition blocks, we will assume a set $H$ of global reachable states of $\aut$ shared by \emph{all} partition blocks.
% That is, we update $C$ to the $P$-successors of $H$, rather than just the successors of $C$.
We formalize $\iwc_P$ as a~partial procedure in the framework
from \cref{sec:basic_synchronous} as follows:

\begin{itemize}
  \item  $\onetypeof {\iwc_P} = 2^P \times 2^P$,
    \hfill
    $\colourtypeof{\iwc_P} = \{\tcacc{1}{0}\}$, %(The color we will emit when $B$ becomes empty),
    \hfill
    $\GetInitOf{\iwc_P} = \{(\inits \cap P, \inits \cap P)\}$, %(At first, we expect all runs entering $P$ to leave, thus $B = C = \inits \cap P$),
  \item  $\GetAccOf{\iwc_P} = \Infof {\tcacc{1}{0}}$, and
    \hfill
    $\GetSuccOf{\iwc_P} (H, (C, B), a) = \{((C', B'), \alpha)\}$ where
    \begin{itemize}
      \item  $C' = \trans(H, a) \cap P$, %(By using global set of reachable states, we can keep track of the runs only in $P$, including those which just entered from other partition blocks),
     %\noindent
     \newline
    \vspace{-1mm}
    \begin{minipage}[t]{0.49\textwidth}
      %\begin{itemize}
      \item  $B' = \begin{cases}
          C'      & \text{if } B^\star = \emptyset \text{ for } B^\star = \trans(B, a) \cap C',  \\%\text{ and } C \text{ is accepting,} \\
        %  \emptyset & \text{if } B = \emptyset \text{ and } C \text{ is rejecting, and} \\
          B^\star & \text{otherwise, and}
        \end{cases}$
      %\end{itemize}
    \end{minipage}
    \noindent
    \hspace*{1cm}
    \begin{minipage}[t]{0.49\textwidth}
      %\begin{itemize}
        %(If all runs already in $P$ before have left $P$, we then put new runs under inspection, i.e., let $B'=C'$.
        %Otherwise, keep inspecting the old runs in $P$. )
      \item  $\alpha = \begin{cases}
          \{\tcacc{1}{0}\} & \text{if } B^\star = \emptyset \text{ and}\\
          \emptyset & \text{otherwise.}
      \end{cases}$
       %\end{itemize}
    \end{minipage}
    \end{itemize}
\end{itemize}
%\ly{We can also let $B^{\star} = \transscc(B, a)$.}

%\ol{note that it is deterministic}

%\ol{theorem about correctness}
We can see that checking whether $w$ is accepted by the complement of $P$
reduces to check whether~$B$ has been cleared infinitely often.
Since every time when $B$ becomes empty, we emit the color $\tcacc{1}{0}$, we have that $w$ is not accepted within $P$ if and only if $\tcacc{1}{0}$ occurs infinitely often.
Note that the transition function $\GetSuccOf{\iwc_P}$ is deterministic, i.e., there is exactly one successor.
% Lemma \ref{lem:correctness-iwcs} guarantees the correctness of our partial complementation algorithm~$\iwc$.
%
\begin{restatable}{lemma}{lemCorrIWCS}
\label{lem:correctness-iwcs}
  The partial algorithm $\iwc$ is correct.
\end{restatable}

%*******************************************************************************
\vspace{-2.0mm}
\subsection{Complementation of Deterministic Accepting Components}\label{sec:dac-complement}
\vspace{-1.0mm}
%*******************************************************************************

%\ol{}

In this section, we give a partial algorithm $\dac$ with the condition $\phi_\dac$ 
specifying that a~partition block~$P$ consists of \emph{DACs}.
Let~$P$ be a~partition block of~$\aut$ such that $\aut_P \models \phi_\dac$. 
% Recall that when dealing with IWCs, checking whether there exists an accepting run is equivalent to checking whether there exists a run staying in the given partition block.
% This is justified by the fact that all cycles in the partition block are accepting.
% While this checking is a bit too strong for DACs and NACs since there can be nonaccepting runs.
%
%It is very tricky to identify nonaccepting runs in NACs because a run can branch into multiple runs.
Our approach is based on the NCSB family of
algorithms~\cite{BlahoudekHSST16,ChenHLLTTZ18,BlahoudekDS20,HavlenaLS22b} for
complementing SDBAs, in particular the NCSB-MaxRank
construction~\cite{HavlenaLS22b}.
The algorithm utilizes the fact that runs in DACs are deterministic, i.e., they
do not branch into new runs.
Therefore, one can check that a~run is non-accepting if there is a~time point
from which the run does not see accepting transitions any more.
We call such a~run that does not see accepting transitions any more \emph{safe}.
Then, an $\omega$-word~$w$ is not accepted in~$P$ iff all runs over~$w$
in~$P$ either (i)~leave~$P$ or (ii)~eventually become safe.

For checking point~(i), we can use a~similar technique as in algorithm~$\iwc$,
i.e., use a~pair $(C,B)$.
Moreover, to be able to check point~(ii), we also use the set~$S$ that
contains runs that are supposed to be \emph{safe}, resulting in macrostates of
the form $(C,S,B)$\footnote{In contrast to $\iwc$, here we use $C \cup S$ rather than
$C$ to keep track of all runs in~$P$.}.
To make sure that all runs are deterministic, we will use~$\transscc$ instead
of~$\trans$ when computing the successors of~$S$ and~$B$ since there may be
nondeterministic jumps between different DACs in~$P$;
we will not miss any run in~$P$ since if a~run moves between DACs of $P$, it can be
seen as the run leaving~$P$ and a~new run entering~$P$.
Since a~run eventually stays in one SCC, this guarantees that the run will not be missed.
% It is known that a run will eventually stay in an SCC, so if the run eventually stay in an SCC of $P$, we will not miss it.

We formalize $\dac_P$ in the top-level framework as follows:
\begin{itemize}
  \item  $\onetypeof{\dac_P} = 2^P \times 2^P \times 2^P$, $\GetInitOf{\dac_P} = \{(\inits \cap P, \emptyset, \inits \cap P)\}$, % with macrostates denoted as $(C, S, B)$, %\ol{intuition}, where $C$ keeps track of the runs entering the partition block $P$ of DACs, $S$ corresponds to the runs that are guessed to not visit accepting transitions any more (or before exiting from $P$) and $B$ relates with the runs that will eventually either $S$ or exit from $P$.
  \item  $\colourtypeof {\dac_P} = \{\tcacc{0}{0}\}$, $\GetAccOf{\dac_P} = \Infof {\tcacc{0}{0}}$, and%(The color we emit when $B$ becomes empty),
  %\item  , %(Again, we want to check all runs entering $P$ at first, so put them also in $B$),
  \item  $\GetSuccOf{\dac_P} (H, (C, S, B), a) = U$ such that
    \begin{itemize}
      \item  if $\transacc(S, a) \neq \emptyset$, then $U = \emptyset$ (Runs in $S$ must be \emph{safe}),% (If the runs in $S$ visit accepting transitions, we terminate the exploration as our guesses for those runs are wrong. Note that we only have accepting transitions within an SCC.),
      \item  otherwise $U$ contains $((C', S', B'), c)$ where
        \begin{itemize}

          \item $S' = \transscc(S, a) \cap P$ %(Only care about the runs within the SCCs of $P$)
          , $C' = (\trans(H, a) \cap P) \setminus S'$, %(Ignore runs that are already being inspected in $S$, and may include new entering runs from outside $P$ through $\trans(H, a)$. Here $C \cup S$ corresponds to the set $C$ for the construction of IWCs as we need to put only new runs in $C$ in $B$ once $B$ becomes empty.)
          \newline
          \begin{minipage}[t]{0.49\textwidth}
          \item $B' = \begin{cases}
              C'      & \text{if } B^\star = \emptyset \text{ for } B^\star = \transscc(B, a) \text{, }\\
            %  \emptyset & \text{if } B = \emptyset \text{ and } C \text{ is rejecting, and} \\
              B^\star & \text{otherwise, and}
            \end{cases}$
          \end{minipage}
          \hspace*{8mm}
          \begin{minipage}[t]{0.49\textwidth}
            %(As before, we care about the runs within the SCCs of $P$ and put new runs under inspection once $B^{\star}$ is empty. Note that $\transscc(B, a) \subseteq C'$ since all successors remain in the same SCC of their predecessors.)

          \item $c = \begin{cases}
              \{\tcacc{0}{0}\} & \text{if } B^\star = \emptyset \text{,}\\
              \emptyset & \text{otherwise}.
          \end{cases}$
          \end{minipage}
        \end{itemize}
        Moreover, in the case $\transacc(B, a) \cap \transscc(B, a) = \emptyset$, then~$U$ also contains $((C'', S'', C''), \{\tcacc{0}{0}\})$ where $S'' = S' \cup B'$ and $C'' = C' \setminus S''$.
        %
        %\begin{itemize}
        %  \item  $S'' = S' \cup B'$ and
        %  \item  $C'' = C' \setminus S''$.
        %\end{itemize}
        %(When all runs in $B$ do not visit accepting transitions at this point, we make a guess that those runs will be safe from now on.
        %Thus we move all the runs into $S$ and put new runs under inspection again.)
    \end{itemize}
  \end{itemize}
%
%     \begin{equation*}
%       U = \begin{cases}
%         \emptyset & \text{if } \trans(S, a) \cap \acc \neq \emptyset \text{ or } \transacc(S, a) \neq \emptyset, \\
%         \{((C', S', C' \setminus S'), \{\tcacc{2}{$C$}\})\} & \text{if }
%         B^\star = \emptyset \text{ for } B^\star = \trans(B, a) \cap C', \text{else}\\
%         \{((C', S', B'), \emptyset), ((C'', S'', C'' \setminus S''), \{\tcacc{2}{$C$}\})\} & \text{if }
%         B \cap \acc = \emptyset \text{ and } \transacc(B, a) \cap C' = \emptyset \text{, and}\\
%         \{((C', S', B'), \emptyset)\} & \text{otherwise,}
%       \end{cases}
%     \end{equation*}
%     %
%     where
%     %
%     \begin{center}
%     \begin{minipage}{5cm}
%     \begin{itemize}
%       \item $S' = \trans(S, a) \cap \states_C$,
%       \item $C' = (\trans(P, a) \cap \states_C) \setminus S'$, 
%       \item $B' = \trans(B, a) \cap C'$,
%     \end{itemize}
%     \end{minipage}
%     \begin{minipage}[b]{5cm}
%       \begin{itemize}
%         \item $S'' = S' \cup B'$, and
%         \item $C'' = C' \setminus S''$.
%       \end{itemize}
%     \end{minipage}
%     \end{center}
%     %
%     % and
%     % $c_i = \begin{cases}
%     %   \{\colourzero\} & \text{if } B_i = \emptyset \text{ and}\\
%     %   \emptyset & \text{otherwise, and}
%     % \end{cases}$
%   \item $\GetAccOf{\dac_C} = \Infof{\tcacc{2}{$C$}}$.
% \end{itemize}

\noindent
Intuitively, when $\transacc(B, a) \cap \transscc(B, a) = \emptyset$, we make two guesses:
\begin{inparaenum}[(i)]
  \item  either the runs in~$B$ all become safe (we move them to~$S$) or
  \item  there might be some unsafe runs (we keep them in~$B$).
\end{inparaenum}
Since the runs in~$B$ are deterministic, the number of tracked runs in~$B$ will not increase.
Moreover, if all runs in $B$ are eventually safe, we are guaranteed to move all of them to $S$ at the right time point, e.g., the maximal time point where all runs are safe since the number of runs is finite.
   
% \ol{finish}
% \ol{when dealing with S, and moving between components, we need to move runs from S to C}

% \ol{also for MH, we can remove runs leaving SCC from breakpoint}

% \vspace{20mm}

\begin{figure}[t]
  \centering
  %\resizebox{\linewidth}{!}{
    \begin{tikzpicture}
      \node[scale=0.7] (aex) {% \begin{tikzpicture}[>=stealth',shorten >=0pt,node distance=1.4cm,
%                     scale=0.8,transform shape,initial text={}]
\begin{tikzpicture}[automaton]
  \tikzstyle{every state}=[inner sep=3pt,minimum size=5pt]
  \tikzstyle{empty}=[]
  \tikzstyle{initstate}=[fill=yellow!30]
  
  \path[use as bounding box] (-0.75,-2.75) rectangle (3.5,2.75);

  \node[state,initial,initstate] (p) at (0,0) {$p$};
  \node[state,above right of=p,xshift=10mm] (q) {$q$};
  \node[state,below right of=p,xshift=10mm] (r) {$r$};
  \node[state] (s) at ($(q)!0.5!(r)$) {$s$};

  \node[font=\Large] at (0,2) {$\autex$};
  
  \draw[dotted] ($(p.north west) + (-0.15,0.9)$) rectangle ($(p.south east) + (0.15,-0.15)$);
  \draw[dotted] ($(q.north west) + (-0.15,0.75)$) rectangle ($(q.south east) + (0.75,-0.15)$);
  \node[anchor=south] at ($(q.north) + (0.30,0.75)$) {$P_{0}$};
  \draw[dotted] ($(s.north west) + (-0.35,0.15)$) rectangle ($(r.south east) + (0.75,-0.15)$);
  \node[anchor=north] at ($(r.south) + (0.20,-0.15)$) {$P_{1}$};

  \path[->]
    (p) edge[loop above] node[auto] {$a,b$} (p)
    (p) edge node[above] {$a,b$} (q)
    % (q) edge[loop above] pic {l=$a$} pic[anchor=center] {acc=0} (q)
    (q) edge[loop above] node[auto] {$a$} node[anchor=center] {$\bullet$} (q)
    (q) edge[loop right] node {$b$} (q)
    (q) edge node[auto] {$a$} (s)
    (p) edge node[below] {$a$} (r)
    (r) edge[loop right] node[auto] {$b$} node[anchor=center] {$\bullet$} (r)
    (r) edge[bend left] node[left] {$b$} (s)
    (s) edge[bend left] node[auto] {$a$} node[anchor=center] {$\bullet$} (r)
    ;
\end{tikzpicture}};
      \node[anchor=west,scale=0.7] (comp) at (aex.east) {% \begin{tikzpicture}[>=stealth',shorten >=0pt,node distance=1.4cm,
%                     scale=0.8,transform shape,initial text={}]
\begin{tikzpicture}[automaton]
  \tikzstyle{every state}=[inner sep=3pt,minimum size=5pt]
  \tikzstyle{empty}=[]
  \tikzstyle{initstate}=[fill=yellow!30]
  % \tikzstyle{every state}=[rectangle,rounded corners,inner sep=3pt,minimum size=5pt]
  \tikzstyle{uberstate}=[
    rounded corners,draw,anchor=base,
    rectangle split,rectangle split horizontal,rectangle split parts=3,
    rectangle split part align=base,
    rectangle split part fill={black!20, blue!30, green!30}]
  \newcommand{\ustate}[6]{$#1$\nodepart{two}$#2, #3, #4$\nodepart{three}$#5, #6$}

  \path[use as bounding box] (8.5,0.4) rectangle (-1.5,-4.7);
  
%  \node [my shape=5, rectangle split horizontal] at (2,2)
%     {1\nodepart{two}2\nodepart{three}3\nodepart{four}4\nodepart{five}5};

  \node[uberstate,initial] (p) at (0,0) {\ustate p \emptyset \emptyset \emptyset \emptyset \emptyset};
  \node[uberstate,below of=p,xshift=0mm] (pq) {\ustate{p+q} q \emptyset q \emptyset \emptyset};
  \node[uberstate,below of=pq,xshift=0mm] (pqsafe) {\ustate{p+q} \emptyset q \emptyset \emptyset \emptyset};

  \node[uberstate,right of=p,xshift=30mm] (pqr) {\ustate{p+q+r} q \emptyset q r r};
  \node[uberstate,below of=pqr,xshift=0mm] (pqrs) {\ustate{p+q+r+s} q \emptyset q {r+s} {r+s}};
  \node[uberstate,below of=pqrs,xshift=0mm] (pqrsbreak) {\ustate{p+q+r+s} q \emptyset q {r+s} {r}};
  \node[uberstate,below of=pqrsbreak,xshift=0mm] (pqrssafe) {\ustate{p+q+r+s} \emptyset q \emptyset {r+s} {r+s}};

  \path[->]
    (p) edge pic[pos=0.3] {acc=0} pic[pos=0.6] {acc=1} pic[auto] {l=$b$} (pq)
    (pq) edge pic[pos=0.3] {acc=0} pic[pos=0.6] {acc=1} pic[auto] {l=$b$} (pqsafe)
    (pq.182) edge[out=210,in=150,loop,distance=9mm] pic {acc=1} pic[auto] {l=$b$} (pq.178)
    (pqsafe.182) edge[out=210,in=150,loop,distance=9mm] pic[pos=0.25] {acc=0} pic[pos=0.75] {acc=1} pic[auto] {l=$b$} (pqsafe.178)
    (p) edge pic[pos=0.41] {acc=0} pic[pos=0.59] {acc=1} pic[auto] {l=$a$} (pqr)
    (pqr) edge[bend left] pic {acc=1} pic[auto] {l=$a$} (pqrs)
    (pqr) edge[bend right] node[left] {$b$} (pqrs)
    (pqrs.12) edge[out=120,in=60,loop,distance=6mm] node[auto] {$b$} (pqrs.11)
    (pq) edge pic {acc=1} pic[auto] {l=$a$} (pqrs)
    (pqrs) edge[bend left] node[auto] {$a$} (pqrsbreak)
    (pqrsbreak) edge[bend left] node[auto] {$b$} (pqrs)
    (pqrsbreak.160) edge[bend left] pic {acc=1} pic[auto] {l=$a$} (pqrs.200)
    (pqrsbreak) edge pic {acc=0} pic[auto] {l=$b$} (pqrssafe)
    (pqrssafe.182) edge[out=210,in=150,loop,distance=9mm] pic {acc=0} pic[auto] {l=$b$} (pqrssafe.178)
    (pqrs.187) edge[bend right] pic {acc=0} pic[left] {l=$b$} (pqrssafe.173)
    (pqr.355) edge[bend left=40] pic {acc=0} pic[auto] {l=$b$} (pqrssafe.2)
    ;
  %   (p) edge[loop above] node {$a,b$} (p)
  %   (p) edge node[above] {$a,b$} (pq)
  %   (q) edge[loop above] pic {l=$a$} pic[anchor=center] {acc=0} (pq)
  %   (pq) edge[loop below] node {$b$} (pq)
  %   (q) edge node[auto] {$a$} (s)
  %   (p) edge node[below] {$a$} (r)
  %   (r) edge[loop below] pic{l=$b$} pic[anchor=center] {acc=0} (r)
  %   (r) edge[bend left] node[above] {$b$} (s)
  %   (s) edge[bend left] pic[below]{l=$a$} pic{acc=0} (r)
  %   ;
\end{tikzpicture}};
    \end{tikzpicture}
  %}
  \caption{%
    Left: BA $\autex$ (dots represent accepting transitions).
    Right: the outcome of $\modcompl(\dac_{P_0}, \iwc_{P_1},\autex)$ with
    $\acccond\colon \Infof{\protect\tacc{0}} \land \Infof{\protect\tacc{1}}$.
    States are given as $(H, (C_0, S_0, B_0), (C_1, B_1))$; 
    to avoid too many braces, sets are given as sums.
    }
  \label{fig:example}
\end{figure}
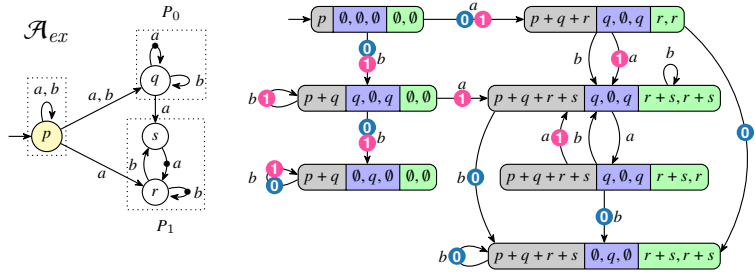
% 
% \begin{minipage}[b]{3.5cm}
% \resizebox{\textwidth}{!}{
%   \input{figs/example_aut.tikz}
% }
% \caption{$\autex$ with $P_0 = \{q\}$ and $P_1 = \{r,s\}$}
% \label{fig:exampleAut}
% \end{minipage}
% \begin{minipage}[b]{9cm}
% \resizebox{\textwidth}{!}{
%   \input{figs/example_result.tikz}
% }
% \caption{Result of running $\modcompl(\dac_{P_0}, \iwc_{P_1},\autex)$.
%   states are generated as $(H, (C_0, S_0, B_0), (C_1, B_1))$. To avoid too many braces, sets are represented as sums}
% \label{fig:exampleResult}
% \end{minipage}
% \end{figure}
   %%%%%%%%%%%%%

As mentioned above, $w$~is not accepted within $P$ iff all runs over $w$ either
(i)~leave~$P$ or (ii)~become safe.
In the context of the presented algorithm, this corresponds to
(i)~$B$~becoming empty infinitely often and
(ii)~$\transacc(S, a)$ never seeing an accepting transition.
Then we only need to check if there exists an infinite sequence
of macrostates $\hat{\rho} = (C_0, S_0, B_0) \ldots$
\mbox{that emits~$\tcacc{0}{0}$ infinitely often.}

\begin{restatable}{lemma}{lemCorrDAC}
  \label{lem:dac-correctness}
  The partial algorithm $\dac$ is correct.
\end{restatable}

It is worth noting that when the given partition block~$P$ contains all DACs
of~$\aut$, we can still use the construction above, while the construction
in~\cite{HavlenaLS22b} only works on~SDBAs.

\vspace{-2mm}
\begin{example}
In \cref{fig:example}, we give an example of the run of our algorithm on the
BA~$\autex$.
The BA contains three SCCs, one of them (the one containing~$p$) non-accepting
(therefore, it does not need to occur in any partition block).
The partition block~$P_0$ contains a~single DAC, so we can use
algorithm~$\dac$, and the partition block~$P_1$ contains a~single accepting
IWC, so we can use~$\iwc$.
The resulting $\modcompl(\dac_{P_0}, \iwc_{P_1},\autex)$ uses two colours,
$\tacc{0}$~from~$\dac$ and~$\tacc{1}$~from~$\iwc$.
The acceptance condition is $\Infof{\tacc 0} \land \Infof{\tacc 1}$.
\qed
\end{example}

% \vspace{20mm}

% \ol{ToDo}
\newcommand{\figInitialDAC}[0]{
\begin{figure}[t]
  \centering
%   \resizebox{\linewidth}{!}{
    \begin{tikzpicture}
      \node (bex) {\begin{tikzpicture}[automaton]
    \tikzstyle{every state}=[inner sep=3pt,minimum size=5pt]
    \tikzstyle{empty}=[]
    \tikzstyle{initstate}=[fill=yellow!30]
  
    \path[use as bounding box] (-1.5,-1.125) rectangle (2.75,1.5);

    \node[state,initial,initstate] (p) at (0,0) {$p$};
    \node[state,right of=p,xshift=0mm] (q) {$q$};
    
    \node[font=\large] at (-1.25,0) {$\butex$};

    \draw[dotted] ($(p.north west) + (-0.15,0.9)$) rectangle ($(p.south east) + (0.15,-0.9)$);
    \node[anchor=south] at ($(p.north) + (0,0.75)$) {$P_{0}$};
    \draw[dotted] ($(q.north west) + (-0.15,0.9)$) rectangle ($(q.south east) + (0.15,-0.15)$);
    \node[anchor=south] at ($(q.north) + (0,0.75)$) {$P_{1}$};

    % \draw[dotted] ($(p.north west) + (-0.15,0.9)$) rectangle ($(p.south east) + (0.15,-0.15)$);
    % \draw[dotted] ($(q.north west) + (-0.15,0.75)$) rectangle ($(q.south east) + (0.75,-0.15)$);
    % \draw[dotted] ($(s.north west) + (-0.35,0.15)$) rectangle ($(r.south east) + (0.75,-0.15)$);
  
    \path[->]
      (p) edge[loop above] node[auto] {$a$} node[anchor=center] {$\bullet$} (p)
      (p) edge[loop below] node[auto] {$b$} (p)
      (p) edge node[above] {$a$} (q)
      (q) edge[loop above] node[auto] {$b$} node[anchor=center] {$\bullet$} (q)
      ;
  \end{tikzpicture}};
      \node[anchor=west] (comp) at (bex.east) {\begin{tikzpicture}[automaton]
    \tikzstyle{every state}=[inner sep=3pt,minimum size=5pt]
    \tikzstyle{empty}=[]
    \tikzstyle{initstate}=[fill=yellow!30]
    % \tikzstyle{every state}=[rectangle,rounded corners,inner sep=3pt,minimum size=5pt]
    \tikzstyle{uberstate}=[
      rounded corners,draw,anchor=base,
      rectangle split,rectangle split horizontal,rectangle split parts=3,
      rectangle split part align=base,
      rectangle split part fill={black!20, blue!30, green!30}]
    \newcommand{\ustate}[4]{$#1$\nodepart{two}$#2$\nodepart{three}$#3, #4$}
  
    \path[use as bounding box] (-1.5,-1.125) rectangle (4.75,1.5);
  
    \node[uberstate,initial] (p) at (0,0) {\ustate p p \emptyset \emptyset};
    \node[uberstate,right of=p,xshift=15mm] (pq) {\ustate{p+q} p q q};

    \path[->]
      (p) edge pic[pos=0.3] {cacc={4}{0}} pic[pos=0.6] {acc=1} pic[auto] {l=$a$} (pq)
      (p) edge[loop above,distance=9mm] pic {acc=1} pic[auto] {l=$b$} (p)
      (pq) edge[loop above,distance=9mm] node {$b$} (pq)
      (pq) edge[in=-115,out=-65,loop,distance=9mm] pic[pos=0.25] {acc=1} pic[pos=0.75] {cacc={4}{0}} pic[auto] {l=$a$} (pq)
      ;
    %   (p) edge[loop above] node {$a,b$} (p)
    %   (p) edge node[above] {$a,b$} (pq)
    %   (q) edge[loop above] pic {l=$a$} pic[anchor=center] {acc=0} (pq)
    %   (pq) edge[loop below] node {$b$} (pq)
    %   (q) edge node[auto] {$a$} (s)
    %   (p) edge node[below] {$a$} (r)
    %   (r) edge[loop below] pic{l=$b$} pic[anchor=center] {acc=0} (r)
    %   (r) edge[bend left] node[above] {$b$} (s)
    %   (s) edge[bend left] pic[below]{l=$a$} pic{acc=0} (r)
    %   ;
  \end{tikzpicture}};
    \end{tikzpicture}
%   }
  \caption{%
    Left: $\butex$.
    \mbox{Right: $\modcompl(\idac_{P_{0}}, \iwc_{P_{1}},\butex)$ with $\acccond\colon \Finof {\protect\tcacc{4}{0}} \wedge \Infof{\protect\tacc{1}}$.}
    }
  \label{fig:example-aut-idec-result}  
%   \begin{minipage}[b]{3.5cm}
%     \centering
%     \scalebox{0.8}{
%     \input{figs/example-aut-idec.tikz}
%     }
%     \caption{$\butex$ with $P_0 = \{p\}$ and $P_1 = \{ q \}$.}
%     \label{fig:example-aut-idec}
%   \end{minipage}
%   \begin{minipage}[b]{9cm}
%     \centering
%     \scalebox{0.9}{
%     \input{figs/ex-idec-res.tikz}
%     }
%     \caption{Result of running $\modcompl(\idac_{P_0}, \iwc_{P_1},\butex)$ with the accepting condition $\Finof {\protect\tacc{0}} \wedge \Infof{\protect\tacc{1}}$.}
%     \label{fig:example-idec-result}
%   \end{minipage}
\end{figure}
}

%*******************************************************************************
\vspace{-2.0mm}
\subsection{Upper-bound for Elevator Automata Complementation}\label{sec:elevator-upper}
\vspace{-1.0mm}
%*******************************************************************************

We now give an upper bound on the size of the complement generated by our
algorithm for elevator automata, which significantly improves the best
previously known upper bound of $\bigO(16^n)$~\cite{HavlenaLS22a} to
$\bigO(4^n)$, the same as for SDBAs, which are a~strict subclass of elevator
automata~\cite{BlahoudekHSST16} (we note that this upper bound cannot be
obtained by a~determinization-based algorithm, since
determinization of SDBAs is in $\Omega(n!)$~\cite{EsparzaKRS17,Loding99}).

\vspace{-1mm}
\begin{restatable}{theorem}{thmElevatorBound}
  \label{thm:elevator_upper_bound}
  Let~$\aut$ be an elevator automaton with~$n$ states.
  Then there exists a~BA with~$\bigO(4^n)$ states accepting the complement
  of~$\langof \aut$.
\end{restatable}
%Assume that $Q_D$ is the union of all SCCs of $\aut$ satisfying $\phi_\dac$ and 
%  $Q_W$ is the union of all SCCs satisfying $\phi_{\iwc}$ and moreover 
%  $Q_D \cap Q_W = \emptyset$. Since $\aut$ is elevator, $Q_D \cup Q_W$ is the union 
%  of all partition blocks of $\aut$.
\begin{proof}[Sketch]
Let~$Q_W$ be all states in accepting IWCs, $Q_D$~be all states in DACs, and
$Q_N$~be the remaining states, i.e., $Q = Q_W \uplus Q_D \uplus Q_N$.
We make two partition blocks: $P_0 = Q_W$ and $P_1 = Q_D$ and use~$\iwc$ and
$\dac$ respectively as the partial algorithms, with macrostates of the form
$(H, (C_0, B_0), (C_1, S_1, B_1))$.
For each state~$q_N \in Q_N$, there are two options: either $q_N \notin H$ or
$q_N \in H$.
For each state~$q_W \in Q_W$, there are three options:
\begin{inparaenum}[(i)]
  \item  $q_W \notin C_0$,
  \item  $q_W \in C_0 \setminus B_0$, or
  \item $q_W \in C_0 \cap B_0$.
\end{inparaenum}
Finally, for each $q_D \in Q_D$, there are four options:
\begin{inparaenum}[(i)]
  \item  $q_D \notin C_1 \cup S_1$,
  \item  $q_D \in S_1$,
  \item  $q_D \in C_1 \setminus B_1$, or
  \item  $q_D \in C_1 \cap B_1$.
\end{inparaenum}
Therefore, the total number of macrostates is $2 \cdot 2^{|Q_N|} \cdot
3^{|Q_W|} \cdot 4^{|Q_D|} \in \bigO(4^n)$ where the initial factor~$2$ is due
to degeneralization from two to one colour (the two colours can actually be
avoided by using our shared breakpoint optimization from
\cref{sec:shared-breakpoint}).
% The complexity result directly follows from that (1) complementing the partition block of IWCs is in $3^{|Q_W|}$ where $Q_W$ is the states in all accepting IWCs, (2) the number of possible macrostates $(C, S, B)$ is $4^{|Q_D|}$ where $Q_D$ contains all DACs and (3) the number of possibilities for states in nonaccepting IWCs in reachable states $H$ is $2^{|Q_N|}$ where $Q_N$ is the union of all nonaccepting SCCs.
%   % From \cref{thm:correctness}, \cref{lem:correctness-iwcs}, and \cref{lem:dac-correctness}
% %  we have that $\langof{\modcompl(\dac_{Q_D}, \iwc_{Q_W},\aut)} = \Sigma^\omega\setminus\langof{\aut}$.
% %  We now compute the number of states of $\modcompl(\dac_{Q_D}, \iwc_{Q_W},\aut)$.
% \iffalse
%   For a state $q \in Q_D$ there are 4 possibilities of distributing $q$ 
%   within $\onetypeof{\dac_{Q_D}}$: 
%   %
%   \begin{inparaenum}[(i)]
%     \item $q\notin C \cup S$,
%     \item $q \in C$,
%     \item $q\in C \cap B$,
%     \item $q\in S$.
%   \end{inparaenum}
%   %
%   For a state $q\in Q_W$ there are 3 possibilities of distributing $q$
%   within $\onetypeof{\iwc_{Q_W}}$:
%   %
%   \begin{inparaenum}[(i)]
%     \item $q\notin C$,
%     \item $q\in C$,
%     \item $q \in C \cap B$.
%   \end{inparaenum} 
% \fi
%   Therefore, the number of macrostates in BAs is given as $2 \times 4^{|Q_D|}\cdot 3^{|Q_W|} \cdot 2^{|Q_N|}\in \bigO(4^n)$ where the factor $2$ is due to the reduction of $2$ colors to $1$ color.
\qed
\end{proof}

%*******************************************************************************
\vspace{-0.0mm}
\section{Optimizations of the Modular Construction}\label{sec:optimizations}
\vspace{-0.0mm}
%*******************************************************************************

%\ol{modify according to what is really here}
In this section, we propose optimizations of the basic modular algorithm. 
In \cref{sec:det-init}, we give a partial algorithm to complement initial partition blocks with DACs.
Further, in \cref{sec:postponed}, we propose the postponed construction allowing to 
use automata reduction on intermediate results. In \cref{sec:round-robin},
we propose the round-robin algorithm alleviating the problem with the explosion 
of the size of the Cartesian product of partial successors.
In \cref{sec:shared-breakpoint}, we 
provide an optimization for partial algorithms that are based on the breakpoint construction, 
and, finally, in \cref{sec:simulations}, we show how to employ simulation to
decrease the size of macrostates in the synchronous construction.

%*******************************************************************************
\vspace{-0.0mm}
\subsection{Complementation of Initial Deterministic Partition Blocks}\label{sec:det-init}
\vspace{-0.0mm}
%*******************************************************************************

Our first optimization is an optimized algorithm~$\idac$ for a~subclass of
partition blocks containing DACs.
In particular, the condition $\varphi_{\idac}$ specifies 
that the partition block~$P$ is deterministic and can be reached only deterministically
in~$\aut$ (i.e., $\aut_P$~after removing redundant states is deterministic). In that case 
we say that~$P$ is an \emph{initial deterministic} partition block.
The algorithm is based on complementation of deterministic BAs into co-\buchi
automata.

The algorithm $\idac_P$ is formalized below:
\begin{itemize}
  \item  $\onetypeof {\idac_P} = P \cup \{ \emptyset \}$,
    \hfill
  % \item 
    $\GetInitOf{\idac_P} = I \cap P$,
    \hfill
  % \item 
    $\colourtypeof{\idac_P} = \{\tcacc{4}{0}\}$,
    \hfill
  % \item
    $\GetAccOf{\idac_P} = \Finof {\tcacc{4}{0}}$,
  \item  $\GetSuccOf{\idac_P} (H, q, a) = \{ (q', \alpha) \}$ where 

    \noindent
    \begin{minipage}[t]{0.49\textwidth}
      \begin{itemize}
        \item $q' = 
        \begin{cases}
          r & \text{if } \delta(H,a)\cap P = \{ r \} \text{ and} \\
          \emptyset & \text{otherwise},
        \end{cases}$ %(Note that $q$ can be an empty set or a state in $P$.)
      \end{itemize}
    \end{minipage}
    \begin{minipage}[t]{0.49\textwidth}
      \begin{itemize}
        \item $\alpha = 
        \begin{cases} 
          \{ \tcacc{4}{0} \} & \text{if } q \ltr a {q'} \in F \text{ and}\\
          \emptyset & \text{otherwise.}
        \end{cases}$
      \end{itemize}
    \end{minipage}
  % \begin{itemize}
  %   \item $q' = 
  %   \begin{cases}
  %     r & \text{if } \delta(H,a)\cap P = \{ r \} \text{ and} \\
  %     \emptyset & \text{otherwise},
  %   \end{cases}$ %(Note that $q$ can be an empty set or a state in $P$.)
  %   % \begin{cases}
  %   %   r & \text{if } U = \{ r \} \\
  %   %   \emptyset & \text{otherwise, where } U = \delta(H,a)\cap P
  %   % \end{cases}$ %(Note that $q$ can be an empty set or a state in $P$.)
  %   \item $\alpha = 
  %   \begin{cases} 
  %     \{ \tcacc{4}{0} \} & \text{if } q \ltr a {q'} \in F \text{ and}\\
  %     \emptyset & \text{otherwise.}
  %   \end{cases}$
  % \end{itemize}
\end{itemize}

% After removing redundant states, $\aut_P$ is deterministic.
Intuitively, all runs reach $P$ deterministically, which means that over a
word~$w$, at most one run can reach~$P$.
Thus, we have $|\trans(H, w_{j}) \cap P| = 1 $ for some $j \geq 0$ if there is a run over $w$ to $P$, corresponding to $\delta(H,a)\cap P = \{r\}$ in the construction.
To check whether $w$ is not accepted in $P$, we only need to check whether the
run from $r \in P$ over~$w$ visits accepting transitions only finitely often.
We give an example of complementation of a~BA containing an initial
deterministic partition block in \cref{fig:example-aut-idec-result} in \cref{sec:examples}.
Notice that the use of the $\Fin$ condition helps to obtain a more concise
automaton with only two states (even in this simple example, using $\dac$
instead of $\idac$ would yield a~TELA with 4~states).
\begin{restatable}{lemma}{lemCorrIDAC}
  The partial algorithm $\idac$ is correct.
\end{restatable}

%%%%%%%%%%%%%%%%%%%%%%%%%%%%%%%%%%%%%%%%%%%%%%%%%%%%%%%%%%%%%%%%%%%%%%%%%

%*******************************************************************************
\vspace{-3.0mm}
\subsection{Postponed Construction}\label{sec:postponed}
\vspace{-2.0mm}
%*******************************************************************************

The modular synchronous construction from \cref{sec:basic_synchronous} utilizes
the assumption that in the simultaneous construction of successors for each
partition block over~$a$, if one partial macrostate~$M_i$ does not have
a~successor over~$a$, then there will be no successor of the $(H, M_1, \ldots,
M_n)$ macrostate in~$\delta^{\cut}$ as well.
This is useful, e.g., for inclusion testing, where it is not necessary to
generate the whole complement.
On the other hand, if we need to generate the whole automaton, a drawback of the 
proposed modular construction is that each partial complementation algorithm itself may generate a lot 
of useless~states. In this section, we propose the \emph{postponed construction}, which
complements the partition blocks (with their surrounding) independently and
later combines the intermediate results to obtain the complement
automaton for~$\aut$.
The main advantage of the postponed construction is that one can apply automata
reduction (e.g., based on removing useless states or using
simulation~\cite{ClementeM19,EtessamiWS05,AbdullaCHV14,bustan2003simulation})
to decrease the size of the intermediate automata.

% For a component $C$ of $\aut$ we use $\aut_C$ to denote the BA $(\states, \trans, \inits, \acc')$
% where $\acc'$ is the restriction of $\acc$ to states/transitions in $C$ only. 
% If $\alg$ is correct partial complementation algorithm for $C$, we, hence, have that
% $\aut_C \modelssurrc \phi_\alg$. 
In the postponed construction, we use automata product operation implementing language 
intersection (i.e., for two TELAs $\but_1$ and $\but_2$, a product automaton
$\but_1\cap\but_2$ satisfying $\langof{\but_1 \cap \but_2} = \langof{\but_1} \cap \langof{\but_2}$\footnote{Alternatively, one might also avoid the
product and generate linear-sized \emph{alternating} TELA, but working with
those is usually much harder and not used in practice.}).
Further, we employ a~function $\reduce$ performing some language-preserving reduction of an input 
TELA. Then, the postponed construction for an elevator automaton $\aut$
with a~partitioning $P_1, \ldots, P_n$ and 
a~sequence of algorithms $\alg^1, \ldots, \alg^n$ such that $\alg^i$ is a~partial 
complementation algorithm for~$P_i$, is defined as follows:
\begin{equation}
  \postponedcompl(\alg^1_{P_1}, \dots, \alg^n_{P_n},\aut) = \bigcap_{i=1}^n \reduce\left(\modcompl(\alg^i_{P_i}, \aut_{P_i})\right).
\end{equation}
The example of the postponed construction applied on the BA from \cref{fig:example} is shown in 
\cref{sec:examples}.
The correctness of the construction is then summarized by the following 
theorem.
\begin{restatable}{theorem}{thmPostp}
  Let $\aut$ be a BA, $P_1, \ldots, P_n$ be a partitioning of $\aut$, and $\alg^1, \ldots, \alg^n$ be
  a~sequence of partial complementation algorithms such that~$\alg^i$ is \emph{correct} for~$P_i$. 
  Then, $\langof{\postponedcompl(\alg^1_{P_1}, \dots, \alg^n_{P_n},\aut)} = \Sigma^\omega\setminus\langof{\aut}$.
\end{restatable}

%*******************************************************************************
\vspace{-3.0mm}
\subsection{Round-Robin Algorithm}\label{sec:round-robin}
\vspace{-2.0mm}
%*******************************************************************************

The proposed basic synchronous approach from \cref{sec:basic_synchronous} may suffer from the 
combinatorial explosion because the successors of a~macrostate are given by 
the Cartesian product of all successors of the partial macrostates.
To alleviate this explosion, we propose a~\emph{round-robin} top-level algorithm.
Intuitively, the round-robin algorithm actively tracks runs in only one
partial complementation algorithm at a~time (while other algorithms stay passive).
The algorithm periodically changes the active algorithm to avoid starvation
(the decision to leave the active state is, however, fully directed
by the partial complementation algorithm).
This can alleviate an explosion in the number of successors for algorithms that
generate more than one successor (e.g., for rank-based algorithms where one
needs to make a~nondeterministic choice of decreasing ranks of states in order
to be able to
accept~\cite{KupfermanV01,FriedgutKV06,Schewe09,ChenHL19,HavlenaL21,HavlenaLS22a};
such a~choice needs to be made only in the active phase while in the passive
phase, the construction just needs to make sure that the run is consistent with
the given ranking, which can be done deterministically).

The round-robin algorithm works on the level of \emph{partial complementation round-robin algorithms}.
Each instance of the partial algorithm provides \emph{passive types} to represent partial macrostates that 
are passive and \emph{active types} to represent currently active 
partial macrostates. In contrast to the basic partial complementation algorithms
from \cref{sec:basic_synchronous}, which
provide only a~single successor function, the round-robin partial algorithms 
provide several variants of them.
In particular, $\GetSuccTOf{}$ returns (passive) successors of a~passive partial
macrostate, $\LiftOf{}$ gives all possible active counterparts of a~passive
macrostate, 
and $\GetSuccAOf{}$ returns successors of an active partial macrostate.
If $\GetSuccAOf{}$ returns a partial macrostate of the passive type, the round-robin algorithm 
promotes the next partial algorithm to be the active one.
For instance, in the round-robin version of $\dac$, the passive type does not
contain the breakpoint and only checks that safe runs stay safe, so it is
deterministic.
Due to space limitations, we give a~formal definition and more details about
the round-robin algorithm in \cref{sec:round_robin_details}.

%*******************************************************************************
\vspace{-3.0mm}
\subsection{Shared Breakpoint}\label{sec:shared-breakpoint}
\vspace{-2.0mm}
%*******************************************************************************

The partial complementation algorithms $\dac$ and $\iwc$ (and later $\nacrnk$ defined 
in \cref{sec:nac-rank-based}) use a~breakpoint to check whether the runs under 
inspection are accepting or not. As an optimization, we consider merging of
breakpoints of several algorithms and keeping only a~single breakpoint for all
supported algorithms.
The top-level algorithm then needs to manage only one breakpoint and emit
a~colour only if this sole breakpoint becomes empty.
This may lead to a smaller number of generated macrostates since 
we synchronize the breakpoint sampling among several algorithms.
The second 
benefit is that this allows us to generate fewer colours (in the case of
elevator automata complemented using algorithms $\dac$ and $\iwc$, we get only
one colour).

%*******************************************************************************
% \vspace{-3.0mm}
\subsection{Simulation Pruning}\label{sec:simulations}
% \vspace{-2.0mm}
%*******************************************************************************

Our construction can be further optimized by a~simulation (or other compatible)
relation for pruning macrostates.\footnote{This optimization can be seen as
a~generalization of the simulation-based pruning techniques that appeared,
e.g., in~\cite{LodingP19,HavlenaLS22b} in the context of concrete
determinization/complementation procedures.  Here, we generalize the technique
to all procedures that are based on run tracking.}
A~simulation is, broadly speaking, a~relation ${\simul} \subseteq Q\times Q$
implying language inclusion of states, 
i.e., $\forall p, q \in Q\colon p\simul q \Longrightarrow \langof{\autof p} \subseteq 
\langof{\autof q}$.
Intuitively, our optimization allows to remove a~state~$p$ from
a~macrostate~$M$ if there is also a~state~$q$ in~$M$ such that
\begin{inparaenum}[(i)]
  \item  $p \simul q$,
  \item  $p$ is not reachable from $q$, and
  \item  $p$ is smaller than~$q$ in an arbitrary total order over~$Q$ (this
    serves as a~tie-breaker for simulation-equivalent mutually unreachable
    states).
\end{inparaenum}
The reason why~$p$ can be removed is that its behaviour can be completely
mimicked by~$q$.
In our construction, we can then, roughly speaking, replace each 
call to the functions~$\trans(U,a)$ and $\acctrans(U,a)$, for a set of
states~$U$, by $\pr(\trans(U,a))$ and $\pr(\acctrans(U,a))$ respectively in
each partial complementation algorithm,
% $\dac$, $\iwc$, and 
% $\idac$ (and later $\nacrnk$ defined in \cref{sec:nac-rank-based})
as well as in the 
top-level algorithm, where~$\pr(S)$ is obtained from~$S$ by pruning all
eligible states.
The details are provided in \cref{sec:sim-details}.

\section{Modular Complementation of Non-Elevator Automata}\label{sec:general}
% \vspace{-2.0mm}
%%%%%%%%%%%%%%%%%%%%%%%%%%%%%%%%%%%%%%%%%%%%%%%%%%%%%%%%%%%%%%%%%%%%%%%%%%%%%%%%

A non-elevator automaton $\aut$ contains at least one NAC, besides 
possibly other IWCs or DACs.
To complement $\aut$ in a modular way, we apply the techniques seen in \cref{sec:modular-elevator} to its DACs and IWCs, while for 
its NACs we resort to a general complementation algorithm $\alg$.
In theory, rank-~\cite{KupfermanV01}, slice-~\cite{KahlerW08},  
Ramsey-~\cite{SistlaVW87}, subset-tuple-~\cite{AllredU18}, and determinization-~\cite{Safra88} based 
complementation algorithms adapted to work on a single partition block instead of the whole 
automaton are all valid instantiations of $\alg$. Below, we give a~high-level description of
\mbox{two such algorithms: rank- and determinization-based.}

\paragraph{Rank-based partial complementation algorithm.}
Working on each NAC independently benefits the complementation algorithm even
if the input BA contains only NACs.
For instance, in rank-based
algorithms~\cite{KupfermanV01,FriedgutKV06,Schewe09,KarmarkarC09,ChenHL19,HavlenaL21,HavlenaLS22a},
the fact whether all runs of~$\aut$ over a given $\omega$-word~$w$ are
non-accepting is determined by \emph{ranks} of states, given by the so-called
\emph{ranking functions}.
A~ranking function is a~(partial) function from~$Q$ to~$\omega$.
The main idea of rank-based algorithms is the following:
\begin{inparaenum}[(i)]
  \item  every run is initially nondeterministically assigned a~rank,
  \item  ranks can only decrease along a~run,
  \item  ranks need to be even every time a~run visits an accepting transition, and
  \item  the complement automaton accepts iff all runs eventually get trapped
    in odd ranks\footnote{Since we focus on intuition here, we use
    runs rather than the directed acyclic graphs of runs.}.
\end{inparaenum}
In the standard rank-based procedure, the initial assignment of ranks to states
in~(i) is a~function $Q \partialto \{0, \ldots, 2n-1\}$ for $n = |Q|$.
Using our framework, we can, however, significantly restrict the considered
ranks in a~partition block~$P$ to only $P \partialto \{0, \ldots, 2m-1\}$ for
$m = |P|$ (here, it makes sense to use partition blocks consisting of single
SCCs).
One can further reduce the considered ranks using the techniques
introduced in, e.g., \cite{HavlenaL21,HavlenaLS22a}.

In order to adapt the rank-based construction as a~partial complementation
algorithm $\nacrnk$ in our framework, we need to 
extend the ranking functions by a~fresh ``box state''~$\extension$ representing states 
outside the partition block. The ranking function then uses~$\extension$ to represent ranks 
of runs newly coming into the partition block.
The box-extension also requires to change the 
transition in a~way that~$\extension$ always represents 
reachable states from the outside. We provide the details of the construction,
which includes the MaxRank optimization from~\cite{HavlenaL21}, in
\cref{sec:nac-rank-based}. 

\paragraph{Determinization-based partial complementation algorithm.}
In~\cite{TsaiFVT14,HavlenaLS22a} we can see that determinization-based
complementation is also a good instantiation of $\alg$ in practice,
so, we also consider the standard Safra-Piterman
determinization~\cite{Safra88,Piterman07,Redziejowski12} as a~choice of~$\alg$
for complementing NACs.
Determinization-based algorithms use a~layered subset construction to organize
all runs over an $\omega$-word~$w$.
The idea is to identify a~subset $S \subseteq H$ of reachable states that occur
infinitely often along reading~$w$ such that between every two occurrences
of~$S$, we have that
(i)~every state in the second occurrence of~$S$ can be reached by a~state in
the first occurrence of~$S$ and
(ii)~every state in the second occurrence is reached by a~state in the first occurrence while seeing an accepting transition.
According to K\"onig's lemma, there must then be an accepting run of $\aut$ over $w$.

The construction initially maintains only one set $H$: the set of reachable
states.
Since~$S$ as defined does not necessarily need to be~$H$, every time there are
runs visiting accepting transitions, we create a~new subset~$C$ for those runs
and remember which subset~$C$ is coming from.
This way, we actually organize the current states of all runs into a~tree
structure and do subset construction in parallel for the sets in each
tree node.
If we find a~tree node whose labelled subset, say~$S'$, is equal to the union of
states in its children, we know the set~$S'$ satisfies the condition above and
we remove all its child nodes and emit a~good event.
If such good event happens infinitely often, it means that $S'$~also occurs infinitely often.
So in complementation, we only need to make sure those good events only happen for finitely many times.
Working on each NAC separately also benefits the determinization-based approach
since the number of possible trees will be less with smaller number of
reachable states.
Following the idea of~\cite{LiTFVZ22}, to adapt for the construction as the
partial complementation algorithm, we put all the newly coming runs from other
partition blocks in a~newly created node without a~parent node.
In this way, we actually maintain a~forest of trees for the partial complementation construction.
We denote the determinization-based construction as
$\detalg$; cf.\ \cite{LiTFVZ22} for details.

\newcommand{\tableResults}[0]{
% \begin{table}[t]
%   \caption{Overview of the tools' outcome on \automatabenchmarks}
%   \setlength{\tabcolsep}{0.5em}
%   \label{tab:overview}
%   \centering
%   \begin{tabular}{c|c|c|rrrr}
%     Outcome & \vbs[+/-] & \kofola[S/P] & \multicolumn{1}{c}{\cola} & \multicolumn{1}{c}{\ranker} & \multicolumn{1}{c}{\seminator} & \multicolumn{1}{c}{\spot} \\
%     \hline
%     Complemented & 39834 & 39738/39750 & 39814 & 38837 & 39026 & 39827 \\
%     Timeout & --na-- & 89/76 & 21 & 61 & 238 & 8 \\
%     Out of memory & --na-- & 10/11 & 0 & 939 & 573 & 0 \\
%     Other failures & --na-- & 0/0 & 2 & 0 & 0 & 2 \\
%     \hline
%     Runtime: mean -- median & 0.05 -- 0.01 & 0.32/0.41 -- 0.03 & 0.17 -- 0.02 & 3.34 -- 0.01 & 1.98 -- 0.03 & 0.08 -- 0.02 \\
%     States: mean -- median & 78/96 -- 3 & 76/86 -- 3 & 80 -- 3 & 45 -- 4 & 247 -- 3 & 160 -- 4 \\
%   \end{tabular}
% \end{table}
\begin{table}[t]
  \caption{
    Statistics for our experiments.
    The column \textbf{unsolved} classifies unsolved instances by the form \emph{timeouts\,:\,out of memory\,:\,other failures}.
    For the cases of \vbs we provide just the number of unsolved cases. 
    The columns \textbf{states} and \textbf{runtime} provide \emph{mean\,:\,median} of the number of states and runtime, respectively.
    } 
  \newcolumntype{g}{>{\columncolor{Gray!30}}r}
  \newcolumntype{f}{>{\columncolor{Gray!30}}l}
  \newcolumntype{h}{>{\columncolor{Gray!30}}c}
  \label{tab:overview}
  \centering
  \scalebox{0.9}{
  \begin{tabular}{lgrcrcrghgrcr}
    \toprule
    \textbf{tool} & \textbf{solved} & \multicolumn{5}{c}{\textbf{unsolved}}  & \multicolumn{3}{h}{\textbf{states}} & \multicolumn{3}{c}{\textbf{runtime}} \\
    \midrule
    \rowcolor{GreenYellow}\kofola[S] & 39,738 & 89 & \!:\! & 10 &\!:\! & 0 & 76 & \!:\! & 3 & 0.32 & \!:\! & 0.03 \\
    \rowcolor{GreenYellow}\kofola[P] & 39,750 & 76 & \!:\! & 11 & \!:\!& 0 & 86 & \!:\! & 3 & 0.41 & \!:\! & 0.03 \\
    \hline
    \vbs[+] & 39,834 & & & 3 & & & 78 & \!:\! & 3 & 0.05 & \!:\! &  0.01 \\
    \vbs[-] & 39,834 & & & 3 & & & 96 & \!:\! & 3 & 0.05 & \!:\! &  0.01 \\
    \bottomrule
  \end{tabular}}
  \hspace{3mm}
  \scalebox{0.9}{
  \begin{tabular}{lgrcrcrghgrcr}
    \toprule
    \textbf{tool} & \textbf{solved} & \multicolumn{5}{c}{\textbf{unsolved}}  & \multicolumn{3}{h}{\textbf{states}} & \multicolumn{3}{c}{\textbf{runtime}} \\
    \midrule
    \cola & 39,814 & 21 & \!:\! & 0 & \!:\! & 2 & 80 & \!:\! & 3 & 0.17 & \!:\! & 0.02 \\
    \ranker & 38,837 & 61 & \!:\! & 939 & \!:\! & 0 & 45 & \!:\! & 4 & 3.31 & \!:\! & 0.01 \\
    \seminator & 39,026 & 238 & \!:\! & 573 & \!:\! & 0 & 247 & \!:\! & 3 & 1.98 & \!:\! & 0.03 \\
    \spot & 39,827 & 8 & \!:\! & 0 & \!:\!& 2 & 160 & \!:\! & 4 & 0.08 & \!:\! & 0.02 \\
    \bottomrule
  \end{tabular}}
\end{table}
}

\newcommand{\figMainPlots}[0]{
\begin{figure}[t]
  \resizebox{\linewidth}{!}{
    \includegraphics{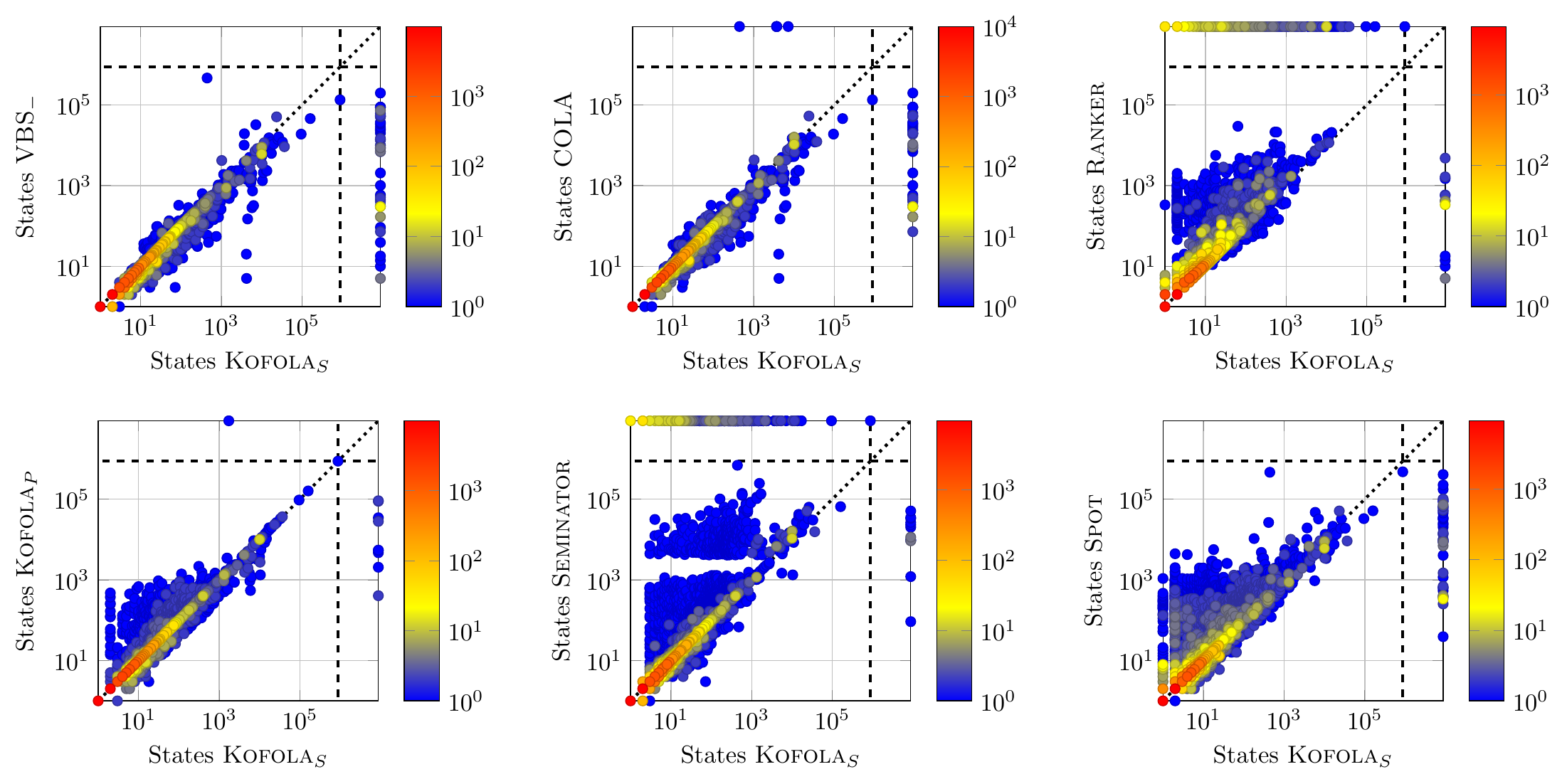}
  }
  \vspace*{-6mm}
  \caption{Scatter plots comparing the numbers of states generated by the tools.}
  \label{fig:experimentsStates}
\end{figure}
}

\newcommand{\figrankerelev}[0]{
\begin{wrapfigure}[10]{r}{5.3cm}
% \vspace*{-10mm}
\vspace*{-6mm}
\hspace*{-2mm}
\scalebox{0.75}{
  \includegraphics{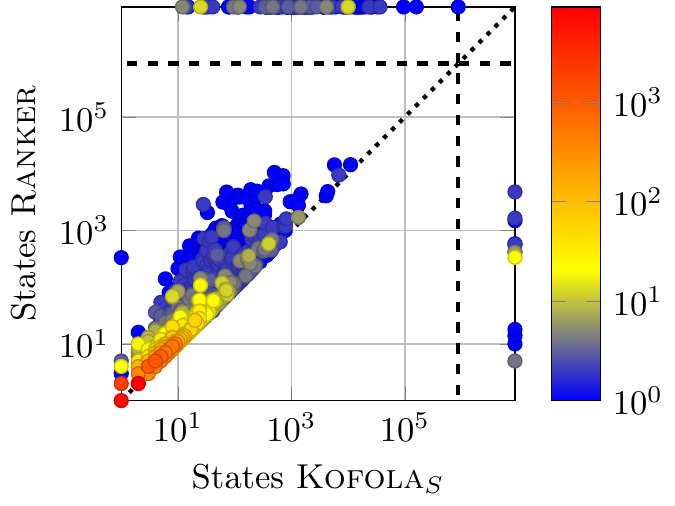}
}
\vspace{-4mm}
% \caption{Comparison for elevator automata}
\label{fig:ranker-elev}
\end{wrapfigure}%
}

%!!!!!!!!!!!!!!!!!!!!!!!!!!!!!!!!
% \enlargethispage{2mm}
%!!!!!!!!!!!!!!!!!!!!!!!!!!!!!!!!

%%%%%%%%%%%%%%%%%%%%%%%%%%%%%%%%%%%%%%%%%%%%%%%%%%%%%%%%%%%%%%%%%%%%%%%%%%%%%%%%
\vspace{-2.0mm}
\section{Experimental Evaluation}\label{sec:evaluation}
\vspace{-1.0mm}
%%%%%%%%%%%%%%%%%%%%%%%%%%%%%%%%%%%%%%%%%%%%%%%%%%%%%%%%%%%%%%%%%%%%%%%%%%%%%%%%

To evaluate the proposed approach, we implemented it in 
a~prototype tool \kofola (written in C++) built on top of
\spot~\cite{Duret-LutzRCRAS22} and compared it against 
\cola~\cite{LiTFVZ22}, 
\ranker~\cite{HavlenaLS22b} (v.~2), 
\seminator~\cite{BlahoudekDS20} (v.~2.0), 
and 
\spot~\cite{Duret-LutzLFMRX16,Duret-LutzRCRAS22} (v.~2.10.6), which are the
state of the art in BA
complementation~\cite{HavlenaLS22a,HavlenaLS22b,LiTFVZ22}.
Due to space restrictions, we give results for only two instantiations of our framework:
\kofolaids and \kofolaidsc.  Both instantiations use $\iwc$ for IWCs, $\dac$
for DACs, and $\detalg$ for NACs.
The partitioning selection algorithm merges all IWCs into one partition block,
all DACs into one partition block, and keeps all NACs separate.
Simulation-based pruning from \cref{sec:simulations} is turned on, and
round-robin from \cref{sec:round-robin} is turned off (since the selected
algorithms are quite deterministic).
\kofolaids employs the \emph{synchronous} and \kofolaidsc employs the
\emph{postponed} strategy.
We also consider the Virtual Best Solver (\vbs), i.e., a~virtual tool that
would choose the best solver for each single benchmark among all tools
(\vbsall) and among all tools except both versions of \kofola (\vbsnokofola).
We ran our experiments on an Ubuntu 20.04.4 LTS system running on a desktop machine 
with 16\,GiB RAM and an Intel 3.6\,GHz i7-4790 CPU.
To constrain and collect statistics about the executions of the tools, we used 
\benchexec~\cite{BeyerLW19} and imposed a memory limit of 12\,GiB and a~timeout
of 10\,minutes;
we used \spot to cross-validate the equivalence of the automata generated by the 
different tools.

As our data set, we used 39,837 BAs from the \automatabenchmarks
repository~\cite{automatabenchmarks} (used before by, e.g.,
\cite{HavlenaLS22a,HavlenaLS22b,LiTFVZ22}), which contains BAs from the
following sources:
\begin{inparaenum}[(i)]
  \item  randomly generated BAs used in~\cite{TsaiFVT14} (21,876~BAs),
  \item  BAs obtained from LTL formulae from the literature and randomly
    generated LTL formulae~\cite{BlahoudekDS20} (3,442~BAs),
  \item  BAs obtained from \ultimizer~\cite{ChenHLLTTZ18} (915~BAs),
  \item  BAs obtained from the solver for first-order logic over Sturmian words
    \pecan~\cite{HieronymiMOS0S22} (13,216~BAs),
  \item  BAs obtained from an S1S solver~\cite{HavlenaLS21} (370~BAs), and
  \item  BAs from LTL to SDBA translation~\cite{SickertEJK16} (18~BAs).
\end{inparaenum}
From these BAs, 23,850 are deterministic, 6,147 are SDBAs (but not
deterministic), 4,105 are elevator (but not SDBAs), and 5,735 are the rest.

\tableResults   %%%%%%%%%%%%%%%%%%%%%%%%%%%

In \cref{tab:overview} we present an overview of the outcomes.
% Besides \cola, \ranker, \seminator, and \spot, we consider two flavors of \kofola: 
% \kofolaids and \kofolaidsc, implementing the synchronous (cf.\@ \cref{sec:basic_synchronous}) and the postponed (cf.\@ \cref{sec:postponed}) approach, respectively.
% We also consider the Virtual Best Solver (\vbs), i.e., a virtual tool that would choose the best solver for each single benchmark among all tools (\vbsall) except \kofola (\vbsnokofola).
Despite being a prototype, \kofola is already able to complement a large
portion of the input automata, with very few cases that can be complemented
successfully only by \spot or \cola.
Regarding the mean number of states, \kofolaids has the \textbf{least mean value} from
all tools (except \ranker, which, however, had 1,000 unsolved cases)
Moreover, \kofola \textbf{significantly decreased the mean number of states} when
included into the \vbs: from 96 to 78! 
We consider this to be a~strong validation of the usefulness of our approach.
Regarding the running time, both versions of \kofola are rather similar;
\kofola is just slightly slower than \spot and \cola but much faster than both \ranker and \seminator
(the runtime plot is in the appendix).
Being a~prototype, there are many engineering opportunities for speed-up.

\figMainPlots    %%%%%%%%%%%%%%%%%

In \cref{fig:experimentsStates} we present a comparison of the number of states
generated by \kofolaids with those generated by the other tools;
we omit \vbsall since the corresponding plot can be derived from the one for
\vbsnokofola (since \ranker and \seminator only output BAs, we compare the
sizes of outputs transformed into BAs for all tools to be fair).
In the plots, the number of benchmarks represented by each mark is given by its color; 
a mark above the diagonal means that \kofolaids generated an automaton smaller
than the other tool while a mark on the top border means that the other tool
failed while \kofolaids succeeded, and symmetrically for the bottom part and
the right-hand border.
Dashed lines represent the maximum number of states generated by one of the
tools in the plot, axes are logarithmic.

\figrankerelev %%%%%%%%%%%%%%%%%%%%%%%%%%%
From the results, \kofolaids clearly dominates state-of-the-art tools that are not
based on SCC decomposition (\ranker, \spot, \seminator).
The outputs are quite comparable to \cola, which also uses SCC decomposition
and can be seen as an instantiation of our framework.
This supports our intuition that working on the single SCCs helps in reducing
the size of the final automaton, confirming the validity of our modular
mix-and-match \buchi complementation approach.
Lastly, in the figure in the right, we compare our algorithm for elevator
automata with the one in \ranker (the only other tool with
a~dedicated algorithm for this subclass).
Our new algorithm clearly dominates the one in \ranker.

% Consider the scatter plot for \kofolaids and \vbsnokofola:
% we have that is \kofolaids is better than (241 cases) or on par with (31492 cases) with the virtual best solver among \cola, \ranker, \seminator, and \spot, while on the remaining 8005 cases it generates more states.
% If we split \vbsnokofola in its components, we have that \kofolaids usually produces much fewer states than \ranker, \seminator, and \spot, the tools that implement the state of the art algorithms for complementing \buchi automata, all working on the whole automaton state space.
% On the other hand, \kofolaids is rather comparable with \cola, the only tool implementing an SCC-based algorithm. 
% This supports our intuition that working on the single SCCs helps in reducing the size of the final automaton, confirming the validity of our modular mix-and-match \buchi complementation approach.

%%%%%%%%%%%%%%%%%%%%%%%%%%%%%%%%%%%%%%%%%%%%%%%%%%%%%%%%%%%%%%%%%%%%%%%%%%%%%%%%
\vspace{-3.0mm}
\section{Related Work}\label{sec:related}
\vspace{-2.0mm}
%%%%%%%%%%%%%%%%%%%%%%%%%%%%%%%%%%%%%%%%%%%%%%%%%%%%%%%%%%%%%%%%%%%%%%%%%%%%%%%%
To the best of our knowledge, we provide the \emph{first general framework}
where one can plug-in different BA complementation algorithms while taking
advantage of the~specific structure of SCCs.
We will discuss the difference between our work and the literature. 

The breakpoint construction~\cite{MiyanoH84} was designed to complement BAs with
only IWCs, while our construction treats it as a~partial complementation
procedure for IWCs and differs in the need to handle incoming states from other
partition blocks.
The NCSB family of
algorithms~\cite{BlahoudekHSST16,ChenHLLTTZ18,BlahoudekDS20,HavlenaLS22b} for
SDBAs do not work when there are nondeterministic jumps between DACs;
they can, however, be adapted as partial procedures for complementing DACs in
our framework, cf.~\cref{sec:dac-complement}.
In~\cite{HavlenaLS22a}, a~deelevation-based procedure is applied to elevator
automata to obtain BAs with a~fixed maximum rank of~3, for which a~rank-based
construction produces a~result of the size in~$\bigO(16^n)$.
In our work, we exploit the structure of the SCCs much more to obtain an
exponentially better upper bound of~$\bigO(4^n)$ (the same as for SDBAs).
% the global maximum rank for complementing elevator BAs is reduced significantly, resulting a new upper bound $\bigO(16^n)$, while our framework improves the bound exponentially to $\bigO(4^n)$ by working on partition blocks separately.
The upper bound~$\bigO(4^n)$ for complementing unambiguous BAs was established
in \cite{LiVZ20}, which is orthogonal to our work, but seems to be possible to
incorporate into our framework in the future.
% We, however, consider incorporating a~specialized algorithm for unambiguous
% NACs in our framework as future work.
 
There is a~huge body of work on complementation of general BAs~\cite{Buchi90,SistlaVW87,BreuersLO12,KupfermanV01,FriedgutKV06,GurumurthyKSV03,ChenHL19,HavlenaL21,HavlenaLS22a,Schewe09,AllredU18,Safra88,Piterman07,Redziejowski12,BlahoudekDS20,TsaiFVT14,KahlerW08,VardiW08,FogartyKVW15,FogartyKWV13};
all of them work on the whole graph structure of the input BAs.
Our framework is general enough to allow including all of them as partial
complementation procedures for NACs.
On the contrary, our framework does not directly allow (at least in the
synchronous strategy) to use algorithms that \emph{do not} work on the
structure of the input BA, such as the learning-based complementation algorithm
from~\cite{LiTZS18}.
The recent determinization algorithm from~\cite{LiTFVZ22}, which serves as our
inspiration, also handles SCCs separately (it can actually be seen as an
instantiation of our framework).
Our current algorithm is, however, more flexible, allowing to mix-and-match
various constructions, keep SCCs separate or merge them into partition blocks,
and allows to obtain the complexity~$\bigO(4^n)$, while~\cite{LiTFVZ22} only
allowed~$\bigO(n!)$ (which is tight since SDBA determinization is
in~$\Omega(n!)$~\cite{EsparzaKRS17,Loding99}).
% In fact, our framework is inspired from \cite{LiTFVZ22}.
% The difference is that our framework is more flexible in the sense that all DACs/DACs can be handled together or separately depending on the given partition block, while \cite{LiTFVZ22} only consider them separately.
% Moreover, for complementing elevator BAs, the worst case complexity is $\bigO(n!)$ in \cite{LiTFVZ22}, while ours is $\bigO(4^n)$ and thus exponentially better.

Regarding the tool \spot~\cite{Duret-LutzLFMRX16,Duret-LutzRCRAS22}, it should
not be perceived as a single complementation algorithm.
Instead, \spot should be seen as a highly engineered platform utilizing
breakpoint construction for inherently weak BAs,
NCSB~\cite{BlahoudekHSST16,ChenHLLTTZ18} for SDBAs, and determinization-based
complementation~\cite{Safra88,Piterman07,Redziejowski12} for general BAs, while
using many other heuristics along the way.
\seminator uses
semi-determinization~\cite{CourcoubetisY88,BlahoudekDKKS17,BlahoudekDS20} to
make sure the input is an SDBA and then uses
NCSB~\cite{BlahoudekHSST16,ChenHLLTTZ18} to compute the complement. 

\iffalse
%NCSB \cite{BlahoudekHSST16}

%TACAS'22: Elevators \cite{HavlenaLS22a}

CAV'22: Divide and Conquer \cite{LiTFVZ22}

%CAV'22: Ranker \cite{HavlenaLS22b}

%PLDI'18: improved NCSB \cite{ChenHLLTTZ18}

CONCUR'21: reducing ranks \cite{HavlenaL21}

GandALF'20: complementation of unambiguous \cite{LiVZ20}

KV \cite{KupfermanV01}

FKV \cite{FriedgutKV06}

FKVS \cite{Schewe09}

VMCAI'18: Learning to complement \cite{LiTZS18}

Safra \cite{Safra88}

Piterman \cite{Piterman07}

Redziejowski \cite{Redziejowski12}

Spot-CAV'22 \cite{Duret-LutzRCRAS22}

Spot-ATVA'16 \cite{Duret-LutzLFMRX16}

Practical ACD \cite{CasaresDMRS22}

ACD-basic \cite{CasaresCF21}

Seminator~2 \cite{BlahoudekDS20}
\fi

%!!!!!!!!!!!!!!!!!!!!!!!!!!!!!!!!
% \enlargethispage{2mm}
%!!!!!!!!!!!!!!!!!!!!!!!!!!!!!!!!

%%%%%%%%%%%%%%%%%%%%%%%%%%%%%%%%%%%%%%%%%%%%%%%%%%%%%%%%%%%%%%%%%%%%%%%%%%%%%%%%
\vspace{-3.0mm}
\section{Conclusion and Future Work}\label{sec:conclusion}
\vspace{-2.0mm}
%%%%%%%%%%%%%%%%%%%%%%%%%%%%%%%%%%%%%%%%%%%%%%%%%%%%%%%%%%%%%%%%%%%%%%%%%%%%%%%%

We have proposed a general framework for BA complementation where one can
plug-in different partial complementation procedures for SCCs by taking
advantage of their specific structure.
Our framework not only obtains exponentially better upper bound for elevator
automata, but also complements existing approaches well.
As shown by the experimental results (especially for the \vbs), our framework
significantly improves the current portfolio of complementation algorithms.
% Moreover, due to its generality, our framework can serve an ideal testbed for
% experimenting with different BA complementation algorithms.
% In particular, we would like to investigate the opportunaties for its application to on-the-fly BA language inclusion testing.
% Moreover, we may leverage machine learning based approach to find out what combinations of the partial complementations and optimizations that works the best on a given set of benchmarks.
% Lastly, we want to incorporate partial complementation procedure for unambigous NACs based on \cite{LiVZ20} in our framework.

We believe that our framework is an ideal testbed for experimenting with
different BA complementation algorithms, e.g., for the following two reasons:
\begin{inparaenum}[(i)]
  \item  One can develop an efficient complementation algorithm that only works
    for a~quite restricted sub-class of BAs (such as the algorithm for initial
    deterministic SCCs that we showed in \cref{sec:det-init}) and the framework
    can leverage it for complementation of all BAs that contain such
    a~sub-structure.
  \item When one tries to improve a~general complementation algorithm, they can
    focus on complementation of the structurally hard SCCs (mainly the
    nondeterministic accepting SCCs) and do not need to look for heuristics
    that would improve the algorithm if there were some easier substructure
    present in the input BA (as was done, e.g., in~\cite{HavlenaLS22a}).
\end{inparaenum}
From how the framework is defined, it immediately offers opportunities for being
used for on-the-fly BA \emph{language inclusion} testing, leveraging the
partial complementation procedures present.
Finally, we believe that the framework also enables new directions for future
research by developing smart ways, probably based on machine learning, of
selecting which partial complementation
procedure should be used for which SCC, based on their features.
In future, we want to incorporate other algorithms for complementation of NACs,
and identify properties of SCCs that allow to use more efficient algorithms
(such as unambiguous NACs~\cite{LiVZ20}).
Moreover, it seems that generalizing the \delayed optimization
from~\cite{HavlenaL21} on the top-level algorithm could also help reduce the
state space.

\paragraph{Acknowledgements.}
We thank the anonymous reviewers for their useful remarks that helped us improve
the quality of the paper and Alexandre Duret-Lutz for sharing a~Ti\textit{k}Z package for beautiful automata.
This work was supported by 
the Strategic Priority Research Program of the Chinese Academy of Sciences (grant no.\@ XDA0320000);
the National Natural Science Foundation of China (grants no.\@ 62102407 and 61836005);
the CAS Project for Young Scientists in Basic Research (grant no.\@ YSBR-040);
the Engineering and Physical Sciences Research Council (grant no.\ EP/X021513/1);
the Czech Ministry of Education, Youth and Sports project LL1908 of the ERC.CZ programme;
the Czech Science Foundation project GA23-07565S; and
the FIT BUT internal project FIT-S-23-8151.
\protect\includegraphics[height=8pt]{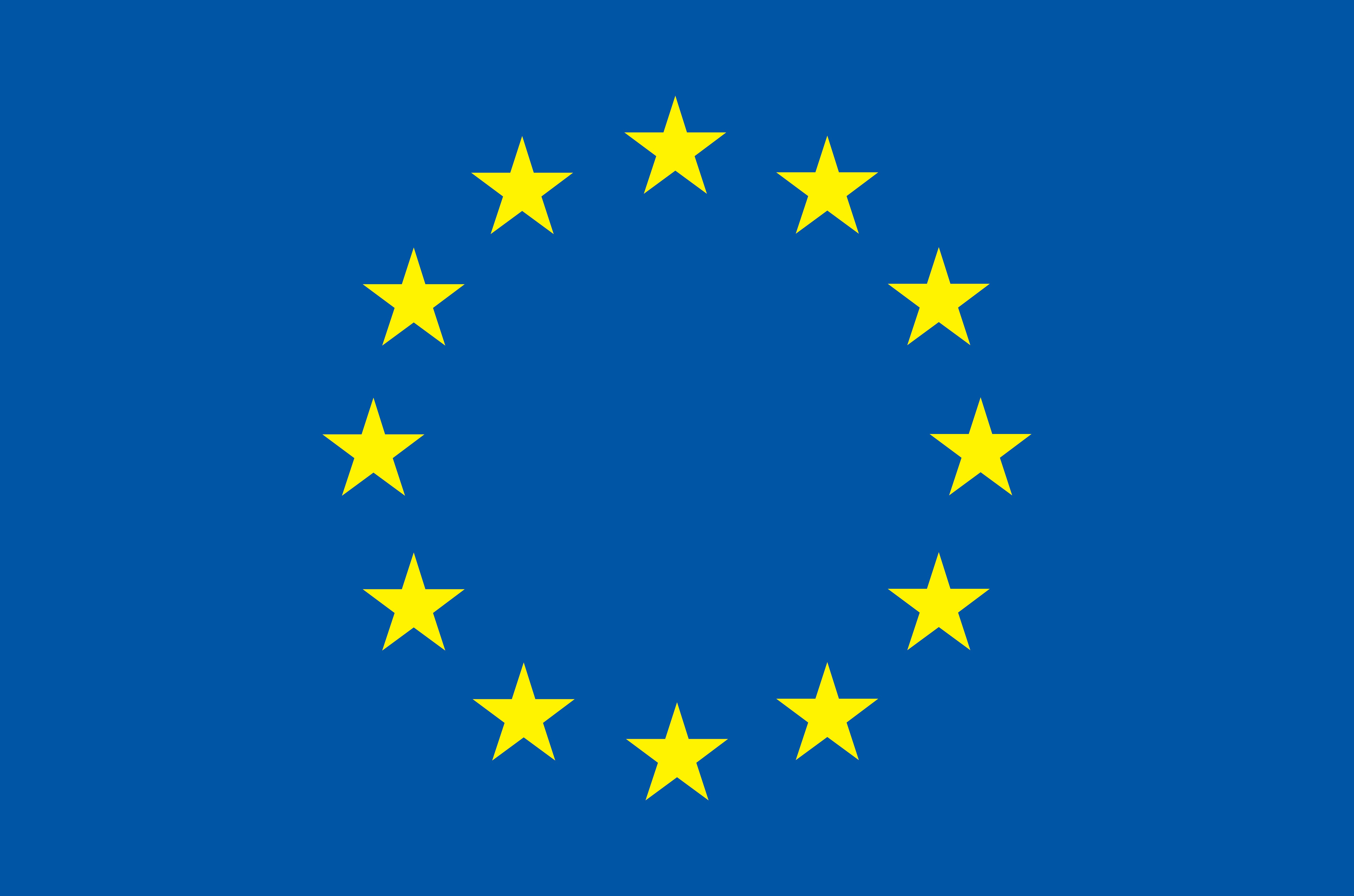} This project has received funding from the European Union’s Horizon 2020 research and innovation programme under the 
% Marie Sk\l{}odowska-Curie 
Marie Sklodowska-Curie 
grant agreement no.\@ 101008233.

\ol{artifact?}

%-----------------------------
\bibliographystyle{splncs04}
\bibliography{literature.bib}
%-----------------------------

\appendix

\clearpage

%*******************************************************************************
\vspace{-0.0mm}
\section{Round-Robin Algorithm}\label{sec:round_robin_details}
\vspace{-0.0mm}
%*******************************************************************************

In this section, we provide details related to the round-robin algorithm. 
The round-robin algorithm works on the level of \emph{partial complementation round-robin algorithms} 
$\algrr$. We require an instance of the partial round-robin 
algorithm $\algrr_P$ to provide the following:

\begin{itemize}
  \item  $\onetypeof{\algrr_P} = \tracktypeof{\algrr_P} \cup \activetypeof{\algrr_P}$ --- the type of the macrostates produced by the
    algorithm consisting of the passive type and the active type; %\vh{disjoint union might be important}

  \item  $\colourtypeof{\algrr_P}$ --- the set of colours the algorithm produces;

  \item  $\GetInitOf{\algrr_P}\colon 2^{\tracktypeof{\algrr_P}}$ --- a function returning
    the set of initial macrostates;

  \item $\GetSuccAOf{\algrr_P}\colon (2^\states \times \activetypeof{\algrr_P} \times \alphabet)
  \to 2^{\onetypeof{\algrr_P} \times \colourtypeof{\algrr_P}}$
  --- transition function returning successors in the active phase; 

  \item $\GetSuccTOf{\algrr_P}\colon (2^\states \times \tracktypeof{\algrr_P} \times \alphabet)
  \to 2^{\tracktypeof{\algrr_P} \times \colourtypeof{\algrr_P}}$ --- function performing lift from a passive state to a set of active states;

  \item $\LiftOf{\algrr_P}\colon  \tracktypeof{\algrr_P} 
  \to 2^{\activetypeof{\algrr_P}}$ --- transition function for switching from passive to active phase;

  \item  $\GetAccOf{\algrr_P}\colon \emersonleiof{\colourtypeof {\algrr_P}}$ ---
    a~function that returns a~formula for the acceptance condition.
\end{itemize}

\SetKwProg{Fn}{Function}{:}{}
\begin{function}[t]
  \Fn{$\Delta^{\cut}((H, M_1, \dots, M_n, \ell), a)$}{
    $\calM := \emptyset$\;
    Let $\algrr_i$ be $\algrr^i_{P_i}$\;
    $\fsucc_i := \GetSuccTOf{\algrr_i}$ for each $1\leq i \leq n$\;
    $\fsucc_\ell := \GetSuccAOf{\algrr_\ell}$\;
    % \lIf{$M_\ell \in \tracktypeof{\algrr^i_{C_i}}$} {
    %   $\fsucc_\ell := \GetSuccTAOf{\algrr^i_{C_i}}$
    % }
    % \lElse{
    %   $\fsucc_\ell := \GetSuccAOf{\algrr^i_{C_i}}$
    % }
    %$\fsucc_\ell := \GetSuccAOf{\algrr^i_{C_i}}$\;
    \For{$( (M'_1, c_1), \ldots, (M'_n, c_n)) \in [\fsucc_i(H, M_i, a)]_{i=1}^n$} {
      $\ell' := \ell$\;
      $L_{i} := \{ M'_i \}$ for each $1\leq i \leq n$\;
      \If{$M'_\ell \in \tracktypeof{\algrr_\ell}$} {
        $\ell' := (\ell \mod n) + 1$\;
        $L_{\ell'} := \LiftOf{\algrr_{\ell'}}(M'_{\ell'})$\;
      } 
      $\calM := \calM \cup \{ (\delta(H,a), (M''_1, c_1), \ldots, (M''_n, c_n), \ell') \mid M''_i \in L_i, \forall 1 \leq i \leq n \}$\;
    }
    \Return $\calM$\;
}
\end{function}

Let $\aut = (\states, \trans, \inits, \acc)$ be a BA,
$P_1, \ldots, P_n$ be a partitioning, and $\algrr^1, \ldots, \algrr^n$ be
a~sequence of algorithms such that~$\algrr^i$ is a~partial round-robin complementation
algorithm for~$P_i$. 
The complementation algorithm then produces the \\TELA~$\modcomplrr(\algrr^1_{P_1}, \dots, \algrr^n_{P_n},\aut) = (\states^{\cut},
\trans^{\cut}, \inits^{\cut}, \colourset^{\cut}, \colouring^{\cut},
\acccond^{\cut})$ whose components are defined as follows:
\begin{itemize}
  \item  $\states^{\cut} = 2^\states \times [\onetypeof{\algrr^i_{P_i}}]_{i=1}^n \times \{ 1, \dots, n \}$,\\[-3mm]
  \item  $\inits^{\cut} = \{\inits\} \times \Lift\big(\GetInitOf{\algrr^1_{P_1}}\big) \times [\GetInitOf{\algrr^i_{P_i}}]_{i=2}^n \times \{ 1 \}$,
  \item  $\colourset^{\cut} =  \{0, \ldots, \renumof{k^{\algrr^n_{P_n}}-1}{n}\}$,
  \item  $\acccond^{\cut} = \bigwedge_{i=1}^n \renumof{\GetAccOf{\algrr^i_{P_i}}} i$, and
  \item  $\trans^{\cut}$ and $\colouring^{\cut}$ are defined such that
    if
    $$(H', (M'_1, c_1), \ldots, (M'_n, c_n), \ell') \in \Delta^{\cut}((H, M_1, \dots, M_n, \ell), a),$$
    then $\trans^{\cut}$ contains the transition
    $t\colon (H, M_1, \ldots, M_n, \ell) \ltr a (H', M'_1, \ldots, M'_n, \ell')$,
    whose colouring is set to
    $\colouring^{\cut}(t) = \bigcup\{c_i \mid 1 \leq i \leq n\}$,
    and~$\trans^{\cut}$ is the smallest such a~set.
    We also set $\colouring^{\cut}(q) = \emptyset$ for every $q \in
    \states^{\cut}$.
\end{itemize}

% To formulate the correctness condition for the round-robin algorithm, we need to specify 
% the scheduler resolving when to switch from the passive phase to the active phase.
% The scheduler $\sched$ is a mapping $\sched: \omega \to \omega$ such that $\sched(\ell) > \ell$ 
% for each $\ell\in\omega$. In our context, if $\ell$ is the position of switching from the 
% active to the track phase, then $\sched(\ell)$ determines 
% the position when to switch back to the active phase. 

In the following, we focus on the correctness condition of our round-robin algorithm.
For a run $\rho$, we use $\rho_{i:j}$ where $i\leq j$ to denote the sequence $\rho_i,\rho_{i+1},\dots,\rho_j$.
Let $\algrr_P$ be an instance of partial round-robin complementation algorithm.
We use $\algunion(\algrr_P)$ to denote the partial complementation algorithm having the same type, 
the set of colors, the set of initial macrostates, and the acceptance condition as $\algrr_P$.
The successor function is given as
\[
\GetSuccOf{\algunion(\algrr_P)}(H,M,a) = 
\begin{cases}
  \GetSuccAOf{\algrr_P}(H,M,a) & \text{if $M \in \activetypeof{\algrr_P}$,} \\
  S \cup \bigcup_{(M',c) \in S}\LiftOf{\algrr_P}(H,M')  & \text{if $M \in \tracktypeof{\algrr_P}$}
\end{cases}
\]
where $S = \GetSuccTOf{\algrr_P}(H,M,a)$.

% Let $\sched$ be a scheduler and let $\rho$ be an accepting run on a word $\word$ in the automaton $\modcompl(\algunion(\algrr_P), \aut)$. 
% We say that $\rho$ is \emph{fair} wrt $\sched$ iff 
% \begin{inparaenum}[(i)]
%   \item $\rho_{j} \in \activetypeof{\algrr_P}$ and $\rho_{j+1} \notin \activetypeof{\algrr_P}$ for infinitely many $j$s, and
%   \item if $\rho_i \in \activetypeof{\algrr_P}$ and $\rho_{i+1} \notin \activetypeof{\algrr_P}$,
%   then $\rho_j \in \tracktypeof{\algrr_P}$ for each $i < j\leq \sched(i+1)$ and $\rho_{\sched(i+1)+1} \in \activetypeof{\algrr_P}$.
% \end{inparaenum}
%
Moreover, for $\aut$ such that $\aut \modelssurrp \phi_\algrr$ we say that $\algrr_P$ is \emph{consistent} if in the automaton $\modcompl(\algunion(\algrr_P), \aut)$ the following holds:
\begin{itemize}
  \item[(C1)] for each $\word \in \langof{\modcompl(\algunion(\algrr_P), \aut)}$ there is 
  an accepting run $\rho$ on $\word$ such that $\rho_{j} \in \activetypeof{\algrr_P}$ and $\rho_{j+1} \notin \activetypeof{\algrr_P}$ for infinitely many $j$s;
  % \item for each accepting run $\rho$ and each $i\in\omega$ such that $\rho(i)\in\activetypeof{\algrr_P}$, there is an accepting run $\rho'$ such that $\rho'(\{ 1, \dots, i-1 \}) = \rho(\{ 1, \dots, i-1 \})$ and 
  % $\rho'(i) \in\activetypeof{\algrr_P}$.
  \item[(C2)] for each accepting run $\rho$ 
  and each $i\in\omega$, we have that there are accepting runs $\rho'$ and $\rho''$ such that $i > 1 \Rightarrow\rho_{i-1} \in\tracktypeof{\algrr_P}$ and $\rho'_{1:i-1} = \rho''_{1:i-1} = \rho_{1:i-1}$
  and $\rho'_i \in\tracktypeof{\algrr_P}$ and $\rho''_i \in\activetypeof{\algrr_P}$;

  \item[(C3)] let $\rho$ be a run. If $\rho_{j} \in \activetypeof{\algrr_P}$ and 
  $\rho_{j+1} \notin \activetypeof{\algrr_P}$ for infinitely many $j$s, then $\rho$ is accepting.
\end{itemize}

Intuitively, the first condition ensures that for an accepted word, there is a run containing infinitely many switches between the passive and active type.
The second condition then expresses that the switch to the active phase can be postponed by a finite number 
of steps but still preserving the acceptance. The last one expresses that if a run encounters 
infinitely many switches between the active and passive, this run is accepting.
The correctness condition on the partial round-robin algorithm is then given as follows:

% if for each its accepting run $\rho$ and schedulers $\sched, \sched'$
% such that $\rho$ is fair wrt $\sched$ and $\sched(\{ 1, \dots, n \}) = \sched(\{ 1, \dots, n \})$ for some $n\in\omega$, it holds that
% there is an accepting run $\rho'$ such that $\rho'(\{ 1, \dots, n \}) = \rho'(\{ 1, \dots, n \})$ and $\rho'$ is fair wrt $\sched'$.

\begin{definition}
  We say that $\algrr$ is \emph{correct} if for each
  $\aut$ such that $\aut \modelssurrp \phi_\algrr$ 
  we have that $\algrr_P$ is consistent and $\langof{\modcompl(\algunion(\algrr_{P}), \aut)} = \Sigma^\omega\setminus\langof{\aut}$.
\end{definition}

\begin{restatable}{theorem}{thmRRCorrectness}
  \label{thm:rr-correctness}
  Let $\aut$ be a BA, $P_1, \ldots, P_n$ be a partitioning of $\aut$, and $\algrr^1, \ldots, \algrr^n$ be
  a~sequence of partial round-robin complementation algorithms such that~$\algrr^i$ is \emph{correct} for~$P_i$. 
  Then, $\langof{\modcomplrr(\algrr^1_{P_1}, \dots, \algrr^n_{P_n},\aut)} = \Sigma^\omega\setminus\langof{\aut}$.
\end{restatable}

\begin{proof}
  First, we propose the following auxiliary claim:
  
  \begin{claim}\label{claim-postponing}
    Let $\rho$ be an accepting run in $\modcompl(\algunion(\algrr_P), \aut_P)$ such that there are $i_1, i_2$, $i_1 \leq i_2$ and 
    $\forall i_1 \leq \ell \leq i_2$ we have $\rho_\ell\in \tracktypeof{\algrr_P}$ and $\rho_{i_2 + 1} \in \activetypeof{\algrr_P}$.
    Then, for each $k \geq i_1$ there is an accepting run $\rho'$ such that $\rho'_{k+1}\in \activetypeof{\algrr_P}$, $\forall i_1 \leq \ell \leq k$ we have $\rho'_\ell\in \tracktypeof{\algrr_P}$, and $\rho'_{1:\min(k,i_2)} = \rho_{1:\min(k,i_2)}$.
  \end{claim}
  \begin{claimproof}
    Follows directly from a multiple application of (C2).
  \end{claimproof}

  Since $P_1, \ldots, P_n$ is a partitioning of $\aut$,
  we have that 
  \begin{equation}\label{eq:part}
    \bigcup_{i=1}^n \langof{\aut_{P_i}} = \langof{\aut}.
  \end{equation}

  Now, we proceed to the proof of the theorem.
  
  $(\subseteq)$ Let $\varrho = (M^1_1, \dots, M^1_n, i_1)\dots (M^k_1, \dots, M^k_n, i_k)\dots$ be an 
  accepting run in $\modcomplrr(\algrr^1_{P_1}, \dots, \algrr^n_{P_n},\aut)$ on $\word$.
  From the construction, we have that $M^1_\ell \dots $ is an accepting run in $\modcompl(\algunion(\algrr^{\ell}_{P_{\ell}}), \aut_{P_\ell})$. 
  Therefore, $\word \in \bigcap_{\ell = 1}^n \langof{\modcompl(\algunion(\algrr^{\ell}_{P_{\ell}}), 
  \aut_{P_\ell})} = \bigcap_{\ell = 1}^n \Sigma^\omega\setminus \langof{\aut_{P_\ell}} = \Sigma^\omega \setminus \langof{\aut}$.
  The latter follows from the correctness condition on $\algrr_P$ and from~\eqref{eq:part}.

  $(\supseteq)$ Consider a word $\word \in \Sigma^\omega \setminus \langof{\aut}$.
  We show by induction that there is also an accepting run $\varrho = (M^1_1, \dots, M^1_n, i_1)\dots (M^k_1, \dots, M^k_n, i_k)\dots$ in the automaton $\modcomplrr(\algrr^1_{P_1}, \dots, \algrr^n_{P_n},\aut)$ 
  and moreover for each $j$, $M^1_\ell\dots M^j_\ell$ is a prefix of an accepting run of $\word$ in $\modcompl(\algunion(\algrr^{\ell}_{P_{\ell}}), \aut_{P_{\ell}})$ for 
  each $\ell$. In the following, when we use (accepting) run, we implicitly mean on $\word$.
  \begin{itemize}
    \item Base case: Since $\rho_1$ satisfies the condition (C2), there is an accepting run 
    $\rho'$ from (C2) such that $\rho'(1)\in\activetypeof{\algrr^1_{P_1}}$. Moreover, there are 
    also runs $\rho'_2, \dots, \rho'_n$ such that $\rho'_i(1)\in\tracktypeof{\algrr^1_{P_1}}$.
    We hence set $\varrho_1 = ((\rho'_1)_1, \dots, (\rho'_n)_1, 1)$.
    
    \item Inductive case: Let $(M^1_1, \dots, M^1_n, i_1)\dots (M^k_1, \dots, M^k_n, i_k)$ be a sequence of first $k$ macrostates of $\varrho$.
    We prove that there is also $(k+1)$-th macrostate of $\varrho$. Since $M^k_{i_k} \in\activetypeof{\algrr^{i_k}_{P_{i_k}}}$ and 
    moreover, $M^1_{i_k}\dots M^k_{i_k}$ is a prefix of an accepting run $\rho'$ in the automaton $\modcompl(\algunion(\algrr^{i_k}_{P_{i_k}}), \aut_{P_{i_k}})$.
    Now assume that $(\rho'_{i_k})_{k+1} \in \activetypeof{\algrr^{i_k}_{P_{i_k}}}$. Since $p_\ell = M^1_{\ell}\dots M^k_{\ell}$ is the prefix of 
    an accepting run in $\modcompl(\algunion(\algrr^{\ell}_{P_{\ell}}), \aut_{P_{\ell}})$ for each $\ell\neq i_k$ and on top of that 
    $M^k_{\ell} \in \tracktypeof{\algrr^{\ell}_{P_{\ell}}}$. Therefore, from (C2) there are accepting runs $\rho'_\ell$ from (C2) 
    satisfying that $(\rho'_\ell)_{k+1}$ is of the corresponding passive type. We hence set $((\rho'_1)_{k+1}, \dots, (\rho'_n)_{k+1}, i_k)$ as 
    the $(k+1)$-th macrostate of $\varrho$.
    Now assume that $(\rho'_{i_k})_{k+1} \notin \activetypeof{\algrr^{i_k}_{P_{i_n}}}$. Let $j = (i_k \mod n) + 1$.
    $p_j = M^1_{j}\dots M^k_{j}$ is the prefix of 
    an accepting run in $\modcompl(\algunion(\algrr^{j}_{P_{j}}), \aut_{P_{j}})$, from \cref{claim-postponing} we obtain that there is also 
    an accepting run $\rho'_j$ extending $p_j$ such that $(\rho'_j)_{k+1} \in \activetypeof{\algrr^{j}_{C_{j}}}$. Further, since 
    $p_\ell = M^1_{\ell}\dots M^k_{\ell}$ is the prefix of 
    an accepting run in $\modcompl(\algunion(\algrr^{\ell}_{P_{\ell}}), \aut_{P_{\ell}})$ for each $\ell\notin \{j, i_k\}$ and on top of that 
    $M^k_{\ell} \in \tracktypeof{\algrr^{\ell}_{P_{\ell}}}$. Therefore, from (C2) there are accepting runs $\rho'_\ell$ from (C2) 
    satisfying that $(\rho'_\ell)_{k+1}$ is of the corresponding passive type. We set $((\rho'_1)_{k+1}, \dots, (\rho'_n)_{k+1}, j)$ as 
    the $(k+1)$-th macrostate of $\varrho$.
  \end{itemize}
  It can be easily shown that $\varrho$ is a run in $\modcomplrr(\algrr^1_{P_1}, \dots, \algrr^n_{P_n},\aut)$. Since each component 
  contains infinitely many switches between the passive and the active phase, we have from (C3) that each partial run is 
  accepting in the corresponding $\modcompl(\algunion(\alg^i_{P_i}))$ and hence $\varrho$ is accepting as well.
  \qed
\end{proof}

% ---------------------------
\subsection{Complementation of Inherently Weak Components}
In this section, we define the algorithm $\algrr = \iwcrr$ with a condition $\phi_\iwcrr$ 
specifying the condition that a~partition $P$ is inherently weak and \emph{accepting}. 

We formalize the instance $\iwcrr_P$ as follows:

\begin{itemize}
  \item $\onetypeof {\iwcrr_P} = \tracktypeof{\iwcrr_P} \cup \activetypeof{\iwcrr_P}$ where $\tracktypeof{\iwcrr_P} = 2^P$ and $\activetypeof{\iwcrr_P} = 2^P \times 2^P$ ,
  \item  $\colourtypeof{\iwcrr_P} = \{\tcacc{1}{0}\}$ ,
  \item  $\GetInitOf{\iwcrr_P} = \{\inits \cap P\}$,
  \item $\GetSuccAOf{\iwcrr_P}(H, (C, B), a) = 
    \begin{cases}
        \{ (\trans(H, a) \cap P, \{ \tcacc{1}{0} \}) \} & \text{if } B' = \emptyset \text{ for } \\  & B' = \trans(B,a) \cap P \\ 
        \{ ((\trans(H, a) \cap P, B'), \emptyset) \} & \text{otherwise},
    \end{cases}$
  \item $\GetSuccTOf{\iwcrr_P}(H, C, a) = \{ (\trans(H, a) \cap P, \emptyset) \}$,
  \item $\LiftOf{\iwcrr_P}(C) = \{(C, C)\}$ ,
  \item $\GetAccOf{\iwc_P} = \Infof {\tcacc{1}{0}}$.
\end{itemize}

\begin{lemma}
  The partial round-robin algorithm $\iwcrr$ is correct.
\end{lemma}
\begin{proof} \emph{(Sketch)}
   Consider a~BA $\aut$ such that $\aut \models_P \phi_{\iwcrr}$ and a~word \\ $\word\in\langof{\modcompl(\algunion(\iwcrr_P), \aut)}$. An accepting run $\phi$ must contain infinitely many accepting transitions labeled with \tcacc{1}{0}. Such transitions go from the active to the passive state. We can therefore switch to the passive state after emptying $B'$ and then after at least one step switch back to the active state. That satisfies the condition (C1). 
   It is not important how many steps we make before switching back to the active state, but it has to be a finite number. That satisfies the condition (C2). The condition (C3) is also satisfied because we switch from active to passive phase only when $B'$ is empty, i.e., only when an accepting transition is taken. 
   \qed
\end{proof}

% ---

\subsection{Complementation of Deterministic Accepting Components}
In the following, we define the algorithm $\alg = \dacrr$ with a condition $\phi_\dacrr$ 
specifying the condition that a~partition $P$ is \emph{deterministic within the SCCs}.

We formalize the instance $\dacrr_P$ as below:

\begin{itemize}
  \item $\onetypeof{\dacrr_P} = \tracktypeof{\dacrr_P} \cup \activetypeof{\dacrr_P}$ where $\tracktypeof{\dacrr_P} = 2^P \times 2^P$ and $\activetypeof{\dacrr_P} = 2^P \times 2^P \times 2^P$,
  \item  $\colourtypeof {\dac_P} = \{\tcacc{0}{0}\}$,
  \item $\GetInitOf{\dacrr_P} = \{(\inits \cap P, \emptyset)\}$
  \item $\GetSuccAOf{\dacrr_P}(H, (C, S, B), a) = U$ such that 
    \begin{itemize}
      \item if $\transacc(S, a) \cap P \neq \emptyset$, then $U = \emptyset$,
      \item  otherwise $U$ contains the pair $(V, c)$ where
        \begin{itemize}
          \item $V = \begin{cases}
              (C', S') & \text{if } B^\star = \emptyset \text{ for } B^\star = \transscc(B, a) \\
              (C', S', B^\star) & \text{otherwise}  
          \end{cases}$
          \item $S' = \transscc(S, a) \cap P$, 
          \item $C' = (\trans(H, a) \cap P) \setminus S'$, 
          \item $c = \begin{cases}
              \{\tcacc{0}{0}\} & \text{if } B^\star = \emptyset \text{ and}\\
              \emptyset & \text{otherwise}.
          \end{cases}$
        \end{itemize}
        and if $\transacc(B, a) \cap \transscc(B, a) = \emptyset$, then~$U$ also contains the pair
        $((C'', S''), \{\tcacc{0}{0}\})$ where
        \begin{itemize}
          \item  $S'' = S' \cup B'$ and
          \item  $C'' = C' \setminus S''$.
        \end{itemize}
    \end{itemize}
  \item $\GetSuccTOf{\dacrr_P}(H, (C, S), a) = U$ such that 
    \begin{itemize}
      \item if $\transacc(S, a) \cap P \neq \emptyset$, then $U = \emptyset$,
      \item otherwise $U = \{ ((C', S'), \emptyset) \}$ where 
        \begin{itemize}
          \item $S' = \trans(S, a) \cap P$, and 
          \item $C' = (\trans(C, a) \cap P) \setminus S'$,
        \end{itemize}
    \end{itemize}
    
  \item $\LiftOf{\dacrr_P}(C, S) = \{(C, S, C)\}$  
  \item $\GetAccOf{\dacrr_P} = \Infof{\tcacc{0}{0}}$.
  \end{itemize}

  \begin{lemma}
    The partial round-robin algorithm $\dacrr$ is correct.
  \end{lemma}
  \begin{proof} \emph{(Sketch)}
     Consider a~BA $\aut$ such that $\aut \models_P \phi_{\dacrr}$ and a~word \\ $\word\in\langof{\modcompl(\algunion(\dacrr_P), \aut)}$. An accepting run $\phi$ must contain infinitely many accepting transitions labeled with \tcacc{0}{0}. Such transitions go from the active to the passive state. We can therefore switch to the passive state after emptying $B^\star$ and then after at least one step switch back to the active state. That satisfies the condition (C1). 
     It is not important how many steps we make before switching back to the active state, but it has to be a finite number. That satisfies the condition (C2). The condition (C3) is also satisfied because we switch from active to passive phase only when $B^\star$ is empty, i.e., only when an accepting transition is taken. 
     \qed
  \end{proof}

%%%%%%%%%%%%%%%%%%%%%%%%%%%%%%%%%%%%%%%%%%%%%%%%%%%%%%%%%%%%%%%%%%%%%%%%%%%%%%%%%%%%%%%%%%%%%%%%%%%%%%%%%%%%%%%%%%%%%%%%%%%%%%%%%%%%%%%%%%%%%%%%%%%%%%%%%%
\section{Simulation-based Optimizations}
\label{sec:sim-details}
%%%%%%%%%%%%%%%%%%%%%%%%%%%%%%%%%%%%%%%%%%%%%%%%%%%%%%%%%%%%%%%%%%%%%%%%%%%%%%%%%%%%%%%%%%%%%%%%%%%%%%%%%%%%%%%%%%%%%%%%%%%%%%%%%%%%%%%%%%%%%%%%%%%%%%%%%%

In the following, we, as in the main text, fix a BA $\aut = (Q, \delta, I, F)$ with
the set of maximal SCCs $\{ \C_1, \dots, \C_n \}$. In this section, we generalize the results introduced in~\cite{HavlenaLS22b}. We use $p\leadsto q$ to
denote that $q$ is reachable from $p$. Let $\word\in\Sigma^\omega$ be a word.
Let $\Pi, \Pi'$ be sets of traces over $\word$. We say that $\Pi$ and $\Pi'$
are \emph{acc-equivalent}, denoted as $\Pi \sim \Pi'$ if $\exists\pi\in\Pi: \pi$
is accepting in $\aut$ iff $\exists\pi'\in\Pi': \pi'$ is accepting in $\aut$.

We recall here the definition of the tie-breaking function and the pruning 
relation from \cref{sec:simulations}. Let $\num \colon Q \to \omega$ be a function 
satisfying the following conditions for each $p,q \in Q$:
\begin{inparaenum}[(i)]
  \item $\num(p) = \num(q)$ iff $p,q \in\C_i$ and
  \item if $p\leadsto q$ then $\num(p) \leq \num(q)$.
\end{inparaenum}
Let ${\sqsubseteq}\subseteq Q\times Q$ be a~relation on the states of~$\aut$
defined as follows:
$p \sqsubseteq q$ iff
\begin{inparaenum}[(i)]
	\item $p \simul q$ and
	\item $\num(p) < \num(q)$
\end{inparaenum}
The pruning function $\pr\colon 2^Q \to 2^Q$ is then defined for
each $S \subseteq Q$ as $\pr(S) =S'$ where $S'\subseteq S$ is the smallest set such that
$\forall q \in S \exists q'\in S' \colon q\sqsubseteq q'$
% \begin{itemize}
% \vspace{-1mm}
%   \item  \emph{pruning}: $\pr(S) =S'$ where $S'\subseteq S$ is the
%       lexicographically smallest set (given a~fixed ordering on~$Q$) such that
%       $\forall q \in S \exists q'\in S' \colon q\sqsubseteq q'$ and
%   % \item  \emph{saturating}: $\sat(S) = \lfloor S \rfloor_{\langsimbya}$, where
%   %   $\lfloor S \rfloor_{\langsimbya} = \{ p\in Q \mid \exists q \in Q\colon p
%   %   \langsimbya q \}$.
% \end{itemize}
%
Informally, $\pr$ removes simulation-smaller states. 
% and $\sat$ saturates
% a~macrostate with all simulation-smaller states.

%
Let $\rho = S_1S_2\dots$ be a sequence of sets of states and $\word$ be a word.
We define $\Pi_\rho$ to be a set of traces over $\word$ matching the sets of
states. Formally, $\Pi_\rho = \{ \pi \mid \pi \text{ over } \word, \pi_i \in
S_i \text{ for each } i \}$. 
%We also define $\Pi_\rho^\cup =\bigcup_{i\in\omega}\Pi_{\rho_{i:\omega}}$. 
Further, for a set of states $B$ we
use $\rho_\word^B$ to denote the sequence $S_1S_2\dots$ such that $S_1 = B$,
$S_{i+1} = \delta(S_i, \word_i)$ for each $i\in\omega$. We use $\rho_\word$ to
denote $\rho_\word^I$. Moreover, for a given mapping $\theta \colon 2^Q \to 2^Q$ and
a sequence of sets of states $\rho_\word^B$ we define $\theta(\rho_\word^B) =
\theta(B)B_2 \dots$ where $B_{i+1} = \theta(\delta(B_i, \word_i))$ for each
$i\in\omega$. A trace $\pi$ is \emph{eventually simulated} by $\pi'$ if there is
some $i \in\omega$ such that $\pi_{i:\omega} \simul \pi'_{i:\omega}$.

% \begin{lemma}\label{lem:trace-union}
% 	Let $\alpha$ be a word, $\Pi_{\rho_\alpha}\sim \Pi_{\rho_\alpha'}$, and
% 	$\Pi_{\rho_\alpha}\subseteq \Pi_{\rho_\alpha'}$. Then, $\Pi_{\rho_\alpha}\sim
% 	\Pi_{\rho_\alpha'}^\cup$.
% \end{lemma}
% \begin{proof}
% 	Assume that $\Pi_{\rho_\alpha}\sim \Pi_{\rho_\alpha'}$, and
% 	$\Pi_{\rho_\alpha}\subseteq \Pi_{\rho_\alpha'}$. Since $\Pi_{\rho_\alpha'}
% 	\subseteq \Pi_{\rho_\alpha'}^\cup$, if there is an accepting trace in
% 	$\Pi_{\rho_\alpha}$, there is (the same) accepting trace in
% 	$\Pi_{\rho_\alpha'}^\cup$. If there is no accepting trace in
% 	$\Pi_{\rho_\alpha}$, it means that all traces contain infinitely many accepting
% 	states. Hence, every infinite suffix is also an accepting trace and therefore,
% 	$\Pi_{\rho_\alpha'}^\cup$ contain all traces that are not accepting (with
% 	infinitely many accepting states).
% \end{proof}

\begin{lemma}\label{lem:pr-trace}
	Let $\word$ be a word. Then, $\Pi_{\rho_\word}\sim \Pi_{\pr(\rho_\word)}$.
\end{lemma}
\begin{proof}
	First observe that $\Pi_{\pr(\rho_\word)} \subseteq \Pi_{\rho_\word}$.
	Therefore, it suffices to show that if there is an accepting trace $\pi \in
	\Pi_{\rho_\word}$, then there is also an accepting trace $\pi' \in
	\Pi_{\pr(\rho_\word)}$. We assume the former and we now show that there is $\pi'
	\in \Pi_{\pr(\rho_\word)}$ such that $\pi$ is eventually simulated by
	$\pi'$. 
%     If $\pi' = \pi$ we are done. 
    If $\pi \in \Pi_{\pr(\rho_\word)}$ we are done.
    Now, assume that this is not the case and
	that there is a maximum set of traces $P = \{ \pi^1, \pi^2, \dots \}\subseteq
	\Pi_{\rho_\word}$ with indices $\ell_1 < \ell_2 < \dots$ such that $p_i =
	\pi^i_{\ell_i} \sqsubseteq \pi^{i+1}_{\ell_i} = p_i'$ for each $i$, and
	moreover $\pi_1 = \pi$. 
  From the definition, we further have $\num(p_i) < \num(p_i')$. Since the
  numbers are finite and you cannot reach a state with lower number, the
  sequence eventually stabilizes and hence $P$ is finite. Since the set $P = \{
  \pi_1, \dots, \pi_n \}$ is maximum and finite, we have $\pi_n \in
  \Pi_{\pr(\rho_\word)}$. Moreover, $\pi' = \pi_n$ eventually $\simul$-simulates
  $\pi$ (given by the step-wise property of simulation), which concludes
  the proof.
  \qed
\end{proof}

% \begin{lemma}\label{lem:sat-trace}
% 	Let $\alpha$ be a word. Then, $\Pi_{\rho_\alpha}\sim \Pi^\cup_{\sat(\rho_\alpha)}$.
% \end{lemma}
% \begin{proof}
% 	First observe that $\Pi_{\rho_\alpha} \subseteq \Pi^\cup_{\sat(\rho_\alpha)}$.
% 	Therefore, it suffices to show that if there is an accepting trace $\pi \in
% 	\Pi^\cup_{\sat(\rho_\alpha)}$, there is also an accepting trace $\pi' \in
% 	\Pi_{\rho_\alpha}$. We fix $\rho = \rho_\alpha$. Consider some accepting trace
% 	$\pi \in \Pi^\cup_{\sat(\rho_\alpha)}$. If $\pi \in \Pi_{\rho_\alpha}$, we are
% 	done. If not, there is some position $\ell$ such that $\pi \in
% 	\Pi_{\rho_{\ell:\omega}}$ and $\pi_1 \langsimbya q$ where $q\in\rho_\ell$.
% 	Therefore, there is some trace $\pi' \in\rho$ such that $\pi'_\ell = q$. Moreover,
% 	$\pi$ is accepting, hence there is a trace $\pi''$ leading from $q$, which is
% 	accepting as well. Hence, $\pi'_{1:\ell}.\pi'' \in \rho$ and moreover this
% 	trace is accepting.
% \end{proof}

Consider a function $\theta \colon 2^Q \to 2^Q$ and
let the \emph{run DAG} of~$\aut$ over
a~word~$\word$ wrt.\@ $\theta$ be a~DAG (directed acyclic graph) $\dagw^\theta = (V,E)$ containing
vertices~$V$ and edges~$E$ such that
\begin{itemize}
  \setlength{\itemsep}{0mm}
  \item  $V \subseteq Q \times \omega$ such that $(q, i) \in V$ iff $q\in\tau_i$ where $\tau = \theta(\rho_\word^I)$,
  \item  $E \subseteq V \times V$ such that~$((q, i), (q',i')) \in E$ iff $i' = i+1$
    and $q' \in \trans(q, \wordof i)$.
\end{itemize}
We use $\dagw$ to denote $\dagw^\id$. We say that $\dagw^\theta$ is accepting if
there is a path in the graph encountering infinitely many times a vertex
corresponding to accepting state/transition. 

\begin{lemma}
  Let $\word$ be a word. Then, $\dagw^\pr$ is accepting iff $\dagw$ is accepting.
\end{lemma}
\begin{proof}
  Follows directly from Lemma~\ref{lem:pr-trace}.
  \qed
\end{proof}

All of the $\iwc$, $\dac$, $\idac$, and $\nacrnk$ complementation algorithms 
can be seen as
procedures taking a run DAG as an input and checking whether this graph is
accepting or not. Therefore, we can change the DAG in an ``arbitrary'' way, if we ensure 
that the modified DAG is accepting iff the original one is. Therefore, we 
can modify each of the algorithms to taking into account the pruned run DAG. This 
means that we can change each call of the functions $\delta(W,a)$ and $\acctrans(W,a)$ for 
$\pr(\delta(W,a))$ and $\pr(\acctrans(W,a))$ respectively. It is due to the fact, that 
each of the algorithms work on the level of run DAGs, which are created by these transition 
functions.

% All the Miyano-Hayashi, NCSB-based, and rank-based constructions can be seen as
% procedures taking a run DAG as an input and checking whether this graph is
% accepting or not. Therefore, we can change the DAG in an ``arbitrary'' way, if we ensure 
% that the modified DAG is accepting iff the original one is.

%*******************************************************************************
\vspace{-0.0mm}
\section{Rank-based Complementation of a NAC}\label{sec:nac-rank-based}
\vspace{-0.0mm}
%*******************************************************************************

In this section, we define the algorithm~$\nacrnk$ with the condition $\phi_\nacrnk$ 
specifying that a~partition block~$P$ is a~general 
nondeterministic accepting component.
Let~$P$ be a~partition of~$\aut$ such that $\aut_P \models \phi_\nacrnk$. 
Here, we consider an instantiation of the framework for
a~rank-based complementation procedure for $P$ (we use
a~modification of Schewe's optimal algorithms from~\cite{Schewe09}) and its
optimizations from~\cite{HavlenaL21,HavlenaLS22a}.
Let us start with some definitions.

% We assume a~BA $\aut = (Q, \delta, I, \delta_F)$ and its $j$-th
% SCC~$\C_j$ with the set of states~$Q_j$.

For $i \in \omega$ we use $\evenceil{i}$ to denote the largest even
number smaller or equal to~$i$, e.g., $\evenceil{42} = \evenceil{43} = 42$.
Let $f \colon X \to Y$ be a function. We use $\dom(f)$ to denote the domain of $f$. 
For a function $g$, we use $f\vartriangleleft g$ to denote the function such that for each $x \in X$
if $x \in \dom(g)$ it returns $g(x)$, otherwise $f(x)$.
Let $\qextof P = P \cup \{\extension\}$ where $\extension$ is a~fresh
symbol (which is used to represent ranks of runs outside~$P$; we pronounce
$\extension$ as ``\emph{box}'') and $\reachof q$ denote the set of states reachable from~$q$ in~$\aut$.
Given a~set of states~$T \subseteq Q$, we define $\deltaext T P$ as follows:
\begin{itemize}
  \vspace{-3mm}
  \item  $\deltaextof T P \extension a = (\delta(T \setminus P, a) \cap P) \cup A$\quad where
          $A = \begin{cases}
            \{\extension\} & \text{if } \reachof{\delta(T, a)\setminus P} \cap P \neq \emptyset, \\
            \emptyset & \text{otherwise.}
          \end{cases}$
  \item  $\deltaextof T P q a = \delta(q,a) \cap P$\quad for $q \neq \extension$.
\end{itemize}
Intuitively, $\deltaext T P$ is used to take into account
runs outside of~$P$ (represented collectively by~$\extension$).
We also extend $\deltaext T P$ to sets of states as usual.

Now, we proceed to the definition of \emph{rankings} for the modified
rank-based procedure.
A~\emph{$\qextof P$-ranking} is a~(partial) function $f\colon
\qextof P \partialto \{0, \ldots,
2|\qextof P|\}$.
%such that $f(Q_F \cap Q_j) \subseteq \{0, 2, \ldots\}$.
The \emph{rank} of~$f$ is the value $\rankof f = \max\{f(q) \mid q \in \qextof
P\}$.
For a~set $S \subseteq \qextof P$, a~ranking~$f$ is called $S$-tight if
\begin{inparaenum}[(i)]
	\item $r = \rankof f$ is an odd number,
	\item $f$ is onto $\{1, 3, \dots, r\}$, and
	\item $\dom(f) = S$.
\end{inparaenum}
Note that because in our definition, a~ranking is a~partial function, the
ranking's domain tells us which states are active; therefore, we do not need to
keep a~separate set for this (the set~$S$ used in~\cite{Schewe09}).
Furthermore, we say that a~$\qextof P$-ranking~$f$ is $\extension$-tight iff the following holds:
\begin{enumerate}[(i)]
	\item $f$ is $\dom(f)$-tight and
  \item if $\extension \in \dom(f)$ then $f(\extension) = \rankof f$, and
    $\rankof{f \setminus \{\extension \mapsto f(\extension)\}} < \rankof f$. %, where $S = \dom(f)$.
%
% 	\item if $\extension \notin \dom(f)$ then $f$ is $S$-tight, and
% 	\item if $\extension \in \dom(f)$ then $f$ is $(S\cup\{\extension\})$-tight, $f(\extension) = \max(f)$, and $\max(f_{|S}) < \max(f)$, where $S = \dom(f)$.
\end{enumerate}
Intuitively, a $\extension$-tight ranking is tight over its domain and if the
domain contains~$\extension$, the rank of~$\extension$ is strictly larger than
the rank of any state from~$P$.

For a~pair of $\qextof P$-rankings~$f$ and~$g$, we define $f \transconsistof a
T g$ iff the following hold:
\begin{enumerate}[(i)]
  \item  $\dom(g) = \delta_P^T(\dom(f), a)$,
  \item  for each $q \in \dom(f)$ and $q'\in\deltaextof T P q a$ we have $g(q') \leq f(q)$, and
  \item  for each $q \in \dom(f) \setminus \{\extension\}$ and $q'\in
    \transacc(q,a) \cap P$  it holds that $g(q') \leq \evenceil{f(q)}$.
\end{enumerate}
We further define $f_1 \rankleq f_2$ iff $\rankof{f_1} = \rankof{f_2}$ and for
each $q\in\dom(f_1)$ we have $f_1(q) \leq f_2(q)$.
We also define
\begin{equation}
\maxrank_T(f,a) = \begin{cases}
  g_{\max} & \text{if } g_{\max} = \max_{\rankleq}\{ g\mid f \transconsistof a T
  g \} \text{ is tight} \\
  \bot & \text{otherwise}
\end{cases}
\end{equation}

\newcommand{\hide}[1]{}

\hide{
Note that $\maxrank_T(f,a)$ can be computed easily by the following simple
procedure:

\begin{algorithm}[H]
  \caption{$\maxrank_T(f,a)$}
  % $\maxrank_T(f,a)$:\;
  \lForEach{$q' \in\qextof P$}{
    $K(q') \gets \{ k \in \omega \mid f(q) = k, q' \in \deltaextof T P q a\}$
  }
  $A \gets \{q' \in Q_P \mid q' \in \transacc(\domof f, a)\}$
\tcp*{or $q' \in \acc$}
  $g \gets \{q' \mapsto k_{\mathit{min}} \mid  q' \in \domof K, k_{\mathit{min}} = \min(K(q'))\}$\;
  $g \gets \{q' \mapsto g(q') \mid q' \notin A\}
  \cup \{q' \mapsto \evenceil{g(q')} \mid q' \in A\}$\;
  \lIf{$\rankof{g} \neq \rankof{f}$}{\Return $\bot$}
  \lIf{$g$ is not $\qextof P$-tight}{\Return $\bot$}
  \Return $g$\;
\end{algorithm}

We can use other heuristics to get a better max. ranking \ol{}

Further, we can try to refine the computed ranking by using
a~so-called~\emph{tight rank upper bound} (TRUB) (cf.~\cite{HavlenaLS22a}),
which gives conservative bounds on the maximum needed ranks of individual states
in macrostates.

\begin{algorithm}[H]
  \KwData{a~TRUB $\rankrestr\colon 2^{\qextof P} \to (\qextof P \to \omega)$}
  \caption{$\maxrank_T(f,a)$}
  % $\maxrank_T(f,a)$:\;
  \lForEach{$q' \in\qextof P$}{
    $K(q') \gets \{ k \in \omega \mid f(q) = k, q' \in \deltaextof T P q a\}$
  }
  $A \gets \{q' \in Q_P \mid q' \in \transacc(\domof f, a)\}$
\tcp*{or $q' \in \acc$}
  $g \gets \{q' \mapsto k_{\mathit{min}} \mid  q' \in \domof K, k_{\mathit{min}} = \min(K(q'))\}$\;
  $g \gets \{q' \mapsto \min(g(q'), r) \mid r = \rankrestr(\domof g, q'))\}$\;
  $g \gets \{q' \mapsto g(q') \mid q' \notin A\}
  \cup \{q' \mapsto \evenceil{g(q')} \mid q' \in A\}$\;
  \lIf{$\rankof{g} \neq \rankof{f}$}{\Return $\bot$}
  \lIf{$g$ is not $\qextof P$-tight}{\Return $\bot$}
  \Return $g$\;
\end{algorithm}
}

Finally, we are ready to give an instantiation of the rank-based
complementation procedure for the decomposition-based construction.
We start with the definition of types:
\begin{enumerate}
  \item  $\tracktypeof{\nacrnk_P} = \tracktype^{\wait}_P \cup \tracktype^{\tight}_P$ where
    \begin{itemize}
      \item  $\tracktype^{\wait}_P = 2^{\qextof P}$.
        This part represents the \emph{waiting part} of the complemented BA.
        Note that we use $2^{\qextof P}$ instead of $2^{P}$; this is because
        we need to keep track whether there is some run that can reach~$P$
        (represented using~$\extension$).
      \item  $\tracktype^{\tight}_P = \{f \mid %S \subseteq \qextof j,
        % f \in \tightranks_j,
        f \text{ is a } \extension\text{-tight } \qextof P\text{-ranking}
        %i \in \{0, 2, \ldots, \rankof f - 1\}
        \}$.
    \end{itemize}

  \item  $\activetypeof{\nacrnk_P} = \activetype^{\wait}_P \cup \activetype^{\tight}_P$ where
    \begin{itemize}
      \item  $\activetype^{\wait}_P = 2^{\qextof P}$ and
      \item  $\activetype^{\tight}_P = \{(f, O, i) \mid %S \subseteq \qextof j,
        % O \subseteq S \cap f^{-1}(i),
        f \text{ is a } \extension\text{-tight } \qextof P\text{-ranking},
        O \subseteq \dom(f) \cap f^{-1}(i),$ \\
        % f \in \tightranks_j,
        \hspace*{31mm}$i \in \{0, 2, \ldots, \rankof f - 1\}\}$.
    \end{itemize}
\end{enumerate}

The instance $\nacrnk_P$ implements the MaxRank construction (the version from
the paper~\cite{HavlenaL21}) without $\eta_4$ ($\eta_4$ is responsible
for nondeterministically decreasing ranks) for $\GetSuccTOf{\nacrnk_P}$.
Then, for $\GetSuccAOf{\nacrnk_P}$ we use the standard MaxRank, where a~macrostate may have two
nondeterministic successors.
Formally, the instance $\nacrnk_P$ provides the following functions:
\begin{itemize}
  \item  $\GetInitOf{\nacrnk_P} = \big\{(I \cap P) \cup A \mid \text{if }
    \reachof{I}\cap P \neq \emptyset \text{ then } A = \{\extension\} \text{ else } A =
    \emptyset\big\}$.
  \item  $\GetSuccTOf{\nacrnk_P}$ is defined as follows based on the type of the partial macrostate:
    %\vspace{-2mm}
    \begin{itemize}
      \item $\GetSuccTOf{\nacrnk_P}(T, S, a) = \big\{ (\deltaextof T P S a, \emptyset ) \big\}
        \qquad \text{ if } S \in \tracktype^{\wait}_P$ else
      \item $\GetSuccTOf{\nacrnk_P}(T, f, a) = \begin{cases}
          \{ (f', \emptyset) \} &\text{if } f' = \maxrank_T(f,a) \neq \bot,\\
          \emptyset & \text{otherwise}.
        \end{cases}$
    \end{itemize}

  \vspace{2mm}
  \item  $\LiftOf{\nacrnk_P}$ also takes into account the type of the partial macrostate:
   % \vspace{-2mm}
    \begin{itemize}
      \item  $\LiftOf{\nacrnk_P}(S) = \big\{ S \big\}
             \cup \big\{ (g, g^{-1}(0), 0) \mid
        g \in \max_{\rankleq} \{ f \mid
        f \text{ is } \extension\text{-tight}$, \\[1mm]
        \hspace*{70mm}
        $\dom(f) =  S \} \big\}$
        \vspace{-2mm}
      % \item $\GetSuccTrackToActive(T, f, a) = \{ ((g, O', 0), \colourzero) \}$ where
      %   \begin{itemize}
      %     \item  $g \in  \GetSuccTrack(T, f, a)$ and
			% % &\qquad i' = (i+2) \mod \max(g)+1\\
      %     \item $O' = g^{-1}(0)$.
      %     \item Note that it can happen that $\GetSuccTrackToActive(T, f, a) =
      %       \emptyset$, e.g., when $\GetSuccTrack(T, f, a) = \emptyset$.
      %   \end{itemize}

        \item $\LiftOf{\nacrnk_P}(f) = \{ (f, O', 0) \}$ where 
          \begin{itemize}
            \item $O' = f^{-1}(0)$.
            \item Note that it can happen that $\LiftOf{\nacrnk_P}(f) =
              \emptyset$, e.g., when \\$\GetSuccTOf{\nacrnk_P}(T, f, a) = \emptyset$.
          \end{itemize}
    \end{itemize}

  \vspace{2mm}
  \item  $\GetSuccAOf{\nacrnk_P}$ is defined as follows based on the type of the partial macrostate:
    \begin{itemize}
      \item  $\GetSuccAOf{\nacrnk_P}(T, S, a) = \begin{cases}
          \big\{ (U,  \{ \tcacc{3}{0} \} ) \big\} \text{ with } U = \deltaextof T P S a \in \tracktypeof{\nacrnk_P}  &
          \text{if } S \subseteq \{\extension\} \\
          \LiftOf{\nacrnk_P}(U) \in \activetypeof{\nacrnk_P}  & \text{otherwise}
      \end{cases}$
      \item  $\GetSuccAOf{\nacrnk_P}(T, (f,O,i), a)$: let us first define functions
        $\eta_3$ and $\eta_4$ as follows (for a~$\qextof P$-ranking~$g$
        and $i \in \omega$, we use $i\iincrof g = (i + 2) \mod (\rankof g + 1)$):
        \vspace{1mm}
        \begin{itemize}
          \item  if $\GetSuccTOf{\nacrnk_P}(T, f, a) = \{g\}$ then
            \begin{itemize}
              \item  $\eta_3 = \begin{cases}
                  \big\{ (g, \deltaextof T P O a \cap g^{-1}(i), i)\big\} &
                    \text{if } O \neq \emptyset\\
                  \big\{ (g, \deltaextof T P {\dom(f)} a \cap g^{-1}(i\iincrof g),
                    i\iincrof g)\big\} & \text{if } O = \emptyset
                \end{cases}$
              \item  $\eta_4 = \big\{ (g', M, i) \mid
                g' = g \vartriangleleft \{ q
                \mapsto g(q) - 1 \mid q \in M, i \neq 0
                \}\big\}$ where $M = \deltaextof T j O a \cap g^{-1}(i)$
            \end{itemize}
          \item  else  $\eta_3(T, (f, O, i), a) = \eta_4(T, (f, O, i), a) = \emptyset$.
        \end{itemize}
        Let $U = \{(f', O', i') \in \eta_3 \cup \eta_4 \mid O' = \emptyset \land i'
        = \rankof{f'} -1\}$.
        We then define
        $\GetSuccAOf{\nacrnk_P}(T, (f,O,i), a) = {}$
        $$\{((f', O', i'), \emptyset) \mid (f', O', i') \in \eta_3 \cup \eta_4 \setminus U\}
        \cup \{(f', \{\tcacc{3}{0}\}) \mid (f', O', i') \in U\}$$
        Note that in the previous, the first set is from $\activetypeof{\nacrnk_P}$ and
        the other is from~$\tracktypeof{\nacrnk_P}$.
    \end{itemize}
\end{itemize}

\noindent
The correctness is then summarized by the following lemmas.

\begin{lemma}\label{lem:rnk-corr}
  The partial algorithm $\algunion(\nacrnk)$ is correct.
\end{lemma}
\begin{proof} \emph{(Sketch)}
  In this proof, we use $\algmaxrank$ to denote the MaxRank algorithm 
  from~\cite{HavlenaL21} and $\algdeelev$ to denote 
  the procedure performing subset construction of the initial part of the 
  automaton with no accepting transitions.
  Consider a BA $\but$ such that $\but\modelssurrp \varphi_{\nacrnk}$. Moreover, we assume 
  that $\but$ has not redundant states. Consider a run $\rho = (H_1, M_1)(H_2, M_2)\dots$
  over a word $\word$ in $\modcompl(\algunion(\nacrnk_P), \but)$. We can construct a run $\rho'$ 
  over $\word$ in $\algdeelev(\algmaxrank(\but))$ such that $\rho'_i$ is obtained from $\rho_i$ 
  by replacing $\extension$ by $H_i \setminus \dom(f_i)$ where $f_i$ is the ranking 
  function of the macrostate $M_i$. It can be quite easily shown that $\rho$ is 
  accepting iff $\rho'$ is accepting. 
  \qed
\end{proof}

\begin{lemma}
  The partial round-robin algorithm $\nacrnk$ is correct.
\end{lemma}
\begin{proof} \emph{(Sketch)}
  We start with the condition (C1). For an arbitrary word \\$\word\in\langof{\modcompl(\algunion(\nacrnk_P), \but)}$ 
  we can construct an accepting run $\rho$ such that after we flush the $O$-set 
  ($O=\emptyset$), we can switch for a single step 
  to the passive state and then back to the active in the following step. Since we need to empty the $O$-set 
  infinitely often, it does not matter when we make a new sample (we must just ensure that we do a 
  finite number of steps in the passive phase)---which also fulfills the condition (C2). 
  The condition (C3) follows from the fact that the switch from the active to passive phase is 
  done only if the $O$-set becames empty (hence infinitely many switches mean that the run is accepting).
  The rest of the correctness follows from \cref{lem:rnk-corr}. \qed
\end{proof}

%%%%%%%%%%%%%%%%%%%%%%%%%%%%%%%%%%%%%%%%%%%%%%%%%%%%%%%%%%%%%%%%%%%
\section{Additional Examples}
\label{sec:examples}
%%%%%%%%%%%%%%%%%%%%%%%%%%%%%%%%%%%%%%%%%%%%%%%%%%%%%%%%%%%%%%%%%%%

\newcommand{\figPostponed}[0]{
\begin{figure}[t]
  \begin{subfigure}[b]{0.48\textwidth}
    \centering
    \scalebox{0.7}{% \begin{tikzpicture}[>=stealth',shorten >=0pt,node distance=1.4cm,
%                     scale=0.8,transform shape,initial text={}]
\begin{tikzpicture}[automaton]
    \tikzstyle{every state}=[inner sep=3pt,minimum size=5pt]
    \tikzstyle{empty}=[]
    \tikzstyle{initstate}=[fill=yellow!30]
    % \tikzstyle{every state}=[rectangle,rounded corners,inner sep=3pt,minimum size=5pt]
    \tikzstyle{uberstate}=[
      rounded corners,draw,anchor=base,
      rectangle split,rectangle split horizontal,rectangle split parts=2,
      rectangle split part align=base,
      rectangle split part fill={black!20, blue!30, green!30}]
    \newcommand{\ustate}[4]{$#1$\nodepart{two}$#2, #3, #4$}
  
    \node[uberstate,initial] (p) {\ustate p \emptyset \emptyset \emptyset};
    \node[uberstate,below of=p,xshift=0mm] (pq) {\ustate{p+q} q \emptyset q};
    \node[uberstate,below of=pq,xshift=0mm] (pqsafe) {\ustate{p+q} \emptyset q \emptyset};
  
    \node[uberstate,right of=p,xshift=22mm] (pqr) {\ustate{p+q+r} q \emptyset q};
    \node[uberstate,below of=pqr,xshift=0mm] (pqrs) {\ustate{p+q+r+s} q \emptyset q };
    \node[uberstate,below of=pqrs,xshift=0mm] (pqrssafe) {\ustate{p+q+r+s} \emptyset q \emptyset};
  
    \path[->]
      (p) edge pic[pos=0.5] {acc=0} pic[auto] {l=$b$} (pq)
      (pq) edge pic[pos=0.5] {acc=0} pic[auto] {l=$b$} (pqsafe)
      (pq.182) edge[out=210,in=150,loop,distance=9mm] pic[auto] {l=$b$} (pq.178)
      (pqsafe.182) edge[out=210,in=150,loop,distance=9mm] pic[pos=0.5] {acc=0} pic[auto] {l=$b$} (pqsafe.178)
      (p) edge pic[pos=0.5] {acc=0} pic[auto] {l=$a$} (pqr)
      (pqr) edge[bend left] pic[auto] {l=$a$} (pqrs)
      (pqr) edge[bend right] node[left] {$b$} (pqrs)
      (pqrs.12) edge[out=120,in=60,loop,distance=6mm] node[auto] {$a,b$} (pqrs.11)
      (pq) edge pic[auto] {l=$a$} (pqrs)
      (pqrssafe.182) edge[out=210,in=150,loop,distance=9mm] pic {acc=0} pic[auto] {l=$b$} (pqrssafe.178)
      (pqrs) edge pic {acc=0} pic[left] {l=$b$} (pqrssafe)
      (pqr.355) edge[bend left=40] pic {acc=0} pic[auto] {l=$b$} (pqrssafe.2)
      ;
  \end{tikzpicture}

  }
    \caption{$\aut_1 = \modcompl(\dac_{P_0}, \autex)$}
  \end{subfigure}
  \begin{subfigure}[b]{0.48\textwidth}
    \centering
    \scalebox{0.7}{% \begin{tikzpicture}[>=stealth',shorten >=0pt,node distance=1.4cm,
%                     scale=0.8,transform shape,initial text={}]
\begin{tikzpicture}[automaton]
    \tikzstyle{every state}=[inner sep=3pt,minimum size=5pt]
    \tikzstyle{empty}=[]
    \tikzstyle{initstate}=[fill=yellow!30]
    % \tikzstyle{every state}=[rectangle,rounded corners,inner sep=3pt,minimum size=5pt]
    \tikzstyle{uberstate}=[
      rounded corners,draw,anchor=base,
      rectangle split,rectangle split horizontal,rectangle split parts=2,
      rectangle split part align=base,
      rectangle split part fill={black!20, green!30}]
    \newcommand{\ustate}[3]{$#1$\nodepart{two}$#2, #3$}
  
  %  \node [my shape=5, rectangle split horizontal] at (2,2)
  %     {1\nodepart{two}2\nodepart{three}3\nodepart{four}4\nodepart{five}5};
  
    \node[uberstate,initial] (p) {\ustate p \emptyset \emptyset};
    \node[uberstate,below of=p,xshift=0mm] (pq) {\ustate{p+q} \emptyset \emptyset};
    %\node[uberstate,below of=pq,xshift=0mm] (pqsafe) {\ustate{p+q}  \emptyset \emptyset};
  
    \node[uberstate,right of=p,xshift=22mm] (pqr) {\ustate{p+q+r}  r r};
    \node[uberstate,below of=pqr,xshift=0mm] (pqrs) {\ustate{p+q+r+s}  {r+s} {r+s}};
    \node[uberstate,below of=pqrs,xshift=0mm] (pqrsbreak) {\ustate{p+q+r+s} {r+s} {r}};
    %\node[uberstate,below of=pqrsbreak,xshift=0mm] (pqrssafe) {\ustate{p+q+r+s}  {r+s} {r+s}};
  
    \path[->]
      (p) edge pic[pos=0.5] {acc=1} pic[auto] {l=$b$} (pq)
      (pq.182) edge[out=210,in=150,loop,distance=9mm] pic {acc=1} pic[auto] {l=$b$} (pq.178)
      (p) edge pic[pos=0.5] {acc=1} pic[auto] {l=$a$} (pqr)
      (pqr) edge[bend left] pic {acc=1} pic[auto] {l=$a$} (pqrs)
      (pqr) edge[bend right] node[left] {$b$} (pqrs)
      (pqrs.12) edge[out=120,in=60,loop,distance=6mm] node[auto] {$b$} (pqrs.11)
      (pq) edge pic {acc=1} pic[auto] {l=$a$} (pqrs)
      (pqrs) edge[bend left] node[auto] {$a$} (pqrsbreak)
      (pqrsbreak) edge[bend left] node[auto] {$b$} (pqrs)
      (pqrsbreak.160) edge[bend left] pic {acc=1} pic[auto] {l=$a$} (pqrs.200)
      ;
  \end{tikzpicture}

  }
    \caption{$\aut_2 = \modcompl(\iwc_{P_1}, \autex)$}
  \end{subfigure}

  \vspace{5mm}
  \begin{subfigure}[b]{0.25\textwidth}
    \centering
    \scalebox{0.8}{\begin{tikzpicture}[automaton]
    \tikzstyle{every state}=[inner sep=3pt,minimum size=5pt]
    \tikzstyle{empty}=[]
    \tikzstyle{initstate}=[fill=yellow!30]

    \node[state, initstate, initial] (p) {$p$};
    \node[state, right of=p] (q) {$q$};
  
    \path[->]
      (p) edge pic {acc=0} pic[auto] {l=$b$} (q)
      (p) edge[loop above] node {$a,b$} (p)
      (q) edge[loop above] pic {acc=0} pic[auto] {l=$b$} (q)
      ;
  \end{tikzpicture}}
    \caption{$\reduce(\aut_1)$}
  \end{subfigure}
  \begin{subfigure}[b]{0.25\textwidth}
    \centering
    \scalebox{0.8}{\begin{tikzpicture}[automaton]
    \tikzstyle{every state}=[inner sep=3pt,minimum size=5pt]
    \tikzstyle{empty}=[]
    \tikzstyle{initstate}=[fill=yellow!30]

    \node[state, initstate, initial] (p) {$0$};
    \node[state, right of=p] (q) {$2$};
    \node[state, below of=p] (r) {$1$};
    \node[state, right of=r] (s) {$3$};
  
    \path[->]
      (p) edge node[above] {$a$} (q)
      (p) edge node[left] {$b$} (r)
      (r) edge[loop below] pic {acc=1} pic[auto] {l=$b$} (r)
      (r) edge node[above] {$a$} (s)
      (q) edge[bend left] pic {acc=1} pic[left] {l=$a$} (s)
      (q) edge[bend right] node[left] {$b$} (s)
      (s) edge[loop below] node {$b$} (s)
      (s) edge[out=10,in=-10] node[right] {$a$} (q)
      ;
    %   
    %   (s) edge[loop below] node {$b$} (s)
    %   (s) edge[out=10] node {$a$} (q)
    %   ;
  \end{tikzpicture}}
    \caption{$\reduce(\aut_2)$}
  \end{subfigure}
  \begin{subfigure}[b]{0.45\textwidth}
    \centering
    \scalebox{0.8}{\begin{tikzpicture}[automaton]
    \tikzstyle{every state}=[inner sep=3pt,minimum size=5pt,rectangle, rounded corners=2mm]
    \tikzstyle{empty}=[]
    \tikzstyle{initstate}=[fill=yellow!30]

    \node[state, initstate, initial] (p) {$(p,0)$};
    \node[state, right of=p] (q) {$(p,2)$};
    \node[state, below of=p] (r) {$(p,1)$};
    \node[state, right of=r] (s) {$(p,3)$};
    \node[state, above of=p] (t) {$(q,1)$};
    \node[state, above of=q] (u) {$(q,3)$};
  
    \path[->]
      (p) edge node[above] {$a$} (q)
      (p) edge node[left] {$b$} (r)
      (r) edge[loop below] pic {acc=1} pic[auto] {l=$b$} (r)
      (r) edge node[above] {$a$} (s)
      (q) edge[bend left] pic {acc=1} pic[left] {l=$a$} (s)
      (q) edge[bend right] node[left] {$b$} (s)
      (s) edge[loop below] node {$b$} (s)
      (s) edge[out=10,in=-10] node[left] {$a$} (q)
      (p) edge pic {acc=0} pic[left] {l=$b$} (t)
      (q) edge pic {acc=0} pic[left] {l=$b$} (u)
      (r) edge[out=140,in=-140] pic {acc=0} pic[left] {l=$b$} (t)
      (s) edge[out=10,in=-10] pic {acc=0} pic[left] {l=$b$} (u)
      (u) edge[loop above] pic {acc=0} pic[auto] {l=$b$} (u)
      (t) edge[loop above,min distance=10mm,in=60,out=120] pic[pos=0.3] {acc=0} pic[pos=0.7] {acc=1} pic[auto] {l=$b$} (t)
      ;
  \end{tikzpicture}}
    \caption{$\postponedcompl(\dac_{P_0}, \iwc_{P_1},\autex)$}
  \end{subfigure}

  \caption{Example of the postponed construction applied on $\autex$ with the 
  result's accepting condition $\acccond\colon \Infof{\protect\tacc{0}} \land \Infof{\protect\tacc{1}}$.
  }
  \label{fig:postponed-example}
\end{figure}
}

\begin{figure}[t]
  \begin{subfigure}[b]{0.48\textwidth}
    \centering
    \scalebox{0.7}{% \begin{tikzpicture}[>=stealth',shorten >=0pt,node distance=1.4cm,
%                     scale=0.8,transform shape,initial text={}]
\begin{tikzpicture}[automaton]
    \tikzstyle{every state}=[inner sep=3pt,minimum size=5pt]
    \tikzstyle{empty}=[]
    \tikzstyle{initstate}=[fill=yellow!30]
    % \tikzstyle{every state}=[rectangle,rounded corners,inner sep=3pt,minimum size=5pt]
    \tikzstyle{uberstate}=[
      rounded corners,draw,anchor=base,
      rectangle split,rectangle split horizontal,rectangle split parts=2,
      rectangle split part align=base,
      rectangle split part fill={black!20, blue!30, green!30}]
    \newcommand{\ustate}[4]{$#1$\nodepart{two}$#2, #3, #4$}
  
    \node[uberstate,initial] (p) {\ustate p \emptyset \emptyset \emptyset};
    \node[uberstate,below of=p,xshift=0mm] (pq) {\ustate{p+q} q \emptyset q};
    \node[uberstate,below of=pq,xshift=0mm] (pqsafe) {\ustate{p+q} \emptyset q \emptyset};
  
    \node[uberstate,right of=p,xshift=22mm] (pqr) {\ustate{p+q+r} q \emptyset q};
    \node[uberstate,below of=pqr,xshift=0mm] (pqrs) {\ustate{p+q+r+s} q \emptyset q };
    \node[uberstate,below of=pqrs,xshift=0mm] (pqrssafe) {\ustate{p+q+r+s} \emptyset q \emptyset};
  
    \path[->]
      (p) edge pic[pos=0.5] {acc=0} pic[auto] {l=$b$} (pq)
      (pq) edge pic[pos=0.5] {acc=0} pic[auto] {l=$b$} (pqsafe)
      (pq.182) edge[out=210,in=150,loop,distance=9mm] pic[auto] {l=$b$} (pq.178)
      (pqsafe.182) edge[out=210,in=150,loop,distance=9mm] pic[pos=0.5] {acc=0} pic[auto] {l=$b$} (pqsafe.178)
      (p) edge pic[pos=0.5] {acc=0} pic[auto] {l=$a$} (pqr)
      (pqr) edge[bend left] pic[auto] {l=$a$} (pqrs)
      (pqr) edge[bend right] node[left] {$b$} (pqrs)
      (pqrs.12) edge[out=120,in=60,loop,distance=6mm] node[auto] {$a,b$} (pqrs.11)
      (pq) edge pic[auto] {l=$a$} (pqrs)
      (pqrssafe.182) edge[out=210,in=150,loop,distance=9mm] pic {acc=0} pic[auto] {l=$b$} (pqrssafe.178)
      (pqrs) edge pic {acc=0} pic[left] {l=$b$} (pqrssafe)
      (pqr.355) edge[bend left=40] pic {acc=0} pic[auto] {l=$b$} (pqrssafe.2)
      ;
  \end{tikzpicture}

  }
    \caption{$\aut_1 = \modcompl(\dac_{P_0}, \autex)$}
  \end{subfigure}
  \begin{subfigure}[b]{0.48\textwidth}
    \centering
    \scalebox{0.7}{% \begin{tikzpicture}[>=stealth',shorten >=0pt,node distance=1.4cm,
%                     scale=0.8,transform shape,initial text={}]
\begin{tikzpicture}[automaton]
    \tikzstyle{every state}=[inner sep=3pt,minimum size=5pt]
    \tikzstyle{empty}=[]
    \tikzstyle{initstate}=[fill=yellow!30]
    % \tikzstyle{every state}=[rectangle,rounded corners,inner sep=3pt,minimum size=5pt]
    \tikzstyle{uberstate}=[
      rounded corners,draw,anchor=base,
      rectangle split,rectangle split horizontal,rectangle split parts=2,
      rectangle split part align=base,
      rectangle split part fill={black!20, green!30}]
    \newcommand{\ustate}[3]{$#1$\nodepart{two}$#2, #3$}
  
  %  \node [my shape=5, rectangle split horizontal] at (2,2)
  %     {1\nodepart{two}2\nodepart{three}3\nodepart{four}4\nodepart{five}5};
  
    \node[uberstate,initial] (p) {\ustate p \emptyset \emptyset};
    \node[uberstate,below of=p,xshift=0mm] (pq) {\ustate{p+q} \emptyset \emptyset};
    %\node[uberstate,below of=pq,xshift=0mm] (pqsafe) {\ustate{p+q}  \emptyset \emptyset};
  
    \node[uberstate,right of=p,xshift=22mm] (pqr) {\ustate{p+q+r}  r r};
    \node[uberstate,below of=pqr,xshift=0mm] (pqrs) {\ustate{p+q+r+s}  {r+s} {r+s}};
    \node[uberstate,below of=pqrs,xshift=0mm] (pqrsbreak) {\ustate{p+q+r+s} {r+s} {r}};
    %\node[uberstate,below of=pqrsbreak,xshift=0mm] (pqrssafe) {\ustate{p+q+r+s}  {r+s} {r+s}};
  
    \path[->]
      (p) edge pic[pos=0.5] {acc=1} pic[auto] {l=$b$} (pq)
      (pq.182) edge[out=210,in=150,loop,distance=9mm] pic {acc=1} pic[auto] {l=$b$} (pq.178)
      (p) edge pic[pos=0.5] {acc=1} pic[auto] {l=$a$} (pqr)
      (pqr) edge[bend left] pic {acc=1} pic[auto] {l=$a$} (pqrs)
      (pqr) edge[bend right] node[left] {$b$} (pqrs)
      (pqrs.12) edge[out=120,in=60,loop,distance=6mm] node[auto] {$b$} (pqrs.11)
      (pq) edge pic {acc=1} pic[auto] {l=$a$} (pqrs)
      (pqrs) edge[bend left] node[auto] {$a$} (pqrsbreak)
      (pqrsbreak) edge[bend left] node[auto] {$b$} (pqrs)
      (pqrsbreak.160) edge[bend left] pic {acc=1} pic[auto] {l=$a$} (pqrs.200)
      ;
  \end{tikzpicture}

  }
    \caption{$\aut_2 = \modcompl(\iwc_{P_1}, \autex)$}
  \end{subfigure}

  \vspace{5mm}
  \begin{subfigure}[b]{0.25\textwidth}
    \centering
    \scalebox{0.8}{\begin{tikzpicture}[automaton]
    \tikzstyle{every state}=[inner sep=3pt,minimum size=5pt]
    \tikzstyle{empty}=[]
    \tikzstyle{initstate}=[fill=yellow!30]

    \node[state, initstate, initial] (p) {$p$};
    \node[state, right of=p] (q) {$q$};
  
    \path[->]
      (p) edge pic {acc=0} pic[auto] {l=$b$} (q)
      (p) edge[loop above] node {$a,b$} (p)
      (q) edge[loop above] pic {acc=0} pic[auto] {l=$b$} (q)
      ;
  \end{tikzpicture}}
    \caption{$\reduce(\aut_1)$}
  \end{subfigure}
  \begin{subfigure}[b]{0.25\textwidth}
    \centering
    \scalebox{0.8}{\begin{tikzpicture}[automaton]
    \tikzstyle{every state}=[inner sep=3pt,minimum size=5pt]
    \tikzstyle{empty}=[]
    \tikzstyle{initstate}=[fill=yellow!30]

    \node[state, initstate, initial] (p) {$0$};
    \node[state, right of=p] (q) {$2$};
    \node[state, below of=p] (r) {$1$};
    \node[state, right of=r] (s) {$3$};
  
    \path[->]
      (p) edge node[above] {$a$} (q)
      (p) edge node[left] {$b$} (r)
      (r) edge[loop below] pic {acc=1} pic[auto] {l=$b$} (r)
      (r) edge node[above] {$a$} (s)
      (q) edge[bend left] pic {acc=1} pic[left] {l=$a$} (s)
      (q) edge[bend right] node[left] {$b$} (s)
      (s) edge[loop below] node {$b$} (s)
      (s) edge[out=10,in=-10] node[right] {$a$} (q)
      ;
    %   
    %   (s) edge[loop below] node {$b$} (s)
    %   (s) edge[out=10] node {$a$} (q)
    %   ;
  \end{tikzpicture}}
    \caption{$\reduce(\aut_2)$}
  \end{subfigure}
  \begin{subfigure}[b]{0.45\textwidth}
    \centering
    \scalebox{0.8}{\begin{tikzpicture}[automaton]
    \tikzstyle{every state}=[inner sep=3pt,minimum size=5pt,rectangle, rounded corners=2mm]
    \tikzstyle{empty}=[]
    \tikzstyle{initstate}=[fill=yellow!30]

    \node[state, initstate, initial] (p) {$(p,0)$};
    \node[state, right of=p] (q) {$(p,2)$};
    \node[state, below of=p] (r) {$(p,1)$};
    \node[state, right of=r] (s) {$(p,3)$};
    \node[state, above of=p] (t) {$(q,1)$};
    \node[state, above of=q] (u) {$(q,3)$};
  
    \path[->]
      (p) edge node[above] {$a$} (q)
      (p) edge node[left] {$b$} (r)
      (r) edge[loop below] pic {acc=1} pic[auto] {l=$b$} (r)
      (r) edge node[above] {$a$} (s)
      (q) edge[bend left] pic {acc=1} pic[left] {l=$a$} (s)
      (q) edge[bend right] node[left] {$b$} (s)
      (s) edge[loop below] node {$b$} (s)
      (s) edge[out=10,in=-10] node[left] {$a$} (q)
      (p) edge pic {acc=0} pic[left] {l=$b$} (t)
      (q) edge pic {acc=0} pic[left] {l=$b$} (u)
      (r) edge[out=140,in=-140] pic {acc=0} pic[left] {l=$b$} (t)
      (s) edge[out=10,in=-10] pic {acc=0} pic[left] {l=$b$} (u)
      (u) edge[loop above] pic {acc=0} pic[auto] {l=$b$} (u)
      (t) edge[loop above,min distance=10mm,in=60,out=120] pic[pos=0.3] {acc=0} pic[pos=0.7] {acc=1} pic[auto] {l=$b$} (t)
      ;
  \end{tikzpicture}}
    \caption{$\postponedcompl(\dac_{P_0}, \iwc_{P_1},\autex)$}
  \end{subfigure}

  \caption{Example of the postponed construction applied on $\autex$ with the 
  result's accepting condition $\acccond\colon \Infof{\protect\tacc{0}} \land \Infof{\protect\tacc{1}}$.
  }
  \label{fig:postponed-example}
\end{figure}
   %%%%%%%%%%%%

\begin{figure}[t]
  \centering
  \scalebox{0.8}{
    % \begin{tikzpicture}[>=stealth',shorten >=0pt,node distance=1.4cm,
%                     scale=0.8,transform shape,initial text={}]
\begin{tikzpicture}[automaton]
    \tikzstyle{every state}=[inner sep=3pt,minimum size=5pt]
    \tikzstyle{empty}=[]
    \tikzstyle{initstate}=[fill=yellow!30]
    % \tikzstyle{every state}=[rectangle,rounded corners,inner sep=3pt,minimum size=5pt]
    \tikzstyle{uberstate}=[
      rounded corners,draw,anchor=base,
      rectangle split,rectangle split horizontal,rectangle split parts=4,
      rectangle split part align=base,
      rectangle split part fill={black!20, blue!30, green!30, orange!30}]
    \newcommand{\ustate}[4]{$#1$\nodepart{two}$#2$\nodepart{three}$#3$\nodepart{four} $#4$}
  
  %  \node [my shape=5, rectangle split horizontal] at (2,2)
  %     {1\nodepart{two}2\nodepart{three}3\nodepart{four}4\nodepart{five}5};
  
    \node[uberstate,initial] (p) {\ustate p {\emptyset,\emptyset,\emptyset} \emptyset 1};
    \node[uberstate,below of=p,xshift=0mm] (pq) {\ustate{p+q} {q,\emptyset,q} \emptyset 1};
    \node[uberstate,below of=pq,xshift=0mm] (pqsafedac) {\ustate{p+q} {\emptyset,q} {\emptyset,\emptyset} 2};
    \node[uberstate,below of=pqsafedac,xshift=0mm] (pqsafeiwa) {\ustate{p+q} {\emptyset,q,\emptyset} {\emptyset} 1};
  
    \node[uberstate,right of=p,xshift=35mm] (pqr) {\ustate{p+q+r} {q,\emptyset,q} r 1};
    \node[uberstate,below of=pqr,xshift=0mm] (pqrs) {\ustate{p+q+r+s} {q,\emptyset,q} {r+s} 1};
    %\node[uberstate,below of=pqrs,xshift=0mm] (pqrsbreak) {\ustate{p+q+r+s} q \emptyset q {r+s} {r}};
    \node[uberstate,below of=pqrs,xshift=0mm] (pqrssafe) {\ustate{p+q+r+s} {\emptyset,q} {r+s, r+s} 2};
  
    \path[->]
      (p) edge pic[pos=0.5] {acc=0} pic[auto] {l=$b$} (pq)
      (pq) edge pic[pos=0.5] {acc=0} pic[auto] {l=$b$} (pqsafedac)
      (pqsafedac) edge[bend left] pic[pos=0.5] {acc=1} pic[auto] {l=$b$} (pqsafeiwa)
      (pqsafeiwa) edge[bend left] pic[pos=0.5] {acc=0} pic[auto] {l=$b$} (pqsafedac)
      (pq.182) edge[out=210,in=150,loop,distance=9mm] pic[auto] {l=$b$} (pq.178)
      %(pqsafedac.182) edge[out=210,in=150,loop,distance=9mm] pic[pos=0.25] {acc=0} pic[pos=0.75] {acc=1} pic[auto] {l=$b$} (pqsafedac.178)
      (p) edge pic[pos=0.5] {acc=0} pic[auto] {l=$a$} (pqr)
      (pqr) edge[bend left] pic[auto] {l=$a$} (pqrs)
      (pqr) edge[bend right] node[left] {$b$} (pqrs)
      (pqrs.12) edge[out=120,in=60,loop,distance=6mm] node[auto] {$a,b$} (pqrs.11)
      (pq) edge pic[auto] {l=$a$} (pqrs)
    %   (pqrs) edge[bend left] node[auto] {$a$} (pqrsbreak)
    %   (pqrsbreak) edge[bend left] node[auto] {$b$} (pqrs)
    %   (pqrsbreak.160) edge[bend left] pic {acc=1} pic[auto] {l=$a$} (pqrs.200)
    %   (pqrsbreak) edge pic {acc=0} pic[auto] {l=$b$} (pqrssafe)
      (pqrssafe.182) edge[out=210,in=150,loop,distance=9mm] node[auto] {$b$} (pqrssafe.178)
      (pqrs) edge pic {acc=0} pic[left] {l=$b$} (pqrssafe)
      (pqr.355) edge[bend left=40] pic {acc=0} pic[auto] {l=$b$} (pqrssafe.2)
      ;
  \end{tikzpicture}

  
  }
  \caption{The outcome of $\modcomplrr(\dacrr_{P_0}, \iwcrr_{P_1},\autex)$ with $\acccond\colon \Infof{\protect\tacc{0}} \land \Infof{\protect\tacc{1}}$ applied on the BA from \cref{fig:example}. }
  \label{fig:round-robin-ex}
\end{figure}

In this section, we provide additional examples to the optimization. The 
example of the postponed construction depicting also intermediate steps of the construction is shown in \cref{fig:postponed-example}.
The example of the round-robin algorithm is shown in \cref{fig:round-robin-ex}.

The example of the complementation of initial deterministic partition block is shown in \cref{fig:example-aut-idec-result}:

\begin{figure}[t]
  \centering
%   \resizebox{\linewidth}{!}{
    \begin{tikzpicture}
      \node (bex) {\begin{tikzpicture}[automaton]
    \tikzstyle{every state}=[inner sep=3pt,minimum size=5pt]
    \tikzstyle{empty}=[]
    \tikzstyle{initstate}=[fill=yellow!30]
  
    \path[use as bounding box] (-1.5,-1.125) rectangle (2.75,1.5);

    \node[state,initial,initstate] (p) at (0,0) {$p$};
    \node[state,right of=p,xshift=0mm] (q) {$q$};
    
    \node[font=\large] at (-1.25,0) {$\butex$};

    \draw[dotted] ($(p.north west) + (-0.15,0.9)$) rectangle ($(p.south east) + (0.15,-0.9)$);
    \node[anchor=south] at ($(p.north) + (0,0.75)$) {$P_{0}$};
    \draw[dotted] ($(q.north west) + (-0.15,0.9)$) rectangle ($(q.south east) + (0.15,-0.15)$);
    \node[anchor=south] at ($(q.north) + (0,0.75)$) {$P_{1}$};

    % \draw[dotted] ($(p.north west) + (-0.15,0.9)$) rectangle ($(p.south east) + (0.15,-0.15)$);
    % \draw[dotted] ($(q.north west) + (-0.15,0.75)$) rectangle ($(q.south east) + (0.75,-0.15)$);
    % \draw[dotted] ($(s.north west) + (-0.35,0.15)$) rectangle ($(r.south east) + (0.75,-0.15)$);
  
    \path[->]
      (p) edge[loop above] node[auto] {$a$} node[anchor=center] {$\bullet$} (p)
      (p) edge[loop below] node[auto] {$b$} (p)
      (p) edge node[above] {$a$} (q)
      (q) edge[loop above] node[auto] {$b$} node[anchor=center] {$\bullet$} (q)
      ;
  \end{tikzpicture}};
      \node[anchor=west] (comp) at (bex.east) {\begin{tikzpicture}[automaton]
    \tikzstyle{every state}=[inner sep=3pt,minimum size=5pt]
    \tikzstyle{empty}=[]
    \tikzstyle{initstate}=[fill=yellow!30]
    % \tikzstyle{every state}=[rectangle,rounded corners,inner sep=3pt,minimum size=5pt]
    \tikzstyle{uberstate}=[
      rounded corners,draw,anchor=base,
      rectangle split,rectangle split horizontal,rectangle split parts=3,
      rectangle split part align=base,
      rectangle split part fill={black!20, blue!30, green!30}]
    \newcommand{\ustate}[4]{$#1$\nodepart{two}$#2$\nodepart{three}$#3, #4$}
  
    \path[use as bounding box] (-1.5,-1.125) rectangle (4.75,1.5);
  
    \node[uberstate,initial] (p) at (0,0) {\ustate p p \emptyset \emptyset};
    \node[uberstate,right of=p,xshift=15mm] (pq) {\ustate{p+q} p q q};

    \path[->]
      (p) edge pic[pos=0.3] {cacc={4}{0}} pic[pos=0.6] {acc=1} pic[auto] {l=$a$} (pq)
      (p) edge[loop above,distance=9mm] pic {acc=1} pic[auto] {l=$b$} (p)
      (pq) edge[loop above,distance=9mm] node {$b$} (pq)
      (pq) edge[in=-115,out=-65,loop,distance=9mm] pic[pos=0.25] {acc=1} pic[pos=0.75] {cacc={4}{0}} pic[auto] {l=$a$} (pq)
      ;
    %   (p) edge[loop above] node {$a,b$} (p)
    %   (p) edge node[above] {$a,b$} (pq)
    %   (q) edge[loop above] pic {l=$a$} pic[anchor=center] {acc=0} (pq)
    %   (pq) edge[loop below] node {$b$} (pq)
    %   (q) edge node[auto] {$a$} (s)
    %   (p) edge node[below] {$a$} (r)
    %   (r) edge[loop below] pic{l=$b$} pic[anchor=center] {acc=0} (r)
    %   (r) edge[bend left] node[above] {$b$} (s)
    %   (s) edge[bend left] pic[below]{l=$a$} pic{acc=0} (r)
    %   ;
  \end{tikzpicture}};
    \end{tikzpicture}
%   }
  \caption{%
    Left: $\butex$.
    \mbox{Right: $\modcompl(\idac_{P_{0}}, \iwc_{P_{1}},\butex)$ with $\acccond\colon \Finof {\protect\tcacc{4}{0}} \wedge \Infof{\protect\tacc{1}}$.}
    }
  \label{fig:example-aut-idec-result}  
%   \begin{minipage}[b]{3.5cm}
%     \centering
%     \scalebox{0.8}{
%     \input{figs/example-aut-idec.tikz}
%     }
%     \caption{$\butex$ with $P_0 = \{p\}$ and $P_1 = \{ q \}$.}
%     \label{fig:example-aut-idec}
%   \end{minipage}
%   \begin{minipage}[b]{9cm}
%     \centering
%     \scalebox{0.9}{
%     \input{figs/ex-idec-res.tikz}
%     }
%     \caption{Result of running $\modcompl(\idac_{P_0}, \iwc_{P_1},\butex)$ with the accepting condition $\Finof {\protect\tacc{0}} \wedge \Infof{\protect\tacc{1}}$.}
%     \label{fig:example-idec-result}
%   \end{minipage}
\end{figure}
   %%%%%%%%%%%%%%%

%%%%%%%%%%%%%%%%%%%%%%%%%%%%%%%%%%%%%%%%%%%%%%%%%%%%%%%%%%%%%%%%%%%
\section{Missing Proofs from the Main Text}
%%%%%%%%%%%%%%%%%%%%%%%%%%%%%%%%%%%%%%%%%%%%%%%%%%%%%%%%%%%%%%%%%%%

%%%%%%%%%%%%%%%%%%%%%%%%%%%%%%%%%%%%%%%%%%%%%%%%%%%%%%%%%%%%%%%%%%%
\subsection{Proofs of \cref{sec:modular}}
%%%%%%%%%%%%%%%%%%%%%%%%%%%%%%%%%%%%%%%%%%%%%%%%%%%%%%%%%%%%%%%%%%%

\thmSynchrCorr*

\begin{proof}
  Let $\aut = (\states, \trans, \inits, \acc)$ be a BA. 
  Moreover, since $P_1, \ldots, P_n$  is partitioning of $\aut$, 
  we have that 
  \begin{equation}\label{eq:union-all}
    \bigcup_{i=1}^n \langof{\aut_{P_i}} = \langof{\aut}
  \end{equation}
  In the first part of the proof, we prove the following claim
  
  \begin{claim}\label{claim-intersection}
  \begin{equation}\label{eq:intersection}
    \bigcap_{i=1}^n\langof{\modcompl(\alg^i_{P_i}, \aut_{P_i})} = 
    \langof{\modcompl(\alg^1_{P_1}, \dots, \alg^n_{P_n},\aut)}
  \end{equation}
  \end{claim}
  \begin{claimproof}
  $(\subseteq)$ Consider a word $\alpha \in \bigcap_{i=1}^n\langof{\modcompl(\alg^i_{P_i}, \aut_{P_i})}$. 
  Then, there are accepting runs $\rho_i$ of the form $\rho_i = (H_i^1,M_i^1)(H_i^2,M_i^2)\dots$ such that 
  $H_i^{\ell+1} = \delta(H_i^\ell, \alpha_\ell)$ and $M_i^{\ell+1} \in \GetSuccOf{\alg^i_{P_i}}(H_i^\ell, M_i^\ell, \alpha_\ell)$
  for each $1\leq i \leq n$. From the definition of these runs we have that $H_i^\ell = H_j^\ell$ for 
  each $1\leq i,j \leq n$. Therefore, there is also a run 
  $\varrho = (H_1^1,M_1^1,M_2^1,\dots,M_n^1)(H_1^2,M_1^2,M_2^2,\dots,M_n^2)\dots$ over $\alpha$ in 
  the automaton $\modcompl(\alg^1_{P_1}, \dots, \alg^n_{P_n},\aut)$.
  Since $\rho_i \models \GetAccOf{\alg^i_{P_i}}$ for each $i$, we also have 
  $\varrho \models \bigwedge_{i=1}^n\GetAccOf{\alg^i_{P_i}}$ implying that $\varrho$ is 
  accepting in $\modcompl(\alg^1_{P_1}, \dots, \alg^n_{P_n},\aut)$.

  $(\supseteq)$  
  Consider a word $\alpha \in\langof{\modcompl(\alg^1_{P_1}, \dots, \alg^n_{P_n},\aut)}$. Then, there 
  is an accepting run $\varrho = (H^1,M_1^1,M_2^1,\dots,M_n^1)(H^2,M_1^2,M_2^2,\dots,M_n^2)\dots$ over $\alpha$.
  From the definition of $\modcompl$, we have that there are runs $\rho_i = (H^1,M_i^1)(H^2,M_i^2)\dots$
  on $\alpha$ in $\modcompl(\alg^i_{P_i}, \aut_{P_i})$ for each $1\leq i \leq n$.
  Since, $\varrho \models \bigwedge_{i=1}^n\GetAccOf{\alg^i_{P_i}}$, we have that 
  each $\rho_i$ is accepting as well. 
  \end{claimproof}
  
  Then we proceed as follows. From \cref{def:alg-correct} we get 
  $$
    \langof{\modcompl(\alg^i_{P_i}, \aut_{P_i})} = \Sigma^\omega\setminus\langof{\aut_{P_i}}
  $$
  and hence using \eqref{eq:union-all}
  $$
  \bigcap_{i=1}^n\langof{\modcompl(\alg^i_{P_i}, \aut_{P_i})} = \Sigma^\omega\setminus \left( \bigcup_{i=1}^n \langof{\aut_{P_i}} \right) = \Sigma^\omega\setminus \langof{\aut},
  $$
  which, together with \eqref{eq:intersection}, concludes the proof.
  \qed
\end{proof}

%%%%%%%%%%%%%%%%%%%%%%%%%%%%%%%%%%%%%%%%%%%%%%%%%%%%%%%%%%%%%%%%%%%
\subsection{Proofs of \cref{sec:modular-elevator}}
%%%%%%%%%%%%%%%%%%%%%%%%%%%%%%%%%%%%%%%%%%%%%%%%%%%%%%%%%%%%%%%%%%%

\lemCorrIWCS*

\begin{proof}
  Let $w$ be an $\omega$-word.
  Our goal is to prove that $w$ is not accepted within $P$ if and only if $\tcacc{1}{0}$ occurs infinitely often.

\iffalse  
  First, assume that there exists an accepting run $\rho $ of $\aut$ over $w$ that eventually stay in the partition $P$.
  Our goal is to prove that the sequence of macrostates $(C_0, B_0) \cdots (C_i, B_i) \cdots$ will not emit infinitely often color $\tcacc{1}{0}$.
  By assumption, there exists a smallest integer $k > 0$ such that $\rho_{j} \in P$ for all $j \geq k$.
  By definition of $C' = \trans(H, a) \cap P$, we have $\rho_k \in C_k$.
  If $\trans(B_j, w_j) \cap P \neq \emptyset$ for all $j \geq k$, the goal clearly holds
  (Note that $\trans(B_j, w_j) \cap P \equiv \trans(B_j, w_j) \cap (\trans(H_j, w_j) \cap P)$ since $B_j \subseteq C_j \subseteq H_j$.).
  Otherwise there must exist a smallest integer, say $\ell > k$, such that $B^{\star}_{\ell+1} = \trans(B_{\ell}, w_{\ell}) \cap P = \emptyset$.
  It follows that $B_{\ell + 1} = C_{\ell+1} = \trans(H_{\ell}, w_{\ell}) \cap P$.
  Since $(\rho_{\ell}, w_{\ell}, \rho_{\ell+1}) \in \trans$ and $\rho_{\ell + 1} \in \trans(H_{\ell}, w_{\ell}) \cap P$, we have $\rho_{\ell+1} \in B_{\ell+1} = C_{\ell+1}$.
  By definition, $B_{j} \neq \emptyset$ for all $j \geq \ell+1$.
  Thus, we have proved that the sequence of macrostates $(C_0, B_0) \cdots (C_i, B_i) \cdots$ will not emit infinitely often color $\tcacc{1}{0}$.
\fi

  Assume that there exists a sequence $\hat{\rho} = (C_0, B_0) \cdots (C_i, B_i) \cdots$ of macrostates over $w$ that emits infinitely often color $\tcacc{1}{0}$.
  Our goal is to prove that $w$ is not accepted within $P$.
  Note that $\aut$ is complete, so each run $\rho$ of $\aut$ over $w$ is an infinite run.
    Since all SCCs are accepting and inherently weak in $P$, we only need to prove that every run $\rho$ entering $P$ will eventually exit $P$.
  First, we let $\rho$ enter $P$ at some point, say $k > 0$.
  That is, we have $\rho_k \in C_k$.
  Since $\hat{\rho}$ emits infinitely often the color $\tcacc{1}{0}$, there must be an integer $\ell \geq k$ such that $\trans(B_{\ell}, w_{\ell}) \cap C'_{\ell+1} = \emptyset$.
  It follows that $\rho'_{\ell + 1} \in B_{\ell+1} = C'_{\ell+1}$ for all runs branching from $\rho$ with the $(\ell+1)$-th state being $\rho'_{\ell + 1}$.
  So all runs branching from $\rho$ will be present in the $B_{\ell + 1}$-set.
  Again, by assumption, there must be an integer $\ell' \geq \ell$ such that $\trans(B_{\ell'}, w_{\ell'}) \cap C'_{\ell'+1} = \emptyset$.
  It follows that all runs branching from $\rho$ must have left the $B_{j}$-set for all $j > \ell'$.
  Since $\tcacc{1}{0}$ occurs infinitely often, i.e., there are infinitely many empty $B$-sets along $\hat{\rho}$, all runs entering $P$ must eventually exit $P$.
  Thus, $w$ is not accepted within the partition block $P$.
  
  Now we assume that $w$ is not accepted within $P$ and show that $\tcacc{1}{0}$ occurs infinitely often.
  We prove it by contradiction.
  Suppose that $\tcacc{1}{0}$ occurs only for a finite number of times along the sequence $\hat{\rho} = (C_0, B_0) \cdots (C_i, B_i) \cdots$ of macrostates over $w$.
  Then there exists an integer $k>0$ such that $B_{j} \neq \emptyset$ for all $j \geq k$.
  It follows that $B_{j+1} = \trans(B_{j}, w_{j}) \cap C'_{j+1} = \trans(B_{j}, w_{j}) \cap \trans(H_{j}, w_{j}) \cap P $ for all $j \geq k$, i.e., $B_{j+1} \subseteq \trans(B_{j}, w_{j}) \subseteq P$ for all $j \geq k$.
  By K\"onig's lemma, there must be an infinite run $\rho$ within $B \subseteq P$.
  Since all SCCs in $P$ are accepting and inherently weak, we know that $\rho$ must be accepting, which contradicts the assumption that $w$ is not accepted within $P$.
  Thus, we have proved that $\tcacc{1}{0}$ must occur infinitely often.
  
  Therefore, we have proved that $w$ is not accepted within $P$ if and only if $\tcacc{1}{0}$ occurs infinitely often.
  \qed
\end{proof}

% \bigskip

\lemCorrDAC*

\begin{proof}
  Let $w$ be an $\omega$-word.
  Our goal is to prove that $w$ is not accepted within $P$ if and only if there exists an infinite sequence of macrostates $\hat{\rho} = (C_0, S_0, B_)) \cdots$ that emits color $\tcacc{0}{0}$ infinitely often.

First, assume that there exists an infinite sequence of macrostates $\hat{\rho} = (C_0, S_0, B_0) \cdots $ over $w$ that emits infinitely often the color $\tcacc{0}{0}$. 
We then need to prove that a run $\rho$ of $\aut$ over $w$ that enters $P$ will either leave at some time or not be accepting. 
For simplicity, we let $h > 0$ be an integer such that $\rho_{j}$ belongs to the same SCC in $P$ for all $j \geq h$. 
When a state $s$ on a run transitions from an SCC to a state $t$ in another SCC, we say this run dies out and there is a new run entering $P$ from state $t$.
Assume that the run $\rho$ is present in $P$ at time $k \geq h$.
Once $\rho$ is in $P$, we know that $\rho$ is deterministic, i.e., no branching runs will be derived from $\rho$.
Therefore, we only need to focus on the deterministic run.
Let $\rho_k \in C_k$ (the cases when $\rho_k \in B_{k}$ and $\rho_k \in S_{k}$ are easier and will be discussed later).
There must be an integer $\ell \geq k$ such that $B^{\star}_{\ell+1} = \transscc(B_{\ell}, w_{\ell}) = \emptyset$ since $\tcacc{0}{0}$ occurs infinitely often.
Thus $(C_{\ell+1}, S_{\ell+1}, B_{\ell+1})$ has two possibilities:
(1) $(C'_{\ell+1}, S'_{\ell+1}, B'_{\ell+1} = C'_{\ell+1})$, i.e., $\rho$ is moved to to $B$-set and (2) $(C''_{\ell+1}, S''_{\ell+1}, B''_{\ell+1} = C''_{\ell+1})$, i.e., $\rho$ is moved to the $S$-set.
Since $\rho_{\ell+1} \in C'_{\ell+1}$, we either have $\rho_{\ell+1} \in B_{\ell+1}$ or $\rho_{\ell+1} \in S_{\ell+1}$.
Assume that $\rho_{\ell+1} \in S_{\ell+1}$.
Since $\hat{\rho}$ emits infinitely often $\tcacc{0}{0}$, $\hat{\rho}$ must be infinite.
That is, $\transacc(B_{j}, w_{j}) \cap P = \emptyset$ for all $j \geq \ell +1$.
Since $\rho_{j} \in  \transscc(B_{j}, w_{j}) \cap P$ for all $j \geq \ell + 1$, by definition, we have $\transacc(B_{j}, w_{j}) = \emptyset$ for all $j \geq \ell + 1$. 
That is, $\rho$ must not visit accepting transitions any more after $j \geq \ell+1$;
otherwise $\hat{\rho}$ will be finite if there is some $j$ such that $\transacc(B_{j}, w_{j}) \neq \emptyset$.
If $\rho_{\ell+1} \in B_{\ell+1}$, we know that there must exist an integer $\ell' \geq \ell$ such that $\transscc(B_{\ell'}, w_{\ell'}) = \emptyset$.
That is, we have either $\rho_{\ell' + 1} \notin P$ or $\rho_{\ell' + 1} \in P$ but $\rho_{\ell' + 1}$ and $\rho_{\ell' }$ are not in the same SCC.
In the latter case, we treat $\rho$ as died out and there will be a new run in $C$.
Since $\rho$ will eventually stay in an SCC forever, it is easy to see that $\rho$ will either be in $B$ or leave $P$.
%It follows that the states $\rho_{j}$ is not present in $P$ any more for all $j > \ell'$.
%This is because that $\rho_{\ell' + 1} \in \trans(B_{\ell'}, w_{\ell'}) \subseteq \trans(C_{\ell'}, w_{\ell'})$ since $\rho_{\ell'} \in B_{\ell'}$.
%Then it must be the case that $\rho_{\ell' + 1} \notin C'_{\ell'+1}$ since $\trans(B_{\ell'}, w_{\ell'}) \cap C'_{\ell'+1} = \emptyset$.
%If $\rho_{\ell' + 1} \in S'_{\ell'+1} = S_{\ell'+1}$, we can show that $\rho$ is not accepting;
%otherwise $\rho_{\ell' + 1} \notin S'_{\ell'+1} = S_{\ell'+1}$, we have that $\rho_{\ell' + 1} \notin \trans(H_{\ell'},w_{\ell'}) \cap P$.
%Clearly, $\rho_{\ell'+1} \in  \trans(H_{\ell'},w_{\ell'})$.
%So, $\rho_{\ell'+1} \notin P$.
Therefore we have proved that all runs of $\aut$ over $w$ that enter $P$ will either leave at some time or not be accepting.

Now, assume that $w$ is not accepted within $P$.
Our goal is to prove that there exists an infinite sequence of macrostates $\hat{\rho} = (C_0, S_0, B_0) \cdots $ over $w$ that emits infinitely often the color $\tcacc{0}{0}$.
Since every run $\rho$ of $\aut$ over $w$ will either leave $P$ or become safe, we can construct such an infinite sequence $\hat{\rho}$.
First, we need $\hat{\rho}$ to be infinite and we only need to be careful about the condition $\transacc(S, a) \neq \emptyset$.
All runs in $S$ can be seen as coming from $B$ (including $(C'', S'', C'')$ as it still needs to first compute $B'$).
We only need to resolve the nondeterministic choices when constructing $\hat{\rho}$.
If $B$ is empty all the time, we are done.
Otherwise let $k$ be the smallest integer when $B \neq \emptyset$.
That is, the current macrostate is $(C_k, S_k, B_k)$.
Since all runs in $B$ are deterministic, we can do standard construction (by following the successor $(C', S', B')$) until either we reach a point where all runs in the $B$-set die out or become safe.
If $B$ becomes empty, we still follow the successor $(C' ,S', B')$ and the construction will emit $\tcacc{0}{0}$.
It can happen that all runs in $B$ become safe since the number of runs in $B$ is finite and they will be safe eventually by assumption.
In such a case, it is easy to see that $\transacc(B, a) \cap \transscc(B, a) = \emptyset$.
Then we follow the successor $(C'', S'', C'')$ this time and emit $\tcacc{0}{0}$.
Since all the runs we move to $S$ are safe, so $\transacc(S, a) \neq \emptyset$ will not be satisfied in future.
In this way, we obtain a macrostate $(C_{k+\ell}, S_{k+\ell}, B_{k+\ell})$ for some $\ell \geq 1$.
We can repeat the above procedure and construct an infinite sequence of macrostates $\hat{\rho}$ over $w$ that emits infinitely often the color $\tcacc{0}{0}$.
% Thus we have the proof.
\qed
\end{proof}

% \bigskip 

\lemCorrIDAC*

\begin{proof}
  Let $w$ be an $\omega$-word.
  We need to prove that $w$ is not accepted in $P$ if and only if we receive only finitely many times the color $\tcacc{4}{0}$.
  
  First, we prove that direction from right to left by contraposition.
  By assumption, we have finitely many occurrences of the color $\tcacc{4}{0}$ along the word $w$.
  Suppose that $w$ is accepted in $P$.
  There must exist an accepting run $\rho$ that eventually stays in $P$.
  It is easy to see that $\rho$ is accepted by the reduced deterministic BA $\aut_P$.
  Let $k \geq 0$ be the smallest integer such that $\rho_k \in P$.
  Therefore, we have $\GetSuccOf{\idac_P}(H_j, \rho_j, w_j) = \{(\rho_{j+1}, \alpha_{j+1})\}$ for all $j \geq k$.
  Since $\rho$ will visit infinitely many accepting transitions, we will also see infinitely often the color $\tcacc{4}{0}$.
  This leads to a contradiction to our assumption.
  Thus, $w$ cannot be accepted in $P$.
  
  Second, we prove the other direction also by contraposition.
  By assumption, $w$ is not accepted in $P$.
  Assume that we see infinitely many $\tcacc{4}{0}$ and the sequence of macrostates over $w$ is $\hat{\rho}$.
  Then there must be infinitely many integers $k > 0$ such that $\hat{\rho}_{k}, \hat{\rho}_{k+1} \in P$ and $\hat{\rho}_{k}  \ltr {w_k} \hat{\rho}_{k+1} \in \acc$.
  If $\hat{\rho}_{j} \in P$ for all $j \geq k$, we must have an accepting run in $P$, which contradicts the assumption that $w$ is not accepted in $P$.
  So there must be some integer $\ell > k$ such that $\trans(H_{\ell}, w_{\ell}) \cap P = \emptyset$.
  This indicates that every run starting from $\hat{\rho}_{k}$ is finite.
  Since $P$ is deterministic, it follows that every run over $w$ that enters $P$ is finite, therefore $w$ is not accepted in $P$. Contradiction.
  Thus, we have proved that if $w$ is not accepted in $P$, we only can see finitely many times the color $\tcacc{4}{0}$.
  \qed
  \end{proof}

  % \bigskip

\thmElevatorBound*
\begin{proof}
  Assume that $Q_D$ is the union of all SCCs of $\aut$ satisfying $\phi_\dac$, $Q_W$ is the union of all SCCs satisfying $\phi_\iwc$ and $Q_N$ is the union of all nonaccepting SCCs;
  moreover 
  $Q_D \cap Q_W = \emptyset, Q_N \cap Q_W = \emptyset$ and $Q_N \cap Q_D = \emptyset$. Since $\aut$ is elevator, $Q_D \cup Q_W \cup Q_N$ is the set of all states in $\aut$ and $Q_D \cup Q_W$ is the union of all partition blocks of $\aut$.
  From \cref{thm:correctness}, \cref{lem:correctness-iwcs}, and \cref{lem:dac-correctness}
  we have that $\langof{\modcompl(\dac_{Q_D}, \iwc_{Q_W},\aut)} = \Sigma^\omega\setminus\langof{\aut}$.
  We now compute the number of states of $\modcompl(\dac_{Q_D}, \iwc_{Q_W},\aut)$.
  For a state $q \in Q_D$ there are 4 possibilities of distributing $q$ 
  within $\onetypeof{\dac_{Q_D}}$: 
  \begin{inparaenum}[(i)]
    \item $q\notin C \cup S$,
    \item $q \in C$,
    \item $q\in C \cap B$,
    \item $q\in S$.
  \end{inparaenum}
  For a state $q\in Q_W$ there are 3 possibilities of distributing $q$
  within $\onetypeof{\iwc_{Q_W}}$:
  \begin{inparaenum}[(i)]
    \item $q\notin C$,
    \item $q\in C$,
    \item $q \in C \cap B$.
  \end{inparaenum} 
  Lastly, for a state $q\in Q_N$ there are 2 possibilities of distributing $q$ within the reachable states $H$: $q \in H$ or $q\notin H$. 
  Therefore, the number of macrostates is given as $4^{|Q_D|}\cdot 3^{|Q_W|} \cdot 2^{|Q_N|}\in \bigO(4^n)$.
\qed
\end{proof}

%%%%%%%%%%%%%%%%%%%%%%%%%%%%%%%%%%%%%%%%%%%%%%%%%%%%%%%%%%%%%%%%%%%
\subsection{Proofs of \cref{sec:optimizations}}
%%%%%%%%%%%%%%%%%%%%%%%%%%%%%%%%%%%%%%%%%%%%%%%%%%%%%%%%%%%%%%%%%%%

\thmPostp*

\begin{proof}
  From \cref{claim-intersection} and \cref{thm:correctness} we have that 
  $$
    \bigcap_{i=1}^n \langof{\modcompl(\alg^i_{P_i}, \aut_{P_i})} = \Sigma^\omega\setminus\langof{\aut}.
  $$
  Since reduction $\reduce$ preserves the language, we have $\langof{\modcompl(\alg^i_{P_i}, \aut_{P_i})} = 
  \lang\left(\reduce\left(\modcompl(\alg^i_{P_i}, \aut_{P_i})\right)\right)$ for each $i$, which concludes the proof.
  \qed
\end{proof}

\section{Additional Plots from the Experiments}
\label{app:experimentsAdditionalPlots}

In this section we present more plots about the outcomes of the experiments.

\begin{figure}[t]
  \resizebox{\linewidth}{!}{
    \includegraphics{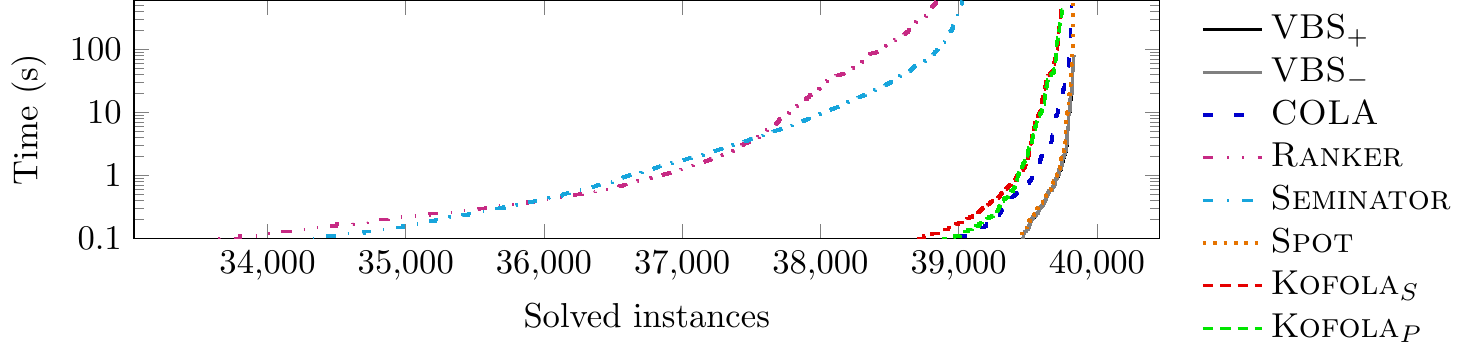}
  }
  \caption{Cactus plot showing the number of instances solved by each tool within the time on the y axis.}
  \label{fig:experimentsCactusPlottime}
\end{figure}
\begin{figure}[t]
  \resizebox{\linewidth}{!}{
    \includegraphics{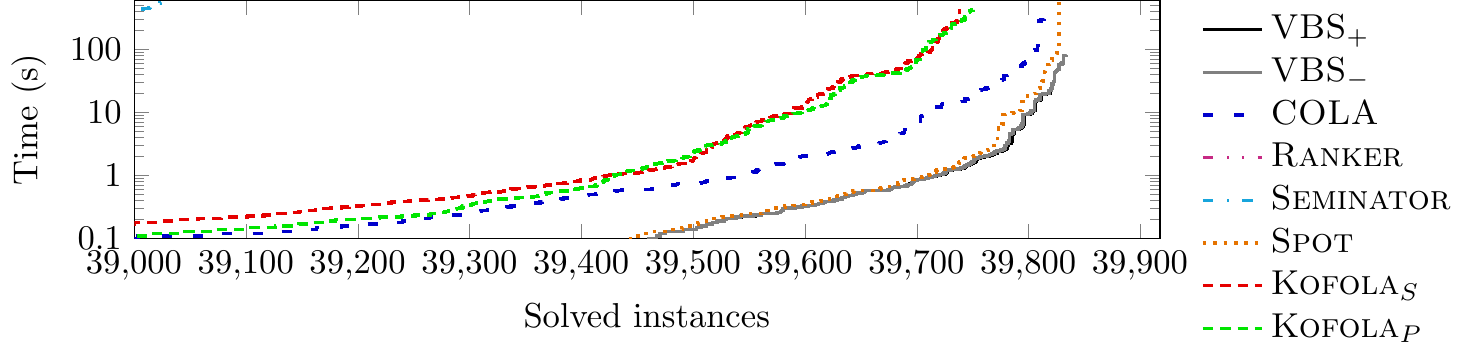}
  }
  \caption{Particular of the cactus plot in \cref{fig:experimentsCactusPlottime}.}
  \label{fig:experimentsCactusPlottimeZoom}
\end{figure}
In \cref{fig:experimentsCactusPlottime} we provide a cactus plot presenting for each tool, including the virtual best solvers, the number of benchmarks (on the x axis) successfully complemented within the time given on the y axis; 
the more the plot is near the right border, the better the tool behaves.
\cref{fig:experimentsCactusPlottimeZoom} provides a clearer view of the part of the plot in \cref{fig:experimentsCactusPlottime} above 39,000 states.
As we can see from the plots, \spot is the clear winner when considering the time needed to complement the input TBA, since its plot is almost superimposed to the one of both \vbs; 
this confirms the high quality and maturity of \spot and the several techniques it implements to manage at the best \buchi automata operations.
\kofolaidsc is slightly better than \kofolaids and very close to \cola on the automata requiring short time to be complemented;
then both versions of \kofola behave similarly with \cola being a bit faster in producing larger automata.

\end{document}